\pdfoutput=1
\documentclass[11pt]{article}

\makeatletter
\newcommand\email[2][]%
   {\newaffiltrue\let\AB@blk@and\AB@pand
      \if\relax#1\relax\def\AB@note{\AB@thenote}\else\def\AB@note{\relax}%
        \setcounter{Maxaffil}{0}\fi
      \begingroup
        \let\protect\@unexpandable@protect
        \def\thanks{\protect\thanks}\def\footnote{\protect\footnote}%
        \@temptokena=\expandafter{\AB@authors}%
        {\def\\{\protect\\\protect\Affilfont}\xdef\AB@temp{#2}}%
         \xdef\AB@authors{\the\@temptokena\AB@las\AB@au@str
         \protect\\[\affilsep]\protect\Affilfont\AB@temp}%
         \gdef\AB@las{}\gdef\AB@au@str{}%
        {\def\\{, \ignorespaces}\xdef\AB@temp{#2}}%
        \@temptokena=\expandafter{\AB@affillist}%
        \xdef\AB@affillist{\the\@temptokena \AB@affilsep
          \AB@affilnote{}\protect\Affilfont\AB@temp}%
      \endgroup
       \let\AB@affilsep\AB@affilsepx
}
\makeatother

\let\OLDthebibliography\thebibliography
\renewcommand\thebibliography[1]{
  \OLDthebibliography{#1}
  \setlength{\parskip}{2pt}
  \setlength{\itemsep}{4pt plus 1pt}
}

\usepackage{geometry}
\usepackage{pifont}
\usepackage{fourier}
\DeclareMathAlphabet{\mathcal}{OMS}{zplm}{m}{n}

\usepackage{booktabs} % For formal tables
\usepackage{tabularx}
\usepackage[ruled, linesnumbered]{algorithm2e} % For algorithms
\usepackage[numbers]{natbib}
\usepackage{amsmath}
\usepackage{amssymb}
\usepackage{amsthm}
\usepackage{thm-restate}
\usepackage{array}
\usepackage{xcolor}
\usepackage{multicol}
\usepackage[colorlinks]{hyperref}
\def\tmp#1#2#3{%
  \definecolor{Hy#1color}{#2}{#3}%
  \hypersetup{#1color=Hy#1color}}
\tmp{link}{HTML}{800006}
\tmp{cite}{HTML}{2E7E2A}
\tmp{file}{HTML}{131877}
\tmp{url} {HTML}{8A0087}
\tmp{menu}{HTML}{727500}
\tmp{run} {HTML}{137776}
\def\tmp#1#2{%
  \colorlet{Hy#1bordercolor}{Hy#1color#2}%
  \hypersetup{#1bordercolor=Hy#1bordercolor}}
\tmp{link}{!60!white}
\tmp{cite}{!60!white}
\tmp{file}{!60!white}
\tmp{url} {!60!white}
\tmp{menu}{!60!white}
\tmp{run} {!60!white}

\usepackage{comment}
\usepackage{url}

\SetAlFnt{\small}
\SetAlCapFnt{\small}
\SetAlCapNameFnt{\small\scshape}
\SetAlgoSkip{smallskip}

\makeatletter
\newenvironment{mechanism}[1][htb]{%
   \begin{algorithm}[#1]%
  }{\end{algorithm}}
\makeatother

\SetCommentSty{mycommfont}
\newcommand{\algcom}[1]{\tcp*[r]{#1}}

\usepackage{mathtools}

\usepackage{xcolor}
\colorlet{agentColor}{gray!20}
\colorlet{resourceColor}{white}

\usepackage{tikz}
\usetikzlibrary{
    graphs,
    graphs.standard,
    shapes,
    arrows.meta,
    backgrounds,
    positioning
}
\tikzset{
	agent/.style={draw = black, fill = agentColor, circle, text = black, inner sep = 0pt, minimum size = 0.55cm},
	resource/.style={draw = black, fill = resourceColor, rectangle, text = black, inner sep = 0pt, minimum size = 0.5cm},
}

\usepackage{wrapfig}
\usepackage{enumitem}
\usepackage{graphicx}
\usepackage{subcaption}
\usepackage{dsfont}
\usepackage{color}
\usepackage[noblocks]{authblk}
\renewcommand\Affilfont{\normalsize}
\usepackage[normalem]{ulem}
\usepackage{nicefrac}
\usepackage{pgfplots}
\usepgfplotslibrary{fillbetween}
\pgfplotsset{compat=1.18} 
\usepackage{bm}
\renewcommand{\vec}[1]{\boldsymbol{#1}}

\newtheorem{theorem}{Theorem}[section]
\newtheorem{lemma}{Lemma}[section]

\newtheorem{corollary}{Corollary}[section]

\theoremstyle{definition}
\newtheorem{definition}{Definition}[section]
\newtheorem{claim}{Claim}[section]
\newtheorem{remark}{Remark}[section]
\newtheorem{observation}{Observation}[section]

\newcommand{\xmark}{\ding{55}}%
\newcommand{\cmark}{\ding{51}}%

\newcommand{\declaredEdges}{D}
\newcommand{\agentSet}{L}
\newcommand{\taskSet}{R}
\newcommand{\privateEdges}{E}
\newcommand{\values}{\bm{v}}
\newcommand{\sizes}{\bm{s}}
\newcommand{\caps}{\bm{C}}
\newcommand{\mstar}[1]{M^{*}_{#1}}
\newcommand{\mech}{\mathcal{M}}
\newcommand{\isize}{s}
\newcommand{\jcap}{C}

\newcommand{\minisec}[1]{\medskip\noindent\textbf{#1.~~}}

\DeclareMathOperator*{\argmax}{arg\,max}

\newcommand{\mrank}{\textsc{order}}
\newcommand{\msort}{\textsc{sort}}
\newcommand{\mlist}[1]{\ensuremath{\langle #1 \rangle}}
\newcommand{\lexext}{\ensuremath{\succeq'}}
\newcommand{\lexexts}{\ensuremath{\succ'}}

\newcommand{\lex}{\ensuremath{\succeq^{\text{lex}}}}

\newcommand{\set}[1]{\ensuremath{\{#1\}}}
\newcommand{\sset}[2]{\ensuremath{\{#1 \; | \; #2\}}}

\newcommand{\bmp}{\ensuremath{\text{BMP}}\xspace}
\newcommand{\mkar}{\ensuremath{\text{RMK}}\xspace}
\newcommand{\emkar}{\ensuremath{\text{ERMK}}\xspace}
\newcommand{\gap}{\ensuremath{\text{GAP}}\xspace}
\newcommand{\vcgap}{\ensuremath{\text{VCGAP}}\xspace}
\newcommand{\vigap}{\ensuremath{\text{AVGAP}}\xspace}
\newcommand{\sigap}{\ensuremath{\text{ASGAP}}\xspace}

\newcommand{\boost}{\textsc{Boost}\xspace}
\newcommand{\trust}{\textsc{Trust}\xspace}
\newcommand{\randomizedBoost}{{\textsc{Boost-or-Trust}\xspace}}
\newcommand{\mkarGreedy}{\textsc{Greedy-by-Theta}\xspace}
\newcommand{\randomizedMKAR}{{\textsc{Greedy-or-Trust}\xspace}}
\newcommand{\greedy}{{\textsc{Greedy}\xspace}}
\newcommand{\randomizedGAP}{\textsc{\boost-or-Greedy-or-Trust}\xspace}

\newcommand{\defacc}{\text{deferred acceptance algorithm}\xspace}

\begin{document}

\title{\bfseries
    To Trust or Not to Trust: Assignment Mechanisms with Predictions in the Private Graph Model
}

\author[1]{Riccardo Colini-Baldeschi}
\author[2,3]{Sophie Klumper}
\author[2,4]{\\Guido Sch\"afer}
\author[2]{Artem Tsikiridis}

\affil[1]{Meta, UK}
\affil[2]{Centrum Wiskunde \& Informatica (CWI), The Netherlands}
%\email{\url{{s.j.klumper,g.schaefer, a.tsikiridis}@cwi.nl}}
\affil[3]{Vrije Universiteit Amsterdam, The Netherlands}
\affil[4]{University of Amsterdam, The Netherlands}

\date{}

\maketitle
\begin{abstract}
\noindent The realm of algorithms with predictions has led to the development of several new algorithms that leverage (potentially erroneous) predictions to enhance their performance guarantees. 
The challenge here is to devise algorithms that achieve optimal approximation guarantees as the prediction quality varies from perfect (consistency) to imperfect (robustness).
This framework is particularly appealing in mechanism design contexts, where predictions might convey private information about the agents. 
This aspect serves as the driving force behind our research:
Our goal is to design strategyproof mechanisms that leverage predictions to achieve improved approximation guarantees for several variants of the Generalized Assignment Problem (GAP) in the private graph model. 
In this model, first introduced by \citet{dughmi10}, the set of resources that an agent is compatible with is private information, and assigning an agent to an incompatible resource generates zero value. 
For the Bipartite Matching Problem (BMP), we give a deterministic group-strategyproof (GSP) mechanism that is $(1 + \nicefrac{1}{\gamma})$-consistent and $(1 + \gamma)$-robust, where $\gamma \ge 1$ is some confidence parameter.
We also prove that this is best possible. 
Remarkably, our mechanism draws inspiration from the renowned Gale-Shapley algorithm, incorporating predictions as a crucial element. Additionally, we give a randomized mechanism that is universally GSP and improves on the guarantees in expectation. 
The other GAP variants that we consider all make use of a unified greedy mechanism that adds edges to the assignment according to a specific order. 
For a special case of Restricted Multiple Knapsack (each agent's value is equal to their size), this results in a deterministic GSP mechanism that is $(1+\nicefrac{1}{\gamma})$-consistent and $(2+\gamma)$-robust. 
We then focus on two variants: the Agent Size GAP (each agent has one size) and the Value Consensus GAP (all agents have the same preference order over the resources). 
Both variants use the same template that leads to a universally GSP mechanism that is $(1+\nicefrac{3}{\gamma})$-consistent and $(3+\gamma)$-robust in expectation. Our mechanism randomizes over the greedy mechanism, our mechanism for BMP and the predicted assignment. 
All our mechanisms also provide more fine-grained approximation guarantees that smoothly interpolate between the consistency and robustness, depending on some natural error measure of the prediction.   
\end{abstract}

\section{Introduction}

Mechanism design is centered around the study of situations where multiple self-interested agents interact within a system. Each agent holds some private information about their preferences (also called \emph{type}), based on which they make decisions. The primary goal is to create systems such that, despite the agents acting in their own self-interest, the outcome is socially desirable or optimal from the designer's perspective. 
One of the key challenges is to design mechanisms that incentivize the agents to reveal their preferences truthfully. 
A prominent notion in this context is \emph{strategyproofness}, which ensures that it is in the best interest of each agent to reveal their preferences truthfully, independently of the other agents.
Unfortunately, strategyproofness often imposes strong impossibility results on achieving the socially desirable objective optimally or even approximately. As a consequence, the wort-case approximation guarantees derived in the literature can be rather disappointing from a practical perspective (see, e.g., \citet{roughgarden19}).

\begin{table}[t]
\centering
{\small
\begin{tabularx}{\textwidth}{>{\raggedright\arraybackslash}Xl}
\toprule
\textbf{GAP Variant} & \textbf{Restrictions} ($\forall i \in L$, $\forall j \in R$) \\
\midrule
\emph{Unweighted Bipartite Matching (U\bmp)} & $v_{ij} = 1$, $\isize_{ij} = 1$, $\jcap_j = 1$ \\
\emph{Bipartite Matching Problem (\bmp)} & $\isize_{ij} = 1$, $\jcap_j = 1$ \\
\textit{Restricted Multiple Knapsack (\mkar)} & $v_{ij} = v_i$, $\isize_{ij} = \isize_i$ \\
\textit{Equal RMK (\emkar)} & $v_{ij} = \isize_{ij} = v_i$ \\
\emph{Value Consensus GAP (\vcgap)} & $\exists \sigma: \ v_{i\sigma(1)} \ge \dots, \ge v_{i\sigma(m)}$ \\
\emph{Agent Value GAP (\vigap)} & $v_{ij} = v_i$ \\
\emph{Resource Value GAP (RVGAP)} & $v_{ij} = v_{j}$ \\
\emph{Agent Size GAP (\sigap)} & $\isize_{ij} = s_i$ \\
\emph{Resource Size GAP (RSGAP)} & $s_{ij}= s_j$ \\
\bottomrule
\end{tabularx}
}
\caption{Overview of GAP variants considered in this paper.\label{tab:GAP-var}} 
\end{table}

\minisec{Mechanism Design with Predictions}
To overcome these limitations, a new line of research, called \emph{mechanism design with predictions}, is exploring how to leverage learning-augmented inputs, such as information about the private types of the agents or the structure of the optimal solution, in the design of mechanisms. 
While this line of research first emerged in the area of online algorithms (see, e.g., \citet{lykouris21}), it is particularly appealing in the context of mechanism design. Namely, in economic environments this information can oftentimes be extracted from data through machine-learning techniques.
In the context of mechanism design, \citet{agrawal22} and \citet{xu22} are among the first works along this line.

In the mechanism design with predictions framework, the designer can exploit the predicted information to improve the worst-case efficiency of their mechanism. However, the predictions might be inaccurate or even entirely erroneous. As a result, the goal is to design mechanisms that guarantee attractive approximation guarantees if the prediction is perfect (referred to as \emph{consistency}), while still maintaining a reasonable worst-case guarantee when the prediction is imperfect (referred to as \emph{robustness}). 
Ideally, the mechanism provides a fine-grained approximation guarantee depending on some measure of the prediction error, which smoothly interpolates between these two extreme cases (referred to as \emph{approximation}).

In this paper, we study how to leverage learning-augmented predictions in the domain of mechanism design \emph{without money}. 
How to design strategyproof mechanisms without leveraging monetary transfers is a much more complicated problem (see, e.g., \citet{schummer2007mechanism} and \citet{procaccia13}).
Indeed, in the standard mechanism design with money literature, monetary transfers can be employed to effectively eliminate the incentives for agents to misreport their types. On the other hand, in some practical settings the designer might not be allowed to leverage monetary transfers for ethical and legal issues (see \citet{roughgarden2010algorithmic}), or due to practical constraints (see \citet{procaccia13}).

\minisec{Generalized Assignment Problem with Predictions}
We focus on the \emph{Generalized Assignment Problem (\gap)}, which is one of the most prominent problems that have been studied in the context of mechanism design without money.
In this problem, we are given a bipartite graph $G=(\agentSet \cup \taskSet,\declaredEdges)$ with a set $\agentSet$ of strategic agents (or jobs) that can be assigned to a set $\taskSet$ of resources (or machines).
Each agent $i \in \agentSet$ has a value $v_{ij}$ and a size $\isize_{ij}$ for being assigned to resource $j \in \taskSet$. The values of the agents are assumed to be private information. 
Further, each resource $j \in \taskSet$ has a capacity $\jcap_j$ (in terms of total size) that must not be exceeded. 
The goal of the designer is to compute a feasible assignment of agents to resources such that the overall value is maximized. 
This problem models several important use cases that naturally arise in applications such as online advertising, crew planning, machine scheduling, etc. 
Unfortunately, it is known that deterministic strategyproof mechanisms are unable to provide bounded approximation guarantees for GAP (see \citet{dughmi10}). 

\begin{wrapfigure}{r}{0.25\linewidth}
\centering
   \begin{subfigure}[b]{0.25\linewidth}
    \hspace*{-1cm}\small
    \begin{tikzpicture}[auto, node_style/.style={circle,draw=black}]
        \node[agent] (v1) at (-1,1) {1}; \node[agent] (v2) at (-1,0) {2}; \node[resource] (v4) at (1,0.5) {$h$}; 
        \draw [line width = 3pt] (v1) edge node{1} (v4); \draw (v4) edge node[xshift=-0.4cm]{$\frac{1}{\alpha} - \epsilon$} (v2);
    \end{tikzpicture}
    \caption{ }
    \label{fig:privateModelA}
  \end{subfigure}
  
  \vspace*{.2cm}

\begin{subfigure}[b]{0.25\linewidth}
    \hspace*{-1cm}\small
    \begin{tikzpicture}[auto, node_style/.style={circle,draw=black}]
        \node[agent] (v1) at (-1,1) {1}; \node[agent] (v2) at (-1,0) {2}; \node[resource] (v4) at (1,0.5) {$h$}; 
        \draw [line width = 3pt] (v1) edge node{1} (v4); \draw (v4) edge node[xshift=-0.4cm]{$\beta + \epsilon$} (v2);
    \end{tikzpicture}
    \caption{}
    \label{fig:privateModelB}
  \end{subfigure}

  \vspace*{-.1cm}

  \caption{Limitations for $\gap$ with predictions.}
  \label{fig:privateModel}
\end{wrapfigure}
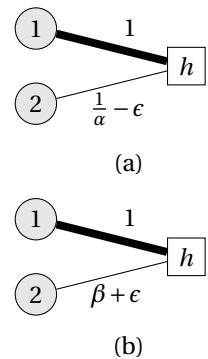
Even if the \gap is enhanced with a prediction suggesting the optimal assignment, this limitation cannot be overcome. This is illustrated by the following example. Consider two instances of the matching problem depicted in Figure \ref{fig:privateModel}.
Let $\mathcal{M}$ be a deterministic strategyproof mechanism with bounded consistency and robustness, i.e., it is $\alpha$-consistent and $\beta$-robust with $1 \le \alpha < \infty$ and $1 \le \beta < \infty$. Consider as input the truthful declarations as in Figure \ref{fig:privateModelA} with $0< \epsilon < \frac{1}{\alpha}$ and a perfect prediction $\set{(1,h)}$ (indicated in bold). To achieve $\alpha$-consistency, $\mathcal{M}$ must return $\{(1,h)\}$. Now consider as input the truthful declarations as in Figure \ref{fig:privateModelB} and the same prediction $\set{(1,h)}$, i.e., an imperfect prediction. To achieve $\beta$-robustness, $\mathcal{M}$ must return $\{(2,h)\}$. However, this contradicts strategyproofness, as agent 2 will unilaterally deviate and declare a value of $\beta + \epsilon$ if the true instance is as in Figure \ref{fig:privateModelA} with a perfect prediction. So the best a deterministic strategyproof mechanism can do is always return the prediction, leading to 1-consistency and unbounded robustness.

\minisec{Private Graph Model}
In light of this, we turn towards a slightly more restrictive (but natural) model for GAP that was introduced by \citet{dughmi10}, called the \emph{private graph model}. Here, the agents' values are assumed to be public information, but whether or not the value $v_{ij}$ can be generated by assigning agent $i$ to resource $j$ is private information. The latter can naturally be interpreted as \emph{compatibility restrictions} that agents have with respect to the available resources. Note that this variation restricts the strategy space of the agents from having the ability to misreports their entire valuation vector to being able to misreport only their compatibility vector.
Despite this restriction, GAP in the private graph model still has several natural applications (see also \cite{dughmi10}).

We study GAP in the private graph model considered in a learning-augmented setting, and assume that the optimal assignment with respect to the actual compatibilities is given as a prediction. Note that this is weaker than assuming that the actual compatibilities are available as a prediction. To see this, note that we can always compute an optimal predicted assignment with respect to some predicted compatibilities (notwithstanding computational constraints). Depending on the underlying application, it seems reasonable to assume that such assignment are learnable (e.g., through deep reinforcement learning, graph convolution neural networks, etc.).

\subsection{Our Contributions} 

We study GAP in the private graph model with predictions. We assume that a (potentially erroneous) prediction of the optimal assignment for the true compatibility graph is given as part of the input. We derive both deterministic and randomized mechanisms that are (universally) group-strategyproof for different variants of GAP; see Table~\ref{tab:GAP-var} for an overview. Our mechanisms are parameterized by a \emph{confidence parameter $\gamma \ge 1$}, which determines the trade-off between the respective consistency and robustness guarantees. Choosing a higher confidence value leads to a better consistency but a worse robustness guarantee, and vice versa. 

All our mechanisms provide more fine-grained approximation guarantees that smoothly interpolate between the consistency and robustness guarantees stated below, depending on some natural error parameter $\hat{\eta}$ of the prediction. More specifically, $1-\hat{\eta}$ measures the relative gap between the value of the predicted assignment and an optimal one; in particular, $\hat{\eta} = 0$ if the prediction is perfect, while $\hat{\eta} = 1$ if the prediction is arbitrarily bad. 

We summarize our main results below. 
\begin{itemize}
    \item We prove a lower bound on the best possible trade-off in terms of consistency and robustness guarantees that is achievable by any deterministic strategyproof mechanism for GAP (Section~\ref{sec:lowerBounds}). More precisely, we show that no deterministic strategyproof mechanism can be $(1+\nicefrac{1}{\gamma})$-consistent and $(1+\gamma-\epsilon)$-robust for any $\epsilon > 0$. 
    In fact, our lower bound also holds for the special case of the Bipartite Matching Problem (\bmp). We also extend our insights to derive a lower bound in terms of consistency and approximation guarantees.
    
    \item For \bmp, we derive a deterministic group-strategyproof mechanism that is $(1 + \nicefrac{1}{\gamma})$-consistent and $(1 + \gamma)$-robust (Section~\ref{sec:matching}). In light of the lower bound above, our mechanism thus achieves the best possible consistency and robustness guarantees, albeit satisfying the stronger notion of group-strategyproofness (GSP). 
    Unlike the mechanism known in the literature for the problem without predictions, we crucially do not consider declarations in a fixed order. Instead, our mechanism draws inspiration from the well-known deferred acceptance algorithm by \citet{GS62}. Here, the agent proposal order is crucial for GSP and the resource preference order is crucial to improve upon the known guarantee for the problem without predictions. If an edge is in the predicted optimal matching, it potentially has a better ranking in the resource preference order, depending on the confidence parameter and the instance at hand. 
    Our mechanism for \bmp extends (with the same approximation guarantees) to many-to-one assignments and RS\gap. In particular, this provides the first deterministic GSP mechanism that is $2$-approximate for RS\gap. 
    
    \item For \gap, we give a deterministic greedy mechanism that greedily adds declared edges (while maintaining feasibility) to an initially empty assignment, according to some order of the declarations (Section~\ref{sec:greedy}). The order of the declarations follows from a specific ranking function, that is given as part of the input. We derive a sufficient condition, called \emph{truth-inducing}, of the ranking function which guarantees that the resulting greedy mechanism is GSP. For the special case of \emkar, we combine the greedy mechanism with a truth-inducing ranking function resulting in a deterministic GSP mechanism that is $(1 + \nicefrac{1}{\gamma})$-consistent and $(2 + \gamma)$-robust. The same approach can be used to obtain a 3-approximate GSP mechanism for the setting without predictions, for which no polynomial time deterministic strategyproof mechanism was known prior to this work. 

    \item For \sigap and \vcgap, we derive randomized universally GSP mechanisms that are $(1 + \nicefrac{3}{\gamma})$-consistent and $(3 + \gamma)$-robust (Section~\ref{sec:randomized}). 
    To this aim, we randomize over three deterministic mechanisms, consisting of our mechanism for \bmp, our greedy mechanism and a third mechanism that simply follows the prediction.     
    As the previously mentioned greedy mechanism is one of the three building blocks, it is crucial that for these variants there exist truth-inducing ranking functions. Notably, none of the three mechanisms achieves a bounded robustness guarantee by itself. Finally, for \bmp and \emkar, we derive randomized universally GSP mechanisms that are $(1 + \nicefrac{1}{\gamma})$-consistent and outperform the robustness guarantees of their respective deterministic counterparts in expectation. In particular, for \bmp this provides a separation result showing that randomized mechanisms are more powerful than deterministic ones (at least in expectation).
\end{itemize}

\subsection{Related Work}

Algorithms with predictions represent one perspective within the ``beyond worst-case'' paradigm. The primary goal is to overcome existing worst-case lower bounds by augmenting each input instance with a prediction, possibly a machine-learned one. Hence, this line of work is also sometimes referred to as ``learning-augmented algorithms''. The conceptual framework that describes the trade-off between $\alpha$-consistency and $\beta$-robustness was introduced by \citet{lykouris21} in the context of online algorithms. Since then, online algorithms have remained a major focus (see e.g., \citet{purohit18, azar21, azar22, banerjee22} for some reference works). Thematically relevant to us, are the works on online matching (e.g., \citet{antoniadis20, antoniadis23, lavastida21, lavastida21-2, jin22, dinitz22}) in non-strategic environments.\footnote{An exception to this is the work by \citet{antoniadis20}, who, even though their main focus is on designing algorithms for online Bipartite Matching, observe that their algorithm implies a strategyproof mechanism if monetary transfers are allowed.} Other domains that have been studied under the lens of predictions include the reevaluation of runtime guarantees of algorithms (see e.g., \citet{dinitz21, chen22, sakaue22} for bipartite matching algorithms), streaming algorithms, data structures, and more. We refer the reader to \citet{mitzenmacher20} for a survey of some of the earlier works.\footnote{An overview of research articles that appeared on these topics is available at \url{https://algorithms-with-predictions.github.io}.}

Recently, \citet{xu22} and \citet{agrawal22} introduced predictions for settings involving strategic agents. In their work, \citet{xu22} showcased four different mechanism design settings with predictions, both with and without monetary transfers. On the other hand, \citet{agrawal22} focused solely on strategic facility location. Most subsequent works with strategic considerations have also been in algorithmic mechanism design (see e.g., \citet{balkanski23, istrate22, balcan23, balkanski2023online}). However, other classic domains of economics and computation literature continue to be revisited in the presence of predictions; see, e.g., the works by \citet{gkatzelis22} on the price of anarchy, \citet{berger23} on voting, \citet{lu23} and \citet{caragiannis24} on auction revenue maximization.

We briefly elaborate on the relation between learning-augmented mechanism design and Bayesian mechanism design. As pointed out by \citet{agrawal22}, the main difference is the absence of worst-case guarantees in the Bayesian setting. Indeed, in the standard Bayesian setting, it is implicitly assumed that one has perfect knowledge of the distribution when analyzing the expected performance of mechanisms. While this is a reasonable assumption in some settings, Bayesian mechanisms do not offer any guarantees if this assumption fails. 

Mechanism design without money has a rich history spanning over fifty years, being deeply rooted in economics and social choice theory. As will be evident in Section \ref{sec:matching}, the seminal works of \citet{GS62, Roth82} and \citet{HM05} on stable matching are particularly relevant to our study. However, our work aligns more closely with the agenda of \emph{approximate} mechanism design without money set forth by \citet{procaccia13} and, in particular, the subsequent work by \citet{dughmi10}. In their work, \citet{dughmi10} introduced the private graph model, that we use in our environment with predictions, and initiated the study of variants of $\gap$ when the agents are strategic (a variant of this model where the resources are strategic instead, was studied by \citet{fadaei17-2}). \citet{dughmi10} obtained a $2$-approximate, strategyproof mechanism for Weighted Bipartite Matching and a matching lower bound. Furthermore, they developed randomized strategyproof-in-expectation\footnote{A randomized mechanism is strategyproof-in-expectation if the true declaration of an agent maximizes their expected utility. Note that this is a weaker notion than universal strategyproofness, which we obtain for the randomized mechanisms in this work.} mechanisms for special cases of $\gap$; namely, a $2$-approximation for $\mkar$, a $4$-approximation for $\sigap$ and a $4$-approximation for a special case of $\vcgap$, termed \emph{Agent Value GAP} (see Section \ref{subsec:sigap-vcgap} for some discussion). 
Finally, they proposed a randomized, strategyproof-in-expectation $O(\log n)$-approximate mechanism for the general case. 
Subsequently, \citet{chen14} improved upon these results by devising mechanisms which satisfy universal strategyproofness, matching the guarantees of \citet{dughmi10} for these special cases. Additionally, they showed an improved $O(1)$-approximation for $\gap$.

Beyond the private graph model, for the setting where values are private information but monetary transfers are allowed, \citet{fadaei17} devised a $\nicefrac{e}{(e-1)}$-approximate, strategyproof-in-expectation mechanism. Finally, from an algorithmic perspective, the best known approximation ratio for $\gap$ is $\nicefrac{e}{(e-1)}-\epsilon$, for a fixed small $\epsilon>0$ due to \citet{feige06}. On the negative side, \citet{chakrabarty10} have shown that $\gap$ does not admit an approximation better that $\nicefrac{11}{10}$, unless $\text{P} = \text{NP}$.

\section{Preliminaries} \label{sec:prelim}

\subsection{Generalized Assignment Problem with Predictions} \label{sec:GAPwP}

In the \emph{Generalized Assignment Problem (\gap)}, we are given a bipartite graph $G=(\agentSet \cup \taskSet,\declaredEdges)$ consisting of a set $\agentSet = [n]$ of $n \ge 1$ agents (or items, jobs) and a set $\taskSet = [m]$ of $m \ge 1$ resources (or knapsacks, machines, respectively).\footnote{Throughout the paper, we use $[n] = \set{1, \dots, n}$ to refer to the set of the first $n \ge 1$ natural numbers.} Each agent $i \in \agentSet$ has a value $v_{ij} > 0$ and a size $\isize_{ij} > 0$ for being assigned to resource $j \in \taskSet$. Further, each resource $j \in \taskSet$ has a capacity $\jcap_j > 0$ (in terms of total size) that must not be exceeded. We assume without loss of generality that $\isize_{ij} \le C_j$ for every $i \in \agentSet$. Below, we use $\values = (v_{ij})_{i \in \agentSet, j \in \taskSet} \in \mathbb{R}^{n \times m}_{>  0}$ to refer to the matrix of all agent-resource values, $\sizes = (\isize_{ij})_{i \in \agentSet, j \in \taskSet} \in \mathbb{R}^{n \times m}_{>  0}$ to refer to the matrix of all agent-resource sizes and $\caps = (C_j)_{j \in \taskSet} \in \mathbb{R}^m_{> 0}$ to refer to the vector of all resource capacities.\footnote{We note that our assumption of all values, sizes and capacities being positive is without loss of generality in our model (which will become clear from further details introduced below).} 

The bipartite graph $G=(\agentSet \cup \taskSet,\declaredEdges)$ encodes compatibilities between agents and resources; we also refer to it as the \emph{compatibility graph}. An agent $i \in \agentSet$ is said to be \emph{compatible} with a resource $j \in \taskSet$ if $(i,j) \in \declaredEdges$; otherwise, $i$ is \emph{incompatible} with $j$.
We use $\declaredEdges_i = \set{(i,j) \in \declaredEdges}$ to denote the set of all compatible edges of agent $i$. Similarly, we use $\declaredEdges_j = \set{(i,j) \in \declaredEdges}$ to refer the set of all compatible edges of resource $j$. 
For example, the compatibility $(i,j) \in \declaredEdges$ might indicate that agent $i$ has access to resource $j$, or that item $i$ can be assigned to knapsack $j$, or that job $i$ can be executed on machine $j$. Generally, $\declaredEdges$ can be any subset of $\agentSet \times \taskSet$. We use $G[\declaredEdges] = (\agentSet \cup \taskSet, \declaredEdges)$ to refer to the compatibility graph induced by the edge set $\declaredEdges \subseteq \agentSet \times \taskSet$.\footnote{We can assume without loss of generality that $G$ does not contain any isolated nodes (as, otherwise, we can simply remove such nodes).}
Further, we write $\mathcal{I}_{\gap}=(G[\declaredEdges], \values, \sizes, \caps)$ to refer to an instance of the Generalized Assignment Problem. 

An \emph{assignment} $M \subseteq \agentSet \times \taskSet$ is a subset of edges such that each agent $i \in \agentSet$ is incident to at most one edge in $M$. Note that each agent is assigned to at most one resource, but several agents might be assigned to the same resource; we also say that $M$ is a \emph{many-to-one assignment}. 
If we additionally require that each resource $j \in \taskSet$ is incident to at most one edge in $M$, then $M$ is said to be a \emph{one-to-one assignment} (or, \emph{matching}, simply). Note that every matching is also an assignment. Given an agent $i \in \agentSet$, we use $M(i) = \sset{j \in \taskSet}{(i,j) \in M}$ to refer to the resource assigned to $i$ (if any); note that $M(i)$ is a singleton set. Also, we have $M(i) = \emptyset$ if $i$ is unassigned. Similarly, for a resource $j \in \taskSet$, we define $M(j) = \sset{i \in \agentSet}{(i,j) \in M}$ as the set of agents assigned to $j$ (if any); note that $M(j) = \emptyset$ if $j$ is unassigned. 

An assignment $M$ is said to be \emph{feasible} for a given compatibility graph $G[\declaredEdges]$ if (1) $M$ is an assignment in $G[\declaredEdges]$, i.e., $M \subseteq \declaredEdges$, and (2) $M$ satisfies all resource capacities constraints, i.e., for each resource $j \in \taskSet$, $\sum_{i \in M(j)} \isize_{ij} \le \jcap_j$. We define the \emph{value} $v(M)$ of an assignment $M$ as the sum of the values of all edges in $M$; more formally,
\begin{equation}\label{def:val}
v(M) = \sum_{(i,j) \in M} v_{ij}.
\end{equation}

We overload this notation slightly and also write $v(M(i))$ and $v(M(j))$ to refer to the total value of all edges assigned to agent $i$ or resource $j$, respectively. Further, we define $v_{i\emptyset} = 0$ and $v_{\emptyset j} = 0$ for notational convenience. We use $\mstar{\declaredEdges}$ to denote a feasible assignment of maximum value in the graph $G[\declaredEdges]$. We also say that $\mstar{\declaredEdges}$ is an \emph{optimal} assignment with respect to $\declaredEdges$.

\newsavebox{\taxonomytable}
\sbox{\taxonomytable}{
\scalebox{0.8}{
\begin{tabular}{@{}lcc@{}}
      \toprule
      \textbf{color} & \textbf{tractable} & \textbf{OPT $\Rightarrow$ SP} \\
      \midrule
      green & \cmark & \cmark \\
      orange & \cmark & \xmark \\
      blue & \xmark & \cmark \\
      red & \xmark & \xmark \\
      \bottomrule
\end{tabular}
}
}

\begin{figure}[t]
    \centering
    \begin{tikzpicture}[node distance=0.5cm, auto, scale=1]
      % Define styles
      \tikzstyle{block} = [ellipse, draw, minimum width=1cm, minimum height=1cm, text centered]
      \tikzstyle{arrow} = [thick,->,>=Stealth]
      % Nodes
      \node[block, fill=green!20] (unweighted_bmp) at (0,-2) {$\text{U\bmp}$};
      \node[block, fill=blue!20] (rmbp) at (2,-4) {ERMK};
      \node[block, fill=orange!20] (bmp) at (2,0) {$\bmp$};
      \node[block, fill=orange!20] (rsigap) at (5,0) {RSGAP};
      \node[block, fill=red!20] (sigap) at (8,0) {$\sigap$};
      \node[block, fill=red!20] (gap) at (10,-2) {$\gap$};
      \node[block, fill=blue!20] (mkar) at (5,-4) {$\mkar$};
      \node[block, fill=red!20] (vcgap) at (8,-4) {$\vcgap$};
      % Arrows
      \draw[arrow] (bmp) --  (unweighted_bmp);
      \draw[arrow] (rmbp) --  (unweighted_bmp);
      \draw[arrow] (sigap) -- (rsigap);
      \draw[arrow] (rsigap) -- (bmp);
      \draw[arrow] (gap) -- (sigap);
      \draw[arrow] (mkar) -- (rmbp);
      \draw[arrow] (vcgap) -- (mkar);
      \draw[arrow] (sigap) -- (mkar);
      \draw[arrow] (gap) -- (vcgap);

      \node at ([xshift=2.8cm, yshift=2.3cm]gap) 
      {\usebox{\taxonomytable}};
    \end{tikzpicture}
  \caption{Taxonomy of GAP variants resulting from \cite{dughmi10} and our results. The respective color (column 1) indicates whether the problem is polynomial-time solvable (column 2) and whether an optimal algorithm gives rise to a strategyproof mechanism that is $1$-efficient (column 3).}
  \label{fig:gap-variants}
\end{figure}
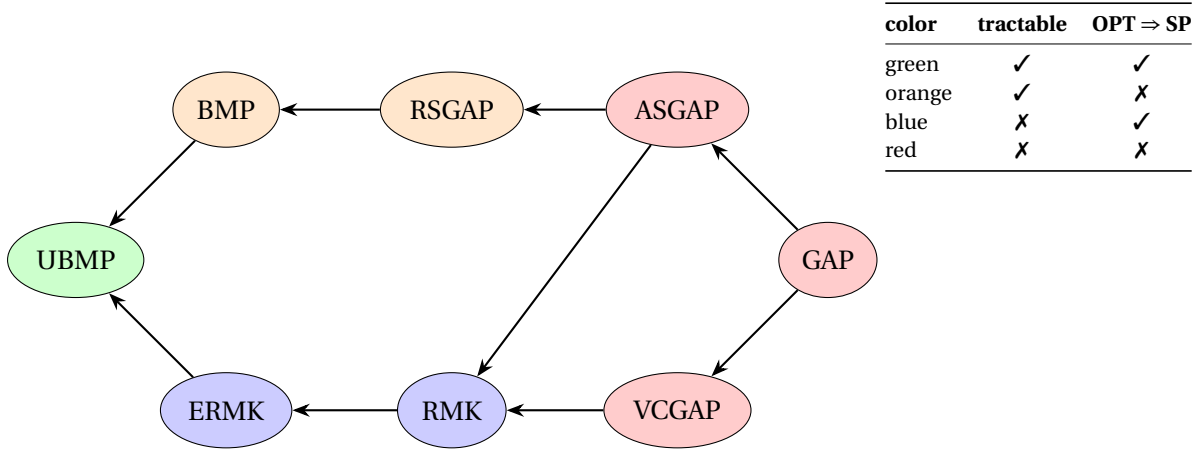

\minisec{$\gap$ Variants}
Below, we give an overview of the different special cases of \gap considered in this paper; see also Figure \ref{fig:gap-variants} for an illustration of the respective relationships between these problems.
\begin{itemize}
    \item \textit{Bipartite Matching Problem (\bmp).}
    Here, each agent $i \in \agentSet$ has unit size, i.e., $\isize_{ij} = 1$ for all $j \in \taskSet$, and each resource $j \in \taskSet$ has capacity $\jcap_j = 1$. We write $\mathcal{I}_{\bmp}=(G[\declaredEdges], \values)$ to denote an instance of \bmp.
    
    \sloppy
    \item \textit{Restricted Multiple Knapsack (\mkar)\footnote{This problem is also known as \emph{Multiple Knapsacks with Assignment Restrictions}; see, e.g., \citet{dawande00, nutov06, aerts03}.}.}
    Here, each agent $i$ has a fixed value $v_i = v_{ij}$ and size $\isize_i = \isize_{ij}$ for all $j \in \taskSet$. 
    We use $\mathcal{I}_{\mkar}=(G[\declaredEdges], (v_i)_{i \in \agentSet}, (\isize_i)_{i \in \agentSet}, \vec{\jcap})$ to denote an instance of \mkar. We mostly focus on the special case \emkar of \mkar, where $v_i = \isize_i$ for all $i$.
    We use $\mathcal{I}_{\emkar}=(G[\declaredEdges], (v_i = s_i)_{i \in \agentSet}, \vec{\jcap})$ to refer to an instance of \emkar.
    \fussy
    
    \item \textit{Value Consensus GAP (\vcgap).}
    Here, agents have some consensus about the value of the resources in the sense that there exists a permutation $\sigma$ of the resources such that for each agent $i \in \agentSet$ we have: $v_{i\sigma(1)} \ge v_{i\sigma(2)} \ge \dots \ge v_{i\sigma(m)}$. We note that both the \emph{Agent Value GAP} (i.e., $\forall i \in \agentSet,\ v_{ij} = v_{ik}\ \forall j,k \in \taskSet$) and the \emph{Resource Value GAP} (i.e., $\forall j \in \taskSet,\  v_{ij} = v_{kj}\ \forall i,k \in \agentSet$) fall into this case. Also, \mkar is a special case of the \vcgap. We write $\mathcal{I}_{\vcgap}=(G[\declaredEdges], \values, \sizes, \caps)$ to denote an instance of the \vcgap.
    
    \item \textit{Agent Size GAP (\sigap).}  
    Here, each agent $i \in \agentSet$ has the same size $s_i$ for all resources $j \in \taskSet$. Observe that both the \bmp and the \mkar are special cases of the \sigap. We use $\mathcal{I}_{\sigap}=(G[\declaredEdges], \values, \bm{s}, \vec{\jcap})$ to denote an instance of \sigap. 
    
    \item \textit{Resource Size GAP (RSGAP).}
    Here, for every resource $j \in \taskSet$, all agents have the same size, i.e., for all $j \in R$ and for all $i,k \in L$, $s_{ij}=s_{kj}$. Observe that the \bmp is a special case of the resource RSGAP. We use $\mathcal{I}_{RSGAP}=(G[\declaredEdges], \values, \bm{s}, \vec{\jcap})$ to refer to an instance of RSGAP. 
\end{itemize}

\minisec{GAP with Predictions}
In the setting with predictions, we are given an instance $\mathcal{I}_{\gap} = (G[\declaredEdges], \values, \sizes, \caps)$ of GAP, and, in addition, a \emph{predicted assignment} $\hat{M} \subseteq \agentSet \times \taskSet$. It is important to realize that the predicted assignment $\hat{M}$ is considered to be part of the input. Generally, $\hat{M}$ can be any assignment in the complete graph $G[\agentSet \times \taskSet]$. In particular, $\hat{M}$ does not necessarily correspond to a feasible assignment in the graph $G[\declaredEdges]$ (as it may contain edges which are not in $\declaredEdges$). We use $\mathcal{I}_{\gap^+} = (G[\declaredEdges], \values, \sizes, \caps, \hat{M})$ to refer to an instance of GAP augmented with a predicted assignment $\hat{M}$. We use a similar notation for the various special cases of \gap we study.

We say that $\hat{M}$ is a \emph{perfect prediction} for $\mathcal{I}_{\gap^+}$ if it corresponds to an assignment of maximum value in the graph $G[\declaredEdges]$, i.e., $v(\hat{M} \cap \declaredEdges) = v(M^*_{\declaredEdges})$. We define an error parameter that measures the quality of the predicted assignment $\hat{M}$ relative to an optimal assignment $\mstar{\declaredEdges}$ of $G[\declaredEdges]$. 
Namely, we define the \emph{prediction error} $\eta(\mathcal{I}_{\gap^+}) \in [0, 1]$ of an instance $\mathcal{I}_{\gap^+} = (G[\declaredEdges], \values, \sizes, \caps, \hat{M})$ as
\begin{equation} \label{eq:eta}
    \eta(\mathcal{I}_{\gap^+}) = 1 - \frac{v(\hat{M} \cap \declaredEdges)}{v(\mstar{\declaredEdges})}.
\end{equation}
\sloppy
Note that with this definition an instance $\mathcal{I}_{\gap^+}$ with a perfect prediction has a prediction error of $\eta(\mathcal{I}_{\textsc{GAP}^+}) = 0$. As the value of the predicted assignment $\hat{M}$ deteriorates from the value of the optimal assignment $M^*_{D}$, the error measure approaches $\eta(\mathcal{I}_{\gap^+}) \rightarrow 1$. If $\eta(\mathcal{I}_{\gap^+}) = 1$, we must have $v(\hat{M} \cap \declaredEdges) = 0$ which means that the prediction $\hat{M}$ does not contain any edge that is also in $\declaredEdges$ (recall that the values are assumed to be positive). 

\fussy
Note that the definition of our error parameter in \eqref{eq:eta} is meaningful as it captures the relative gap between the \emph{values} of the predicted assignment and the optimal one. Alternatively, one could compare structural properties of $\hat{M} \cap \declaredEdges$ and $\mstar{\declaredEdges}$. However, this seems less suitable in our context: For example, under an error notion that is not value-based, a predicted assignment may only miss one edge of an optimal assignment, i.e., $| \mstar{\declaredEdges} \setminus (\hat{M} \cap \declaredEdges) | = 1$, but still be of relatively low value if this missing edge is valuable. Further, a predicted assignment might contain none of the edges of the optimal assignment, i.e., $(\hat{M} \cap \declaredEdges) \cap \mstar{\declaredEdges} = \emptyset$, but still be very useful when its value is close to optimal; in fact, $(\hat{M} \cap \declaredEdges)$ might even be an optimal matching that is disjoint from $\mstar{\declaredEdges}$ (because the optimal assignment might not be unique).
Finally, note that accounting the value of edges in $\hat{M} \setminus \declaredEdges$ in a prediction error notion is not informative, as our goal is to compute a feasible assignment $M$ (i.e., $M \subseteq \declaredEdges$). All these cases are captured by the definition of our prediction error as in \eqref{eq:eta}. 

Given a fixed error parameter $\hat{\eta} \in [0,1]$, instances $\mathcal{I}_{\gap^+} = (G[\declaredEdges], \values, \sizes, \caps, \hat{M})$ with $\eta(\mathcal{I}_{\gap^+}) \le \hat{\eta}$ constitute the class of instances of prediction error at most $\hat{\eta}$. 

\minisec{Approximation Objectives}
We introduce the following three approximation notions for the Generalized Assignment Problem with predictions. 
\begin{itemize}
\sloppy
    \item \textit{Consistency:} A mechanism $\mathcal{M}$ is \emph{$\alpha$-consistent} with $\alpha \ge 1$ if for every instance $\mathcal{I}_{\gap^+} = (G[\declaredEdges], \values, \sizes, \caps, \hat{M})$ with a perfect prediction, i.e., $v(\hat{M} \cap \declaredEdges) = v(M^*_{\declaredEdges})$, the computed matching $M=\mathcal{M}({\declaredEdges})$ satisfies $\alpha \cdot v(M) \ge v(M^*_{\declaredEdges}).$
    
    \item \textit{Robustness:} A mechanism $\mathcal{M}$ is \emph{$\beta$-robust} with $\beta \ge 1$ if for every instance $\mathcal{I}_{\gap^+} = (G[\declaredEdges], \values, \sizes, \caps, \hat{M})$ with an arbitrary prediction $\hat{M}$, the computed matching $M=\mathcal{M}({\declaredEdges})$ satisfies $\beta \cdot v(M) \ge v(M^*_{\declaredEdges}).$
    
    \item \textit{Approximation:} A mechanism $\mathcal{M}$ is \emph{$g(\hat{\eta})$-approximate} with $g(\hat{\eta}) \ge 1$ if for every instance $\mathcal{I}_{\gap^+} = (G[\declaredEdges], \values, \sizes, \caps, \hat{M})$ with a prediction error of at most $\hat{\eta} \in[0,1]$, the computed matching $M=\mathcal{M}({\declaredEdges})$ satisfies $g(\hat{\eta}) \cdot v(M) \ge v(M^*_{\declaredEdges}).$
\end{itemize}

\subsection{Private Graph Model and Assignment Mechanisms} \label{subsec:privateGraph}

We study the Generalized Assignment Problem with predictions in a strategic environment. More specifically, we are interested in the setting where the agents are strategic and might misreport their actual compatibilities. To this aim, we use the \emph{private graph model} introduced by \citet{dughmi10}.

Here, each agent $i \in \agentSet$ has a private compatibility set $E_i \subseteq \set{(i,j) \in \agentSet \times \taskSet}$ specifying the set of edges that are truly compatible for $i$. Crucially, the compatibility set $\privateEdges_i$ is \emph{private} information, i.e., $\privateEdges_i$ is only known to agent $i$. In addition, each agent $i \in \agentSet$ declares a public compatibility set $\declaredEdges_i \subseteq \sset{(i, j)}{j \in \taskSet}$. The interpretation is that $i$ claims to be compatible with resource $j \in \taskSet$ if and only if $(i,j) \in \declaredEdges_i$; but these declarations might not be truthful, i.e., $\declaredEdges_i \neq \privateEdges_i$. We define $\declaredEdges = \cup_{i \in \agentSet} \declaredEdges_i \subseteq \agentSet \times \taskSet$ to refer to the union of all compatibility sets declared by the agents. We use $G[\declaredEdges] = (\agentSet \cup \taskSet, \declaredEdges)$ to refer to the compatibility graph induced by the declared edges in $\declaredEdges$. Similarly, we use $G[\privateEdges]$ to refer to the compatibility graph induced by the true compatibility sets of the agents, i.e., $\privateEdges = \cup_{i \in \agentSet} \privateEdges_i$. We refer to $G[\privateEdges]$ as the \emph{private graph model} (or, \emph{private graph} simply). 

\sloppy
Subsequently, we use $\left (\mathcal{I}_{\gap^+}, G[\privateEdges] \right )$ to refer to an instance $\mathcal{I}_{\gap^+}$ of the Generalized Assignment Problem with predictions in the private graph model $G[\privateEdges]$. 
We note that all input data of $\mathcal{I}_{\gap^+} = (G[\declaredEdges], \values, \sizes, \caps, \hat{M})$ is public information accessible by the mechanism, while the private graph $G[\privateEdges]$ is private information. 
For the sake of conciseness, we often omit input parameters which remain fixed; in fact, most of the time is will be sufficient to refer explicitly to the compatibility declarations $\declaredEdges$ only.

Given an instance $\left (\mathcal{I}_{\gap^+}, G[\privateEdges] \right )$ with compatibility declarations $\declaredEdges$, a deterministic mechanism $\mathcal{M}$ computes an assignment $M = \mathcal{M}(\declaredEdges)$ that is feasible for $\declaredEdges$. The \emph{utility} $u_i$ of agent $i \in \agentSet$ is defined as
\begin{equation} \label{eq:def:utility}
u_i(\declaredEdges) =
    \begin{cases}
	   v_{ij} & \text{if $(i,j) \in M \cap \privateEdges_i$,} \\
        0 & \text{otherwise.}
	\end{cases} 
\end{equation}
Note that the utility of agent $i$ is $v_{ij}$ if (1) $i$ is assigned to resource $j$ in $M$, i.e., $(i,j) \in M$, and (2) $i$ is truly compatible with resource $j$, i.e., $(i,j) \in \privateEdges_i$. In particular, the utility of $i$ is 0 if $i$ is unassigned in $M$, or if $i$ is matched to an incompatible resource. 
We assume that each agent wants to maximize their utility. To this aim, an agent $i$ might misreport their true compatibilities by declaring a compatibility set $\declaredEdges_i \neq \privateEdges_i$.
Note that we are considering a multi-parameter mechanism design problem here.

Note that if $M = \mathcal{M}(\declaredEdges)$ is the assignment computed by $\mathcal{M}$ for truthfully declared compatibilities, i.e., $\declaredEdges = \privateEdges$, then its value $v(M)$ (as defined in \eqref{def:val}) is equal to the sum of the utilities of the agents.

\minisec{Incentive Compatibility Objectives}
The following incentive compatibility notions will be relevant in this paper. 

\begin{itemize}
\sloppy
    \item \textit{Strategyproofness:} A mechanism $\mathcal{M}$ is \emph{strategyproof} (SP) if for every instance $\mathcal{I}_{\gap^+} = (G[\declaredEdges], \values, \sizes, \caps, \hat{M})$ and private graph $G[\privateEdges]$, it holds that for each agent $i \in \agentSet$
    \[
    \forall \declaredEdges'_i: \qquad 
    u_i(\privateEdges_i,{\declaredEdges}_{-i}) 
    \ge u_i(\declaredEdges'_i,{\declaredEdges}_{-i}).
    \]

    \item \textit{Group-Strategyproofness:} A mechanism $\mathcal{M}$ is \emph{group-strategyproof} (GSP) if for every instance $\mathcal{I}_{\gap^+} = (G[\declaredEdges], \values, \sizes, \caps, \hat{M})$ and private graph $G[\privateEdges]$, it holds that for every subset $S \subseteq \agentSet$
    \[
    \forall \declaredEdges'_S: \exists i \in S \qquad 
    u_i({\privateEdges}_S, {\declaredEdges}_{-S}) 
    \ge u_i({\declaredEdges}'_S, {\declaredEdges}_{-S}).
    \]
\end{itemize}

\minisec{Randomized Mechanisms} 
We also devise randomized mechanisms. A randomized mechanism is a probability distribution over a finite set of deterministic mechanisms. By extension, given a randomized mechanism $\mathcal{M}$ and an instance of the Generalized Assignment Problem with predictions $\mathcal{I}_{\gap^+}$, $\mathcal{M}(\mathcal{I}_{\gap^+})$ is a probability distribution over a finite set of feasible assignments for $\mathcal{I}_{\gap^+}$.

All randomized mechanisms we suggest in this work are universally strategyproof. A randomized mechanism $\mech$ is \emph{universally strategyproof} if its probability distribution is over a finite set of deterministic strategyproof mechanisms. The notion of \emph{universally group-strategyproof} is defined analogously. 
The three approximation objectives defined in Section~\ref{sec:GAPwP} extend naturally to randomized mechanisms (simply by replacing the value of an assignment with the expected value in the respective definition). 

\subsection{Stable Matching Preliminaries} \label{sec:stable-matching}
We introduce the notions and results from stable matching theory that we need in this paper. 
We are given a complete bipartite graph $G=(\agentSet \cup \taskSet, \privateEdges)$ consisting of a set of agents $\agentSet = [n]$, a set of resources $\taskSet = [m]$ and a set of edges $\privateEdges = \agentSet \times \taskSet$. Each agent $i \in \agentSet$ has a strict total preference order $\succ_i$ over $\privateEdges_i \cup \set{\emptyset}$, where $\privateEdges_i = \set{(i,j) \in \privateEdges}$ is the set of edges incident to $i$. Given two distinct edges $e, e' \in \privateEdges_i$, agent $i$ prefers $e$ over $e'$ if $e \succ_i e'$. The position of $\emptyset$ in the order indicates whether an edge is acceptable or not: $e \in \privateEdges_i$ is \emph{acceptable} if $e \succ_i \emptyset$; otherwise, $e$ is \emph{unacceptable}. Agent $i$ prefers to remain unmatched rather than being matched through an unacceptable edge. 
Similarly, each resource $j \in \taskSet$ has a strict total preference order $\succ_j$ over $\privateEdges_j \cup \set{\emptyset}$, where $\privateEdges_j = \set{(i,j) \in \privateEdges}$. All notions introduced above naturally extend to resource $j$ with preference order $\succ_j$. We refer to $(G[\agentSet \times \taskSet], (\succ_i)_{i \in\agentSet}, (\succ_j)_{j \in \taskSet})$ as a \emph{standard preference system} (i.e., if $G$ is complete and all preference orders are strict).\footnote{The model can be extended to incorporate arbitrary assignment restrictions $\privateEdges \subseteq \agentSet \times \taskSet$ and more complex resource preferences (see, e.g., \citep{HM05}), but we do not need this here.} 

We say that an edge $(i,j) \in \privateEdges$ is \emph{compatible} if $i$ is acceptable for $j$ and vice versa. A matching $M \subseteq \privateEdges$ is a subset of edges that are compatible such that no two distinct edges in $M$ share a common resource or agent. We use $M(i) = j$ to denote the \emph{mate} of agent $i$ with respect to $M$, i.e., $(i,j) \in M$; we write $M(i) = \emptyset$ if $i$ remains unmatched. Analogously, we use $M(j)$ to denote the mate of $j$. We say that an edge $e = (i,j) \in \privateEdges$ \emph{blocks} $M$ if (1) $i$ prefers to be matched through $e$ instead of $(i, M(i))$, i.e., $e \succ_i (i, M(i))$, and (2) $j$ prefers to be matched through $e$ instead of $(M(j), j)$, i.e., $e \succ_j (M(j),j)$. A matching $M$ is \emph{stable} if it is not blocked by any edge in $\privateEdges$.

In their seminal work, \citet{GS62} proposed an algorithm, also known as the \emph{Gale-Shapley algorithm} or \emph{\defacc}, that computes a stable matching for any given standard preference system. 
This result had several far-reaching consequences; indeed, it sparked the development of an entire theory on stable matchings. 
One of the consequences that will be useful in this paper is in the context of incentive compatibility: If the agents are strategic and can misreport their preference orders arbitrarily, then the agent-proposing \defacc guarantees group-strategyproofness. More precisely, no group of agents can jointly misreport their preferences such that each member of the group is strictly better off.

We summarize this result below. This result is attributed to \citet{Roth82} (strategyproofness) and, independently, to \citet{DF81} (group-strategyproofness). Later, \citet{GS85} gave a greatly simplified proof of group-strategyproofness (which we also adapt in Appendix~\ref{app:GSP} to our setting considered in Section~\ref{sec:matching}).

\begin{theorem}\label{thm:GS}
Let $(G[\agentSet \times \taskSet], (\succ_i)_{i \in\agentSet}, (\succ_j)_{j \in \taskSet})$ be a standard preference system. 
Then the agent-proposing \defacc is group-strategyproof. 
\end{theorem}

Note that for this result to hold it is crucial that only the agents can misreport their preferences; in particular, the preferences of the resources are assumed to be fixed and public.

\subsection{Lexicographic Extensions, Rank and Sort Operators} \label{sec:lex-order-and-sort}

The notions defined here will be useful throughout the paper.

Let $X = \set{x_1, \dots, x_n}$ be a set of $n \ge 1$ elements and assume that each element $x_i \in X$ is associated with $k \ge 1$ numerical values $z_1(x_i), \dots, z_k(x_i)$. We define $\lexext$ as the partial order over $X$ that we obtain by comparing the elements in $X$ lexicographically with respect to $(z_1, \dots, z_k)$, i.e.,
\begin{equation}\label{eq:lexext}
\forall i, j \in [n]: \qquad 
x_i \lexext x_j 
\quad :\Leftrightarrow\quad
\big(z_1(x_{i}), \dots, z_k(x_{i})\big)
\lex
\big(z_1(x_{j}), \dots, z_k(x_{j})\big).
\end{equation}
We also say that $\lexext$ is the \emph{extended lexicographic order} of $X$ with respect to $(z_1, \dots, z_k)$. We write $\lexexts$ instead of $\lexext$ if the order is strict.

We also introduce an operator $\mrank$ that orders the elements in $X$ by lexicographic decreasing order of their values $(z_1, \dots, z_k)$. More formally, given $X$ and $(z_1, \dots, z_k)$ as above, we define $\mrank(X, (z_1, \dots, z_k))= (\pi_1, \dots, \pi_n) \in \mathcal{S}_n$ such that for all $i, j \in [n]$ with $i < j$, we have:\footnote{Here, $\mathcal{S}_n$ denotes the set of all permutations of $[n]$.}
\[ 
\big(z_1(x_{\pi_i}), \dots, z_k(x_{\pi_i})\big)
\lex
\big(z_1(x_{\pi_j}), \dots, z_k(x_{\pi_j})\big).
\]
Further, we define an operator $\msort$ that sorts the elements in $X$ according to this order, i.e., 
\begin{equation}\label{eq:sort}
\msort(X, (z_1, \dots, z_k)) = \mlist{x_{\pi_1}, \dots, x_{\pi_n}}.    
\end{equation}

\section{Impossibility Results and the Baseline Mechanism} \label{sec:lowerBounds}

We prove a lower bound on the best possible trade-off in terms of consistency and robustness guarantees achievable by any deterministic strategyproof mechanism for GAP. We also derive a lower bound in terms of the error parameter $\eta$. Finally, we introduce a trivial mechanism, called \trust, that serves as a baseline mechanism in subsequent sections.

\minisec{Impossibility Results} We prove our lower bound for both the Bipartite Matching Problem (\bmp) and the Value Consensus GAP (\vcgap) with predictions in the private graph model. Clearly, this lower bound extends to all variants of $\gap^+$ that contain $\bmp^+$ or $\vcgap^+$ as a special case (see Figure \ref{fig:gap-variants}). 

\begin{restatable}{theorem}{threeone}\label{th:tradeOffConsisBoundedRobust}
Let $\gamma \ge 1$ be fixed arbitrarily. 
Then no deterministic strategyproof mechanism for $\bmp^+$ can achieve $(1+ \nicefrac{1}{\gamma})$-consistency and $(1 + \gamma - \epsilon)$-robustness for any $\epsilon > 0$.
\end{restatable}
\begin{figure}[t]
\centering
\begin{subfigure}[b]{0.3\linewidth}
    \centering
    \begin{tikzpicture}[auto, node_style/.style={circle,draw=black}]
        \node[agent] (v1) at (-1,1) {1}; \node[agent] (v2) at (-1,-1) {2}; \node[resource] (a) at (1.4,1) {$a$}; \node[resource] (b) at (1.4,-1) {$b$};
        \draw [line width = 3pt] (v1) edge node{$\gamma -\bar{\epsilon}$} (a); \draw (v1) edge node[yshift=0.3cm,xshift=-0.5cm]{$\gamma$} (b); \draw [line width = 3pt] (b) edge node{$1 + 2 \bar{\epsilon}$} (v2);
    \end{tikzpicture}
    \caption{ }
    \label{fig:tradeoffConsisRobusta}
    \end{subfigure}
\begin{subfigure}[b]{0.3\linewidth}
    \centering
    \begin{tikzpicture}[auto, node_style/.style={circle,draw=black}]
        \node[agent] (v1) at (-1,1) {1}; \node[agent] (v2) at (-1,-1) {2}; \node[resource] (a) at (1.4,1) {$a$}; \node[resource] (b) at (1.4,-1) {$b$};
        \draw [line width = 3pt] (v1) edge node{$\gamma -\bar{\epsilon}$} (a); \draw (v1) edge node[yshift=0.3cm,xshift=-0.5cm]{$\gamma$} (b); \draw (v2) edge node[yshift=-0.5cm, xshift=-0.5cm]{$1 - \bar{\epsilon}$} (a);  \draw [line width = 3pt] (b) edge node{$1 + 2 \bar{\epsilon}$} (v2);
    \end{tikzpicture}
    \caption{ }
    \label{fig:tradeoffConsisRobustb}
    \end{subfigure}
\begin{subfigure}[b]{0.3\linewidth}
    \centering
    \begin{tikzpicture}[auto, node_style/.style={circle,draw=black}]
        \node[agent] (v1) at (-1,1) {1}; \node[agent] (v2) at (-1,-1) {2}; \node[resource] (a) at (1.4,1) {$a$}; \node[resource] (b) at (1.4,-1) {$b$};
        \draw (v1) edge node[yshift=0.3cm,xshift=-0.5cm]{$\gamma$} (b); \draw (v2) edge node[yshift=-0.5cm, xshift=-0.5cm]{$1 - \bar{\epsilon}$} (a);  \draw [line width = 3pt] (b) edge node{$1 + 2 \bar{\epsilon}$} (v2);
    \end{tikzpicture}
    \caption{ }
    \label{fig:tradeoffConsisRobustc}
    \end{subfigure}
\caption{Instance used in the lower bound proof of Theorem~\ref{th:tradeOffConsisBoundedRobust}.}
 \label{fig:tradeoffConsisRobust}
\end{figure}
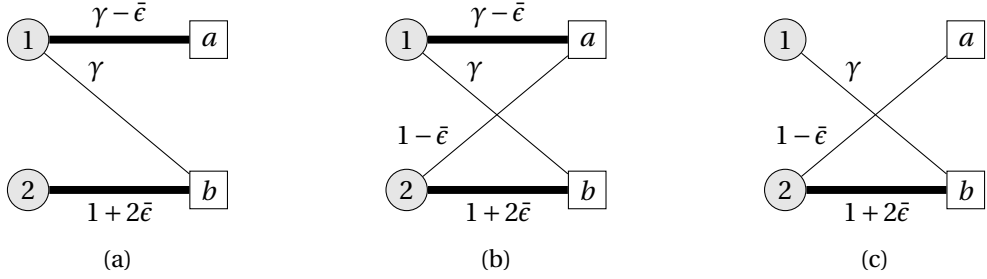
\begin{proof}
    The proof is by contradiction. 
    Suppose that for some $\gamma \ge 1$ there is a deterministic strategyproof mechanism $\mech$ that is $(1+ \nicefrac{1}{\gamma})$-consistent and $(1+\gamma-\epsilon)$-robust for some $\epsilon > 0$.
    Fix $\bar{\epsilon}$ to be a constant in $(0, \frac{\epsilon}{3+2\gamma-2\epsilon})$. 
    
    Consider the instance $\mathcal{I}_{\bmp^+} = (G[\declaredEdges], \values, \hat{M})$ with private graph $G[\privateEdges]$ depicted in Figure~\ref{fig:tradeoffConsisRobustb}, where the compatibility declarations are truthful, i.e., $\declaredEdges = \privateEdges$.
    Note that the values $\values$ of this instance depend on $\gamma$ and $\bar{\epsilon}$. Here, the predicted matching is $\hat{M} = \set{(1,a), (2,b)}$ (indicated in bold). Recall that we use $\mech(\declaredEdges)$ to denote the matching computed by $\mech$ for compatibility declarations $\declaredEdges$. Note that the optimal matching in $G[\declaredEdges]$ is $\mstar{\declaredEdges} = \hat{M}$. Because $\mech$ is $(1+\nicefrac{1}{\gamma})$-consistent, we have
    \[
        \Big(1+\frac{1}{\gamma}\Big) v(\mech(\declaredEdges)) \geq v(\mstar{\declaredEdges}) = 1+\gamma+\bar{\epsilon} >1 + \gamma.
    \]
    Here the last inequality holds because $\bar{\epsilon} > 0$.
    Dividing both sides by $(1+\gamma)/\gamma$ leads to $v(\mech(\declaredEdges)) > \gamma$. 
    Note that there are only two matchings in $G[\declaredEdges]$ with a value strictly larger than $\gamma$, namely $M_1 = \set{(1,b),(2,a)}$ and $M_2 = \set{(1,a), (2,b)}$. In particular, the above implies that the matching output by $\mech$ on $\declaredEdges$ must be either $\mech(\declaredEdges) = M_1$ or $\mech(\declaredEdges) = M_2$. We consider these two cases. 

    \medskip
   \underline{Case 1:} $\mech(\declaredEdges) = M_1 = \set{(1,b),(2,a)}$. Note that the utility of agent $2$ with respect to $\declaredEdges$ is $u_2(\mech(\declaredEdges)) = u_2(M_1) = v_{2a} = 1-\bar{\epsilon}$.
    Consider the compatibility declarations $\declaredEdges' = \declaredEdges \setminus \set{(2,a)}$ that we obtain from $\declaredEdges$ if agent $2$ deviates by hiding their edge $(2,a)$. The respective instance is depicted in Figure~\ref{fig:tradeoffConsisRobusta}. 
    Suppose we run $\mech$ on $\declaredEdges'$.
    Since $\mech$ is strategyproof, edge $(2,b)$ cannot be contained in $\mech(\declaredEdges')$. (To see this, note that otherwise the utility of agent $2$ with respect to $\declaredEdges'$ would be $v_{2b} = 1+2\bar{\epsilon} > 1-\bar{\epsilon} = u_2(\mech(\declaredEdges))$, contradicting that $\mech$ is strategyproof.)
    But then, we must have $v(\mech(\declaredEdges')) \le \max\{v_{1a}, v_{1b}\}= \gamma$. On the other hand, because $\mech$ is $(1+\nicefrac{1}{\gamma})$-consistent, it must hold that 
    \[
        \Big(1+\frac{1}{\gamma}\Big) v(\mathcal{M}(\declaredEdges')) 
        \geq v(\mstar{\declaredEdges'})
        = 1+\gamma+\bar{\epsilon} > 1 + \gamma.
    \]
    Here the last inequality holds because $\bar{\epsilon} > 0$, leading to a contradiction as we obtain $v(\mech(\declaredEdges')) > \gamma$.

    \medskip
    \underline{Case 2:} $\mech(\declaredEdges) = M_2 = \set{(1,a), (2,b)}$. Note that the utility of agent $1$ with respect to $\declaredEdges$ is $u_1(\mech(\declaredEdges)) = u_1(M_2) = v_{1a} = \gamma-\bar{\epsilon}$.
    Consider the compatibility declarations $\declaredEdges' = \declaredEdges \setminus \set{(1,a)}$ that we obtain from $\declaredEdges$ if agent $1$ deviates by hiding their edge $(1,a)$. The respective instance is depicted in Figure~\ref{fig:tradeoffConsisRobustc}. 
    Suppose we run $\mech$ on $\declaredEdges'$.
    Analogously to the argument given above, since $\mech$ is strategyproof, edge $(1,b)$ cannot be contained in $\mech(\declaredEdges')$. (To see this, note that otherwise the utility of agent $1$ with respect to $\declaredEdges'$ would be $v_{1b} = \gamma > \gamma-\bar{\epsilon} = u_1(\mech(\declaredEdges))$, contradicting that $\mech$ is strategyproof.)
    But then, we must have $v(\mech(\declaredEdges')) \le \max\{v_{2a}, v_{2b}\}= 1+2\bar{\epsilon}$. On the other hand, because $\mech$ is $(1+\gamma-\epsilon)$-robust, it must hold that 
    \[
    (1+\gamma-\epsilon) v(\mech(\declaredEdges')) \ge v(\mstar{\declaredEdges'}) = 1 + \gamma - \bar{\epsilon}.
    \]
     Dividing both sides by $(1+\gamma-\epsilon)$, we obtain 
     \[
     v(\mech(\declaredEdges')) \ge \frac{1 + \gamma - \bar{\epsilon}}{1+\gamma-\epsilon} > 1+2\bar{\epsilon},
     \]
     where the last inequality can be verified to hold if $\bar{\epsilon} < \frac{\epsilon}{3 + 2 \gamma - 2\epsilon}$ (which is a restriction we impose on the choice of $\bar{\epsilon}$). We thus obtain a contradiction.
\end{proof}

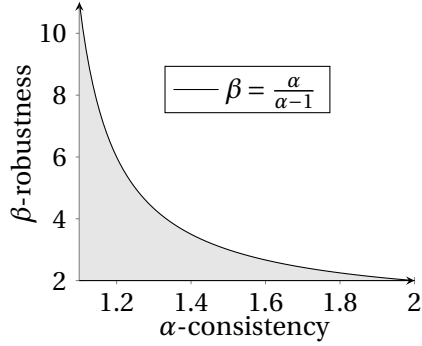
\begin{figure}[t]
\centering
    \begin{tikzpicture}[scale=0.65]
        \begin{axis}[ axis lines = left, xlabel = \( \alpha\text{-consistency } \), ylabel = {\(  \beta\text{-robustness } \)}, legend style={yshift=.5cm, xshift=.5cm}]
        \addplot [name path = A, domain=1.1:2, samples=100, color=black] {x/(x-1)};
        \addlegendentry{\(\beta = \frac{\alpha}{\alpha-1}\)};
        \addplot [name path = B, domain=1.1:2, samples=100, color=black] {2};
        \addplot [gray!20] fill between [of = A and B, soft clip={domain=1.1:2}];
        \end{axis}
    \end{tikzpicture}
    \caption{Impossibility trade-off in terms of $\alpha$-consistency and $\beta$-robustness. No deterministic strategyproof mechanism for \bmp can achieve a combination of $\alpha$ and $\beta$ in the gray area.} 
    \label{fig:tradeoffConsisRobustGraph}
\end{figure}

Note that the lower bound holds independently of any computational assumptions. For the setting without predictions, \citet{dughmi10} proved a lower bound of $2$ for \bmp. An illustration of the trade-off between consistency and robustness proven in Theorem~\ref{th:tradeOffConsisBoundedRobust} 
is given in Figure \ref{fig:tradeoffConsisRobustGraph}.
Note that as $\gamma \rightarrow \infty$, the consistency guarantee converges to $1$, but the robustness will be unbounded in this case. Therefore, no deterministic strategyproof mechanism can achieve $1$-consistency and bounded robustness. 
Also, note that in Theorem \ref{th:tradeOffConsisBoundedRobust}, $\gamma = 1+\epsilon$ for any fixed $\epsilon > 0$ is needed for a strategyproof mechanism to be 2-robust. However, this leads to an impossibility of $(1 + \frac{1}{1+\epsilon})$-consistency for any $\epsilon >0$. 
In particular, there is no deterministic strategyproof mechanism that is 2-robust and has a consistency strictly lower than $2$.

Theorem~\ref{th:tradeOffConsisBoundedRobust} proves a lower bound in terms of consistency versus robustness. The next theorem establishes a lower bound in terms of the error parameter $\eta$ as defined in \eqref{eq:eta} for any deterministic strategyproof mechanism that is $(1+\nicefrac{1}{\gamma})$-consistent.

\begin{restatable}{theorem}{threetwo}\label{th:lowerBoundWithError} 
Let $\gamma \ge 1$ be fixed arbitrarily.
Then no deterministic strategyproof mechanism for $\bmp^+$ can be $(1+\nicefrac{1}{\gamma})$-consistent and $(\frac{1}{1-\eta + \epsilon})$-approximate with $\epsilon > 0$ for any 
$\eta \in (0, \nicefrac{\gamma}{1+ \gamma}]$. 
\end{restatable}

\begin{proof}
The proof is by contradiction. Suppose that for some $\gamma \ge 1$ there is a deterministic strategyproof mechanism $\mech$ that is $(1+ \nicefrac{1}{\gamma})$-consistent and $\frac{1}{1-\eta + \epsilon}$-approximate for some $\epsilon > 0$. 
Let $\bar{\epsilon} > 0$ be some constant. 

Consider the instance $\mathcal{I}_{\bmp^+} = (G[\declaredEdges], \values, \hat{M})$ with private graph $G[\privateEdges]$ depicted in Figure~\ref{fig:tradeoffErrorB} and assume that compatibility declarations are truthful, i.e., $\declaredEdges = \privateEdges$. 
Here, the predicted matching is $\hat{M} = \set{(1,a), (2,b)}$ (indicated in bold). 
Let $\delta \in (\bar{\epsilon}, 1- \bar{\epsilon})$ such that all values are strictly positive. 

First, consider the instance depicted in Figure~\ref{fig:tradeoffErrorA} that we obtain if agent $2$ hides their edge $(2,a)$. Let $\declaredEdges' = \declaredEdges \setminus \set{(2,a)}$ be the corresponding declarations. 
Note that the optimal matching in $G[\declaredEdges']$ is $\mstar{\declaredEdges'} = \hat{M}$ and has value $v(\mstar{\declaredEdges'})= 1+\bar{\epsilon}$. 
Thus, for $\delta < (\nicefrac{\gamma}{1+ \gamma})(1+\bar{\epsilon})$, $\mech$ must compute a matching that contains edge $(2,b)$ to achieve $(1+\nicefrac{1}{\gamma})$-consistency. 

Next, consider the truthful declarations $\declaredEdges = \privateEdges$ as depicted in Figure~\ref{fig:tradeoffErrorB}. As $\mathcal{M}$ is strategyproof, it must compute a matching containing edge $(2,b)$.
Otherwise, agent 2 could increase their utility by hiding edge $(2,a)$, leading to the case as in Figure \ref{fig:tradeoffErrorA}. 
Thus, we must have $(2,b) \in \mech(\declaredEdges)$.

Finally, consider the instance depicted in Figure~\ref{fig:tradeoffErrorC} that we obtain if agent $1$ hides their edge $(1,a)$. Let $\declaredEdges' = \declaredEdges \setminus \set{(1,a)}$ be the corresponding declarations. 
As $\mathcal{M}$ is strategyproof, edge $(1,b)$ cannot be part of the matching $\mech(\declaredEdges')$ output by $\mech$. 
Otherwise, edge $(1,b)$ must be part of $\mech(\declaredEdges)$ as well, which is a contradiction to $(2,b) \in \mech(\declaredEdges)$. 
This implies that $v(\mech(\declaredEdges')) \le \max(v_{2a}, v_{2b}) = 1-\delta + 2\bar{\epsilon}$. Note that the optimal matching in $G[\declaredEdges']$ is $\mstar{\declaredEdges'} = \set{(1,b), (2,a)}$ and has value $v(\mstar{\declaredEdges'})= 1-\bar{\epsilon}$. On the other hand, the predicted matching in $G[\declaredEdges']$ is $\hat{M} \cap \declaredEdges' = \set{(2,b)}$ and has value $v(\hat{M} \cap \declaredEdges') = 1-\delta+2\bar{\epsilon}$. 
Recall that the error $\eta$ of this instance is $\eta = 1- v(\hat{M} \cap \declaredEdges')/v(\mstar{\declaredEdges'})$.
Because $\mech$ is $1/(1-\eta + \epsilon)$-approximate, it must hold that $v(\mech(\declaredEdges')) \ge (1-\eta + \epsilon) v(\mstar{\declaredEdges'})$. 

Combining this with $v(\mech(\declaredEdges')) \le 1-\delta + 2\bar{\epsilon} = v(\hat{M} \cap \declaredEdges')$ (as argued above) and using the definition of $\eta$, we obtain 
\[
(1-\eta) v(\mstar{\declaredEdges'}) = 
v(\hat{M} \cap \declaredEdges') \ge 
v(\mech(\declaredEdges')) 
\ge (1-\eta + \epsilon) v(\mstar{\declaredEdges'}),
\]
which gives a contradiction for any $\epsilon > 0$.

Note that we have $1-\eta = \frac{1 - \delta +2\bar{\epsilon}}{1 -\bar{\epsilon}}$, so $\eta \rightarrow \delta$ as $\bar{\epsilon} \rightarrow 0$. As the restrictions imposed on $\delta$ are $\delta \in (\bar{\epsilon}, 1- \bar{\epsilon})$ and $\delta < (\nicefrac{\gamma}{1+ \gamma}) (1+\bar{\epsilon})$, there exists a $\delta$ so that the contradiction holds for any $\eta \in (0, \nicefrac{\gamma}{1+ \gamma}]$. 
\end{proof}

\begin{figure}[t]
\centering
\begin{subfigure}[b]{0.3\linewidth}
    \centering
    \begin{tikzpicture}[auto, node_style/.style={circle,draw=black}]
        \node[agent] (v1) at (-1,1) {1}; \node[agent] (v2) at (-1,-1) {2}; \node[resource] (a) at (1.4,1) {$a$}; \node[resource] (b) at (1.4,-1) {$b$};
        \draw [line width = 3pt] (v1) edge node{$\delta -\bar{\epsilon}$} (a); \draw (v1) edge node[yshift=0.3cm,xshift=-0.5cm]{$\delta$} (b); \draw [line width = 3pt] (b) edge node{$1-\delta + 2 \bar{\epsilon}$} (v2);
    \end{tikzpicture}
    \caption{ }
    \label{fig:tradeoffErrorA}
    \end{subfigure}
\begin{subfigure}[b]{0.3\linewidth}
    \centering
    \begin{tikzpicture}[auto, node_style/.style={circle,draw=black}]
        \node[agent] (v1) at (-1,1) {1}; \node[agent] (v2) at (-1,-1) {2}; \node[resource] (a) at (1.4,1) {$a$}; \node[resource] (b) at (1.4,-1) {$b$};
        \draw [line width = 3pt] (v1) edge node{$\delta-\bar{\epsilon}$} (a); \draw (v1) edge node[yshift=0.3cm,xshift=-0.5cm]{$\delta$} (b); \draw (v2) edge node[yshift=-0.5cm, xshift=-0.5cm]{ $1-\delta - \bar{\epsilon}$}(a);  \draw [line width = 3pt] (b) edge node{$1-\delta + 2 \bar{\epsilon}$} (v2);
    \end{tikzpicture}
    \caption{ }
    \label{fig:tradeoffErrorB}
    \end{subfigure}
\begin{subfigure}[b]{0.3\linewidth}
    \centering
    \begin{tikzpicture}[auto, node_style/.style={circle,draw=black}]
        \node[agent] (v1) at (-1,1) {1}; \node[agent] (v2) at (-1,-1) {2}; \node[resource] (a) at (1.4,1) {$a$}; \node[resource] (b) at (1.4,-1) {$b$};
        \draw (v1) edge node[yshift=0.3cm,xshift=-0.5cm]{$\delta$} (b); \draw (v2) edge node[yshift=-0.5cm, xshift=-0.5cm]{$1-\delta - \bar{\epsilon}$} (a);  \draw [line width = 3pt] (b) edge node{$1-\delta + 2 \bar{\epsilon}$} (v2);
    \end{tikzpicture}
    \caption{ }
    \label{fig:tradeoffErrorC}
    \end{subfigure}
\caption{Instance used in the lower bound proof of Theorem~\ref{th:lowerBoundWithError}.}
 \label{fig:tradeoffError}
\end{figure}
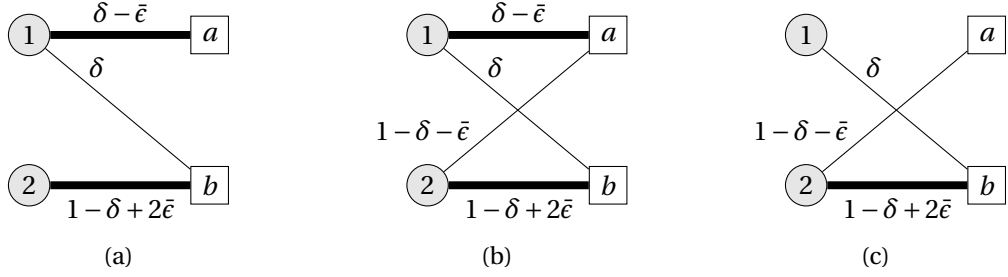

\begin{remark} \label{remark:LBvcgap}
Theorems \ref{th:tradeOffConsisBoundedRobust} and \ref{th:lowerBoundWithError} also hold for $\vcgap^+$. Namely, the instances used in the proofs (see Figures \ref{fig:tradeoffConsisRobust} and \ref{fig:tradeoffError}) are both instances of $\vcgap^+$ with $\sigma(1) = b$, $\sigma(2)=a$ and unit sizes and capacities. 
\end{remark}

\minisec{Baseline Mechanism} We conclude this section by introducing a na\"ive mechanism for $\gap^+$ in the private graph model, which simply adheres to the prediction: Given an instance $\mathcal{I}_{\gap^+} = (G[\declaredEdges], \values, \sizes, \caps, \hat{M})$, the mechanism returns the assignment $\hat{M} \cap \declaredEdges$. We call this mechanism \trust (see Mechanism~\ref{alg:trus}).

\begin{mechanism}[ht]
\SetKwInput{Input}{Input}
\SetKwInput{Output}{Output}
\DontPrintSemicolon
\caption{\label{alg:trus}
$\trust(\mathcal{I}_{\gap^+})$}
\Input{An instance $\mathcal{I}_{\gap^+} =(G[\declaredEdges], \values, \sizes, \vec{\jcap}, \hat{M})$.}
\Output{A feasible assignment for $\mathcal{I}_{\gap^+}$.}
\Return $\hat{M} \cap \declaredEdges$.
\end{mechanism}

It is trivial to see that \trust is $1$-consistent. It is also not hard to prove that \trust\ is group-strategyproof and achieves an optimal approximation guarantee matching the lower bound in Theorem~\ref{th:lowerBoundWithError} (as $\gamma \rightarrow \infty$).

\begin{restatable}{theorem}{trustthm}\label{lemma:trust-gsp}
Fix some error parameter $\hat{\eta} \in [0,1)$. Consider the class of instances of $\gap^+$ in the private graph model with prediction error at most $\hat{\eta}$.
Then, \trust is group-strategyproof and achieves an optimal approximation guarantee of $1/(1-\hat{\eta})$. 
\end{restatable}

\begin{proof}
Let $(I_{\gap^+}, G[E])$ be an instance of $\gap^+$ in the private graph model. 
Note that each agent $i \in \agentSet$ will be included in the assignment $M$ computed by $\trust(\mathcal{I}_{\gap^+})$ if and only if $(i, \hat{M}(i)) \in D_i$. Therefore, for each agent $i$ with $(i, \hat{M}(i)) \in E_i$, it is a dominant strategy to include $(i, \hat{M}(i))$ in their set $\declaredEdges_i$ of compatibility declarations. All other agents $j$ with $(j, \hat{M}(j)) \not \in E_j$ receive zero utility, independently of their compatibility declarations. This proves group-strategyproofness. 

Also, by the definition of the error parameter $\eta(\mathcal{I}_{\gap^+}) \le \hat{\eta}$, the approximation guarantee of \trust follows trivially: 
$v(\hat{M} \cap \declaredEdges) 
= (1- \eta(\mathcal{I}_{\gap^+}))  v(\mstar{\declaredEdges}) 
\ge (1-\hat{\eta}) v(\mstar{\declaredEdges})$.
\end{proof}

We conclude from Theorem~\ref{lemma:trust-gsp} that \trust realizes our strongest notion of incentive compatibility (i.e., group-strategyproofness) and even achieves the best possible consistency and approximation guarantees. But the point is, that it completely fails to achieve any bounded robustness guarantee.\footnote{To see this, just consider an instance $\mathcal{I}_{\gap^+}$ with $\hat{M} \cap \declaredEdges =\emptyset$.}
In a nutshell, this demonstrates that the actual challenge in deriving strategyproof mechanisms for $\gap^+$ in the private graph model is to achieve the best possible trade-off in terms of consistency/approximation \emph{and} robustness guarantees; without the latter, the whole problem becomes trivial (as \trust is the best possible mechanism). 
Despite this deficiency, and perhaps surprisingly, we will use this non-robust mechanism \trust as an important building block in our randomized mechanisms described in Section~\ref{sec:randomized}.

We conclude this section with a simple observation that we already used in the proof of Theorem~\ref{lemma:trust-gsp}. Because we will reuse it several times throughout the paper, we summarize it in the following lemma and refer to it as the Lifting Lemma. 

\begin{lemma}[Lifting Lemma]\label{lem:lifting}
Fix some error parameter $\hat{\eta} \in [0,1)$.
Consider the class of instances of $\gap^+$ in the private graph model with prediction error at most $\hat{\eta}$. 
If for every instance $\mathcal{I}_{\gap^+}$, $\mech(\mathcal{I}_{\gap^+})$ returns an assignment $M$ such that $\alpha \cdot \mathbb{E}[v(M)] \ge v(\hat{M} \cap \declaredEdges)$ for some $\alpha \ge 1$, then $\mech$ is $\alpha/(1-\hat{\eta})$-approximate.
\end{lemma}

\begin{proof}
By definition, we have $v(\hat{M} \cap \declaredEdges) = (1-\eta(\mathcal{I}_{\gap^+})) v(\mstar{\declaredEdges})$. Also, $\eta(\mathcal{I}_{\gap^+})) \le \hat{\eta}$. Thus,
\[
\alpha \cdot \mathbb{E}[v(M)] \ge 
v(\hat{M} \cap \declaredEdges) 
= (1- \eta(\mathcal{I}_{\gap^+})) v(\mstar{\declaredEdges}) 
\ge (1-\hat{\eta}) v(\mstar{\declaredEdges}). \qedhere
\]\qedhere
\end{proof}

\section{Our Mechanism for Bipartite Matching with Predictions} \label{sec:matching}

We introduce our mechanism, called \boost, for \bmp with predictions ($\bmp^+$) in the private graph model. Our mechanism is inspired by the \emph{\defacc} by \citet{GS62}. \boost is parameterized by some $\gamma \ge 1$, which we term the \emph{confidence parameter}. Put differently, \boost defines an (infinite) family of deterministic mechanisms, one for each choice of $\gamma \ge 1$. It is important to realize that all properties proved below, hold for an arbitrary choice of $\gamma \ge 1$. \boost also constitutes an important building block to derive our randomized mechanisms for certain special cases of GAP in Section~\ref{sec:randomized}. 

\subsection{\boost Mechanism} \label{subsec:boosted-proposal}

\boost receives as input an instance $\mathcal{I}_{\bmp^+}=(G[\declaredEdges], \values, \hat{M})$ of $\bmp^+$ and a confidence parameter $\gamma \ge 1$. 
Our mechanism maintains a (tentative) matching $M$ and a subset of agents $A \subseteq \agentSet$ that are called \emph{active}.
An agent $i$ is \emph{active} if it is not tentatively matched to any resource and has some remaining proposal to make; otherwise, $i$ is \emph{inactive}. 
Initially, the matching is empty, i.e., $M = \emptyset$, and all agents are active, i.e., $A = \agentSet$. 
In each iteration, the mechanism chooses an active agent $i \in A$ who then makes an offer to an adjacent resource $j \in \taskSet$ by following a specific proposal order:

\begin{mechanism}[t]
\SetKwInput{Input}{Input}
\SetKwInput{Output}{Output}
\DontPrintSemicolon
\caption{\label{alg:bpm}
$\boost(\mathcal{I}_{\bmp^+}, \gamma)$}
\Input{An instance $\mathcal{I}_{\bmp^+} =(G[\declaredEdges], \values, \hat{M})$, confidence parameter $\gamma \geq 1$.}
\Output{A feasible matching $M$ for $\mathcal{I}_{\bmp^+}$.}
Initialize $M= \emptyset, A=\agentSet$, $P_i=\declaredEdges_i$ for each $i \in \agentSet$. \;
\While{$A \neq \emptyset$}{
    Choose $i \in A$ and let $(i,j) = \argmax \sset{v_{ij}}{(i,j) \in P_i}$. 
    \label{alg:choice}
    \algcom{determine next proposal $(i,j)$}
    Agent $i$ offers $\theta_{ij} = \theta_{ij}(\gamma, \hat{M})$ to resource $j$. \algcom{$i$ makes (boosted) offer to $j$}
    \If(\tcp*[f]{check if $i$'s offer is highest for $j$}){\label{alg:if-line}$\theta_{ij} > \theta_{M(j)j}$}{
        \lIf(\tcp*[f]{$j$ rejects current mate $M(j)$ (if any)}){$M(j) \neq \emptyset$}{$M = M \setminus \set{(M(j), j)}$}
        $M = M \cup \set{(i,j)}$ \algcom{$i$ tentatively matched to $j$}
        $A = A \cup M(j) \setminus \set{i}$ \algcom{update active agents} 
    }
    $P_i=P_i \setminus \{(i,j)\}$ \algcom{update $i$'s proposal set}
    \lIf(\tcp*[f]{remove $i$ from $A$ if no more proposals}){$P_i = \emptyset$}{$A=A\setminus \{i\}$}
}
\Return $M$
\end{mechanism}

\begin{description}
\item[Agent Proposal Order:]

Each agent $i \in \agentSet$ maintains an order on their set of incident edges $\declaredEdges_i = \set{(i,j) \in \declaredEdges}$ by sorting them according to non-increasing values $v_{ij}$. 
We assume that ties are resolved according to a fixed tie-breaking rule $\tau_i$.\footnote{Note that the choice of the edge $(i,j)$ of maximum value $v_{ij}$ in Line~\ref{alg:choice} is uniquely determined by this order.}
\end{description}

\medskip
\noindent
The key idea behind our mechanism is that the value $v_{ij}$ that $i$ proposes to $j$ is \emph{boosted} if the edge $(i,j)$ is part of the predicted matching $\hat{M}$. We make this idea more concrete: 
given an agent $i$ and a declared edge $(i,j) \in \declaredEdges_i$, we define the \emph{offer} $\theta_{ij} = \theta_{ij}(\gamma, \hat{M})$\footnote{Note that $\theta_{ij}(\gamma, \hat{M})$ depends on the confidence parameter $\gamma$ and the predicted matching $\hat{M}$; but we omit these arguments and simply write $\theta_{ij}$ if they are clear from the context.} for resource $j$ as
\begin{equation}\label{eq:offer}
\theta_{ij}(\gamma, \hat{M}) = 
\begin{cases}
v_{ij} & \text{if $(i,j) \notin \hat{M}$,} \\
\gamma v_{ij} & \text{if $(i,j) \in \hat{M}$.}
\end{cases}
\end{equation}
Based on this definition, when it is $i$'s turn to propose to resource $j$, then the offer that $j$ receives from $i$ is the actual value $v_{ij}$ if $(i,j)$ is not a predicted edge, while it is the boosted value $\gamma v_{ij}$ if $(i,j)$ is a predicted edge. Intuitively, this way our mechanisms increases the chance that an agent proposing through a predicted edge is accepted (see below) by amplifying the offered value by a factor $\gamma \ge 1$.

Suppose resource $j$ receives offer $\theta_{ij}$ from agent $i$. Then $j$ \emph{accepts} $i$ if $\theta_{ij}$ is the largest offer that $j$ received so far; otherwise, $j$ \emph{rejects} $i$. We define $\theta_{\emptyset j} = 0$ to indicate that the highest offer that $j$ received is zero if $j$ is still unmatched, i.e., $M(j) = \emptyset$. To this aim, each resource $j$ maintains a fixed preference order over their set of incident edges.

\begin{description}
\item[Resource Preference Order:]
Each resource $j$ maintains an order on their set of incident edges $\declaredEdges_j = \set{(i,j) \in \declaredEdges}$ by sorting them according to non-increasing offer values $\theta_{ij}$. We assume that ties are resolved according to a fixed tie-breaking rule $\tau_j$.\footnote{Note that the comparison in Line~\ref{alg:if-line} is done with respect to this order.}
\end{description}

\medskip\noindent
If $i$ is accepted, then $i$ becomes \emph{tentatively matched} to $j$, i.e., $(i,j)$ is added to $M$, and $i$ becomes inactive. Also, if there is some agent $k$ that was tentatively matched to $j$ before, then $k$ is rejected by $j$, i.e., $(k,j)$ is removed from $M$, and $k$ becomes active again. Whenever an agent gets rejected, it moves on to the next proposal (if any) according to their offer order; in particular, an agent proposes at most once to each adjacent resource. 

The mechanism terminates when all agents are inactive, i.e., $A = \emptyset$. The current matching becomes definite and is output by the mechanism. Note that we do not specify how an agent $i$ is chosen from the set of active agents $A$ in Line~\ref{alg:choice}. In fact, any choice will work here. For example, a natural choice is to always choose an active agent $i \in A$ whose next offer $\theta_{ij}$ is largest.     

Intuitively, the confidence parameter $\gamma \ge 1$ specifies to which extent \boost follows the prediction. On the one extreme, for $\gamma = 1$ our mechanism ignores the prediction, which is the best choice in terms of achieving optimal robustness (at the expense of achieving worst consistency). As $\gamma$ increases, our mechanism follows the prediction more and more. On the other extreme, for $\gamma \rightarrow \infty$ our mechanism becomes \trust\ (as introduced in Section~\ref{sec:lowerBounds}) and simply returns the predicted matching; naturally, this is the best choice in terms of achieving optimal consistency (at the expense of unbounded robustness). 

The following is the main result of this section: 

\begin{restatable}{theorem}{bpmthm}\label{thm:bpm}
 Fix some error parameter $\hat{\eta} \in [0,1]$.
Consider the class of instances of $\bmp^+$ in the private graph model with prediction error at most $\hat{\eta}$.
Then, for every confidence parameter $\gamma \ge 1$, \boost\ is group-strategyproof and has an approximation guarantee of
\begin{equation}\label{eq:apx-BMP}
g(\hat{\eta}, \gamma) = 
\begin{cases}
    \frac{1+\gamma}{\gamma(1-\hat{\eta})} & \text{if $\hat{\eta} \le 1 - \nicefrac{1}{\gamma}$,} \\
    1 + \gamma & \text{otherwise.}
\end{cases}
\end{equation}
In particular, \boost\ is $(1+\nicefrac{1}{\gamma})$-consistent and $(1+\gamma)$-robust, which is best possible.   
\end{restatable}
Note that \boost is not only strategyproof, but satisfies the stronger incentive compatibility notion of group-strategyproofness, which we prove in Section~\ref{subsec:bpm-gsp}.
Also, in light of the lower bound given in Theorem~\ref{th:tradeOffConsisBoundedRobust}, \boost achieves the best possible trade-off in terms of consistency and robustness guarantees. 
Note that the approximation guarantee retrieves the best possible $(1+1/\gamma)$-consistency guarantee for $\hat{\eta} = 0$ and $(1+\gamma)$-robustness guarantee for $\hat{\eta} \ge 1-1/\gamma$. 
For the range $\hat{\eta} \in [0, 1-1/\gamma]$, as $\hat{\eta}$ increases the approximation interpolates between the consistency and the robustness guarantee (see Figure~\ref{fig:predErrorA}). 
For $\hat{\eta} \in (0, 1-1/\gamma)$, the upper bound for \boost as stated above is off by a factor of $1+\nicefrac{1}{\gamma}$ from the lower bound proven in Theorem~\ref{th:lowerBoundWithError} (see Figure~\ref{fig:LBErrorA} and Figure~\ref{fig:LBErrorB}).

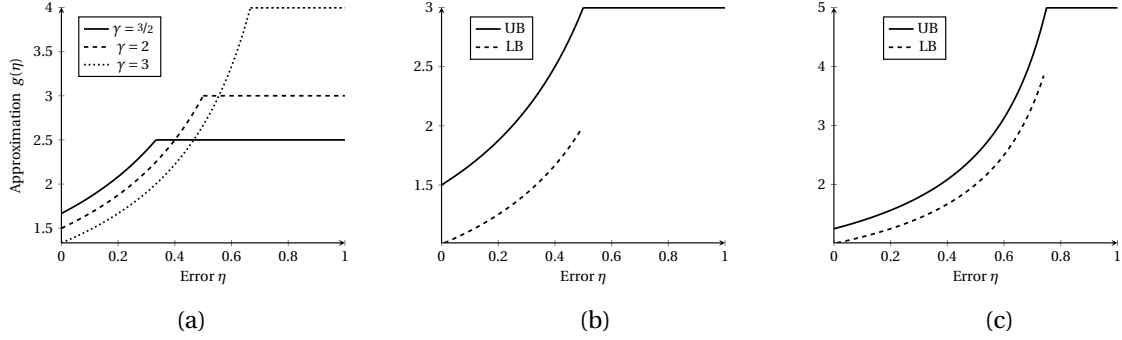
\begin{figure}[t]
\begin{subfigure}[b]{0.3\linewidth}
    \begin{tikzpicture}[scale=0.55]
        \begin{axis}[ thick, axis lines = left, xlabel = \( \text{Error } \eta\), ylabel = {\(  \text{Approximation\ \ } g(\eta)\)}, legend style={yshift=-0.1cm, xshift=-4.3cm}]
        \addplot [domain=0:((2.25-1)/(2.25+1.5)), samples=100, color=black, forget plot, very thick] {(1+1/1.5)/(1-x)};
        \addplot [domain=((2.25-1)/(2.25+1.5)):1, samples=100, color=black, very thick] {1+1.5};
        \addplot [dashed, domain=0:((4-1)/(4+2)), samples=100, color=black, very thick] {(1+1/2)/(1-x)}; 
        \addplot [dashed, domain=((4-1)/(4+2)):1, samples=100, color=black, forget plot, very thick] {1+2};
        \addplot [dotted, domain=0:((9-1)/(9+3)), samples=100, color=black, very thick] {(1+1/3)/(1-x)}; 
        \addplot [dotted, domain=((9-1)/(9+3)):1, samples=100, color=black, forget plot, ultra thick] {1+3};
        \addlegendentry{\( \gamma = \nicefrac{3}{2}\)};        
        \addlegendentry{\( \gamma = 2\)};                
        \addlegendentry{\( \gamma = 3\)};
        \end{axis}
    \end{tikzpicture}
    \caption{}       
    \label{fig:predErrorA}
\end{subfigure}
\quad
\begin{subfigure}[b]{0.3\linewidth}
    \begin{tikzpicture}[scale=0.55]
        \begin{axis}[ thick, axis lines = left, xlabel = \( \text{Error\ } \eta\), legend style={yshift=-0.1cm, xshift=-4.5cm}]
        \addplot [domain=0:((4-1)/(4+2)), samples=100, color=black, very thick] {(1+1/2)/(1-x)}; 
        \addplot [ domain=((4-1)/(4+2)):1, samples=100, color=black, ultra thick, forget plot] {1+2};
        \addplot [dashed, domain=0.01:0.49, samples=100, color=black, very thick] {1/(1-x)}; 
        \addlegendentry{UB};
        \addlegendentry{LB};
        \end{axis}
    \end{tikzpicture}
    \caption{}
    \label{fig:LBErrorA}
\end{subfigure}
\quad
\begin{subfigure}[b]{0.3\linewidth}
    \begin{tikzpicture}[scale=0.55]
        \begin{axis}[ thick, axis lines = left, xlabel = \( \text{Error } \eta\), legend style={yshift=-0.1cm, xshift=-4cm}]   
        \addplot [domain=0:((16-1)/(16+4)), samples=100, color=black, very thick] {(1+1/4)/(1-x)}; 
        \addplot [ domain=((16-1)/(16+4)):1, samples=100, color=black, ultra thick, forget plot] {1+4};
        \addplot [dashed, domain=0.01:0.74, samples=100, color=black, very thick] {1/(1-x)}; 
        \addlegendentry{UB};
        \addlegendentry{LB};
        \end{axis}
    \end{tikzpicture}
    \caption{}
    \label{fig:LBErrorB}
\end{subfigure}
\caption{Approximation guarantee $g(\hat{\eta})$ as a function of $\eta$. 
(a) For $\gamma \in \set{\nicefrac{3}{2}, 2, 3}$. 
(b) Upper vs.~lower bound for $\gamma = 2$. 
(c) Upper vs.~lower bound for $\gamma = 4$.} 
\label{fig:predError}
\end{figure}
\subsection{\boost: Group-Strategyproofness}
\label{subsec:bpm-gsp}

In this section, we show that \boost (Mechanism \ref{alg:bpm}) is group-strategyproof. The idea behind our proof is to reduce our setting to a one-to-one stable matching problem with standard preferences and argue that \boost is an instantiation of the (agent-proposing) \defacc by \citet{GS62}. Our mechanism then inherits the incentive compatibility properties of the \defacc by Theorem~\ref{thm:GS}. (We also provide a self-contained proof in Appendix~\ref{app:GSP}.)

In our proof below, it will be crucial that both the agents and the resources use fixed tie-breaking rules $(\tau_i)_{i \in \agentSet}$ and $(\tau_j)_{j \in \taskSet}$, respectively, as assumed above. 

\begin{proof}[Proof of Theorem~\ref{thm:bpm} (group-strategyproofness)]
Let $\gamma \ge 1$ be fixed. 
Suppose we are given an instance $\mathcal{I}_{\bmp^+}=(G[\declaredEdges], \values, \hat{M})$ of $\bmp^+$ with compatibility declarations $\declaredEdges$ and a private graph $G[\privateEdges]$.
Based on this, we construct a standard preference system $(G[\agentSet\times\taskSet], (\succ_i)_{i \in \agentSet}, (\succ_j)_{j \in \taskSet})$ sucht that
\begin{enumerate}
\item the preference order $\succ_i$ of each agent $i \in \agentSet$ mimics $i$'s proposal order over $\declaredEdges_i$, and
\item the preference order $\succ_j$ of each resource $j \in \taskSet$ encodes $j$'s preference order over $\declaredEdges_j$.
\end{enumerate}

In order to be able to apply Theorem~\ref{thm:GS}, it is crucial that only the agents can manipulate their preference orders (corresponding to compatibility declarations). The preference orders of the resources must be independent of the declarations and remain fixed. 

\medskip
We start with the definition of the preference orders of the agents. 

\begin{description}
\item[Agent Preference Order:]
We define the preference order $\succ_i$ of each agent $i$ as follows. Let $\declaredEdges_i$ be the compatibility set of agent $i$. Let $\lexexts_i$ be the extended lexicographic order of $\declaredEdges_i$ with respect to $(v_{ij}, \tau_i)$ as defined in \eqref{eq:lexext} in Section \ref{sec:lex-order-and-sort}, where $\tau_i$ refers to the tie-breaking rule used for agent $i$. (Note that this ensures that $\lexexts$ is a strict order.) 
Then $\succ_i$ is defined over the set of all edges in $\set{i} \times \taskSet$ as follows: 
\begin{equation}\label{eq:spref-agents}
\succ_i: \quad \mlist{~~\lexexts_i~,~~\emptyset~,~~  (\set{i} \times \taskSet) \setminus \declaredEdges_i~~}.
\end{equation}
That is, $\succ_i$ first orders the edges in $\declaredEdges_i$ as in $\lexexts$, succeeded by $\emptyset$, succeeded by any order of the edges in $(\set{i} \times \taskSet) \setminus \declaredEdges_i$.
\end{description}

By construction, $e$ is acceptable ($e \succ_i \emptyset$) if and only if $e$ is declared ($e \in \declaredEdges_i$). Also, all declared edges in $\declaredEdges_i$ are ordered by non-increasing values $v_{ij}$ (breaking ties according to $\tau_i$). 
The latter is crucial to align the preferences of the agents with their utilities as defined in \eqref{eq:def:utility}. The order on the set of unacceptable edges in $(\set{i} \times \taskSet) \setminus \declaredEdges_i$ is actually irrelevant here (any order will do). 

\medskip
We continue with the definition of the preference orders of the resources. 

\begin{description}
\item[Resource Preference Order:]
We define the preference order $\succ_j$ of each resource $j$ as follows.
Let $\lexexts_j$ be the extended lexicographic order of $\agentSet \times \set{j}$ with respect to $(\theta_{ij}, \tau_j)$ as defined in \eqref{eq:lexext} in Section \ref{sec:lex-order-and-sort}, where $\tau_j$ refers to the tie-breaking rule used for resource $j$. (Note that this ensures that $\lexexts$ is a strict order.) 
Then $\succ_j$ is simply defined as $\lexexts_j$ over the set of all edges in $\agentSet \times \set{j}$. 

That is, $\succ_j$ orders all edges by decreasing offer values $\theta_{ij}$ (breaking ties according to $\tau_j$). In particular, resource $j$ prefers agents with higher offer values. 
\end{description}

Note that by construction this order only depends on $\values$, $\hat{M}$ and $\gamma$; in particular, $\succ_j$ is independent of the compatibility declarations as required. 

\sloppy
Based on the above definitions, running our \boost mechanism on $(\mathcal{I}_{\bmp^+}, \gamma)$ as input is equivalent to running the agent-proposing \defacc on the preference system  $(G[\agentSet\times\taskSet], (\succ_i)_{i \in \agentSet}, (\succ_j)_{j \in \taskSet})$.
Thus, it follows from Theorem~\ref{thm:GS} that \boost is group-strategyproof. 
\fussy
\end{proof}

\subsection{\boost: Consistency, Robustness and Approximation} \label{subsec:boost:apx}

In this section, we prove the bounds on the consistency, robustness and approximation guarantee of \boost. The following lemma will turn out to be useful.

\begin{lemma}
    \label{lem:missing-lemma}
    Let $\gamma \geq 1$. Let $\mathcal{I}_{\bmp^+}=(G[D], \values, \hat{M})$ be an instance of $\bmp^+$ and let $M$ be the matching returned by $\boost(\mathcal{I}_{\bmp^+}, \gamma)$.
    Then 
    $
        2  v(M) + (\gamma-1)  v(M \cap \hat{M}) \geq v(M_D^*).
    $    
\end{lemma}

\begin{proof}
Let $\mstar{\declaredEdges}$ be an optimal matching. We prove that the value of each edge in $\mstar{\declaredEdges}$ can be covered by the value of an edge in $M$ output by \boost$(\mathcal{I}_{\bmp^+},\gamma)$. More precisely, we define a mapping $g: \mstar{\declaredEdges} \rightarrow M$ together with some scalars $(\alpha_e)_{e \in \mstar{\declaredEdges}}$ such that for each edge $e \in \mstar{\declaredEdges}$ it holds that
$
v_e \le \alpha_e \cdot v_{g(e)} 
$
with $\alpha_e \ge 1$. We also say that $e$ is \emph{$\alpha_e$-covered} by edge $g(e) \in M$.

Let $e = (i,j) \in \mstar{\declaredEdges}$. If $e \in M$, we define $g(e) = e$ and $\alpha_e = 1$. Suppose $e = (i,j) \notin M$. We distinguish the following cases: 
\begin{enumerate}

\item $\exists k \in M(i)$ with $v_{ik} \ge v_{ij}$. We define $g(e) = (i,k)$ and $\alpha_e = 1$.

\item $\exists k \in M(i)$ with $v_{ik} < v_{ij}$. In this case, $i$ first proposes to $j$ and only later to $k$. In particular, $j$ must have rejected the offer of $i$ at some stage. Thus, there must be an agent $\ell$ with $(\ell,j) \in M$ whose offer is larger than the one of $i$. 
We define $g(e) = (\ell,j)$ in this case.
The definition of $\alpha_e$ depends on whether $(i,j)$ or $(\ell,j)$ are part of the predicted matching $\hat{M}$: 

(a) $(i,j) \notin \hat{M}$ and $(\ell,j) \notin \hat{M}$: 
We have $v_{\ell j} > v_{ij}$ and define $\alpha_e = 1$.

(b) $(i,j) \notin \hat{M}$ and $(\ell,j) \in \hat{M}$: 
We have $\gamma v_{\ell j} > v_{ij}$ and define $\alpha_e = \gamma$.

(c) $(i,j) \in \hat{M}$ and $(\ell ,j) \notin \hat{M}$: 
We have $v_{\ell j} > \gamma v_{ij}$ and define $\alpha_e = 1$. 

\item $\not\exists k \in M(i)$. Note that $i$ proposed to $j$ at some stage but was rejected (immediately or subsequently) and remained unassigned after all. Similarly to the previous case, this implies that there exists some agent $\ell$ with $(\ell,j) \in M$ whose offer is larger than the one of $i$.
We can follow the same line of arguments as in the previous case. We define $g(e) = (\ell,j)$ and $\alpha_e$ as in the case distinction above.
\end{enumerate}

Note that the mapping $g$ defined above maps each edge $e \in \mstar{\declaredEdges}$ either to itself, i.e., $g(e) = e \in M$, or to an edge $f = g(e) \in M$ that is adjacent to $e$. Also, because $\mstar{\declaredEdges}$ is a matching, there are at most two edges in $\mstar{\declaredEdges}$ which are adjacent to an edge $f$ in $M$. Said differently, each edge $f \in M$ covers at most two edges in $\mstar{\declaredEdges}$.
Moreover, if edge $f = (\ell,j) = g(e) \in M$ $\gamma$-covers an edge $e = (i,j)\in \mstar{\declaredEdges}$ (i.e., Cases (2b) and (3b) above), $f$ and $e$ must share a common resource $j$, and $f$ must be part of $\hat{M}$, i.e. $f \in \hat{M} \cap M$; in particular, the other edge in $\mstar{\declaredEdges}$ that is mapped to $f$ (if any) must be $1$-covered by $f$. Using the above observations, we can now prove the claim:
\[ 
v(\mstar{\declaredEdges}) = \sum_{e \in \mstar{\declaredEdges}} v_{e} \le \sum_{e \in \mstar{\declaredEdges}} \alpha_e v_{g(e)} 
\le  \sum_{f \in M \setminus \hat{M}} 2 v_f + \sum_{f \in M \cap \hat{M}} (1+\gamma) v_f 
= 2 v(M) + (\gamma - 1) v(\hat{M} \cap M). 
\qedhere
\]
\end{proof}

We can now prove that \boost is $(1+\gamma)$-robust. The proof follows easily from Lemma~\ref{lem:missing-lemma}. 

\begin{proof}[Proof of Theorem~\ref{thm:bpm} (robustness)]
Let $\gamma \ge 1$.
Let $\mathcal{I}_{\bmp^+} = (G[\declaredEdges], \values, \hat{M})$ be an instance of $\bmp^+$ and let $M$ be the matching returned by $\boost(\mathcal{I}_{\bmp^+}, \gamma)$. 
Further, let $\mstar{\declaredEdges}$ be an optimal matching. By Lemma~\ref{lem:missing-lemma}, we have
\[ 
v(\mstar{\declaredEdges}) \le
2  v(M) + (\gamma-1)  v(M \cap \hat{M})
\le
2  v(M) + (\gamma-1)  v(M)
\le (1+\gamma) v(M).
\qedhere
\]
\end{proof}

\sloppy
The next lemma shows that the matching computed by \boost is a $(1+\nicefrac{1}{\gamma})$-approximation of the predicted matching $\hat{M} \cap \declaredEdges$ in $G[\declaredEdges]$. It will turn out to be useful when proving the approximation guarantee of \boost. 
\fussy

\begin{lemma}\label{lem:bpm-approx}
Let $\gamma \ge 1$. Let $\mathcal{I}_{\bmp^+}=(G[\declaredEdges], \values, \hat{M})$ be an instance of $\bmp^+$ and let $M$ be the matching returned by $\boost(\mathcal{I}_{\bmp^+}, \gamma)$.
Then $(1+\nicefrac{1}{\gamma}) v(M) \geq v(\hat{M} \cap \declaredEdges)$.
\end{lemma}

\begin{proof}
We prove that the value of each edge in $\hat{M} \cap \declaredEdges$ can be covered by the value of an edge in the matching $M$ output by \boost$(\mathcal{I}_{\bmp^+},\gamma)$. 
More precisely, we define a mapping $g: \hat{M} \cap \declaredEdges \rightarrow M$ together with some scalars $(\alpha_e)_{e \in \hat{M} \cap \declaredEdges}$ such that for each edge $e \in \hat{M} \cap \declaredEdges$ it holds that $\alpha_e \cdot v_e \le v_{g(e)}$ with $\alpha_e \ge 1$. We also say that $e$ is \emph{$(1/\alpha_e)$-covered} by edge $g(e) \in M$.

Let $e = (i,j) \in \hat{M} \cap \declaredEdges$. If $e \in M$, we define $g(e) = e$ and $\alpha_e = 1$. Suppose $e = (i,j) \notin M$. We distinguish the following cases: 
\begin{enumerate}
\item $\exists k \in M(i)$ with $v_{ik} \ge v_{ij}$. We define $g(e) = (i,k)$ and $\alpha_e = 1$.

\item $\exists k \in M(i)$ with $v_{ik} < v_{ij}$. Note that $i$ first proposes to $j$ and only later to $k$. In particular, $j$ must have rejected the offer of $i$ at some stage. Thus, there must be an agent $\ell$ with $(\ell,j) \in M$ whose offer is larger than the one of $i$. Recall that $(i,j) \in \hat{M}\cap \declaredEdges$ and thus $i$ made a boosted offer $\gamma v_{ij}$ to $j$. On the other hand, $(\ell,j) \notin \hat{M} \cap \declaredEdges$ and thus $\ell$ offers $v_{\ell j}$ to $j$. We conclude that $v_{\ell j} > \gamma v_{ij}$. We define $g(e) = (\ell ,j)$ and $\alpha_e = \gamma$.

\item $\not\exists k \in M(i)$. Note that $i$ proposed to $j$ at some stage but was rejected (immediately or subsequently) and remained unassigned after all. Similarly to the previous case, this implies that there exists some agent $\ell$ with $(\ell,j) \in M$ and $v_{\ell j} > \gamma v_{ij}$. We define $g(e) = (\ell ,j)$ and $\alpha_e = \gamma$.
\end{enumerate}

Note that the mapping $g$ defined above maps each edge $e \in \hat{M}\cap \declaredEdges$ either to itself, i.e., $g(e) = e \in M$, or to an edge $f = g(e) \in M$ that is adjacent to $e$. Also, because $\hat{M} \cap \declaredEdges$ is a matching, there are at most two edges in $\hat{M} \cap \declaredEdges$ which are adjacent to an edge $f$ in $M$. Said differently, each edge $f \in M$ covers at most two edges in $\hat{M} \cap \declaredEdges$.
Moreover, if edge $f = (\ell,j) = g(e) \in M$ $({1}/{\gamma})$-covers an edge $e = (i,j)\in \hat{M} \cap \declaredEdges$ (i.e., Cases (2) and (3) above), $f$ and $e$ must share a common resource $j$; in particular, the other edge in $\hat{M} \cap \declaredEdges$ that is mapped to $f$ (if any) must be $1$-covered by $f$. Using the above observations, we can now prove the claim: 
\[ 
v(\hat{M} \cap \declaredEdges) = \sum_{e \in \hat{M} \cap \declaredEdges} v_{e} \le \sum_{e \in \hat{M} \cap \declaredEdges} v_{g(e)} / \alpha_{e} \le \Big (1 + \frac{1}{\gamma} \Big) \sum_{f \in M} v_f = \Big (1 + \frac{1}{\gamma} \Big) v(M). \qedhere
\]
\end{proof}

We can now complete the proof of Theorem~\ref{thm:bpm}.

\begin{proof}[Proof of Theorem~\ref{thm:bpm} (Approximation)]
\sloppy
Let $\gamma \ge 1$ be fixed arbitrarily. 
Consider an instance $\mathcal{I}_{\bmp^+}=(G[\declaredEdges], \values, \hat{M})$ of $\bmp^+$ with prediction error $\eta(\mathcal{I}_{\bmp^+}) \le \hat{\eta}$.
Let $M$ be the matching returned by $\boost(\mathcal{I}_{\bmp^+}, \gamma)$.
Note that by Lemma~\ref{lem:bpm-approx} we have $(1+\nicefrac{1}{\gamma}) v(M) \ge v(\hat{M} \cap \declaredEdges)$. Now, using the Lifting Lemma (Lemma~\ref{lem:lifting}) we conclude that \boost\ is $(1+\nicefrac{1}{\gamma})/(1-\hat{\eta})$-approximate. 
Further, the robustness guarantee of $(1+\gamma)$ holds independently of the prediction error $\hat{\eta}$. The claimed bound on the approximation guarantee $g(\hat{\eta}, \gamma)$ now follows by combining these two bounds. 
\end{proof}

\subsection{Extensions of \boost} \label{rem:boost-extensions}

\boost is rather versatile in the sense that it can be adapted to handle more general settings while retaining its group-strategyproofness property. We summarize a few extensions below.

\begin{enumerate}
    
\item[(E1)] \boost\ can also be run with a many-to-one assignment as input prediction and remain group-strategyproof. We exploit this in Section~\ref{sec:randomized}. In fact, the only change is that the offer function in \eqref{eq:offer} is defined with respect to a predicted many-to-one assignment $\hat{M}$. The proof of group-strategyproofness in Theorem~\ref{thm:bpm} continues to hold without change. 

\item[(E2)] \boost\ can also handle many-to-one assignments instead of matchings. Also here, the offer function in \eqref{eq:offer} is defined with respect to a predicted many-to-one assignment $\hat{M}$. Further, each resource $j$ now accepts the at most $\jcap_j$ highest offers among the set of proposing agents and rejects the remaining ones.\footnote{Recall that in the many-to-one setting each agent has unit size and all resources have integer capacities.} The resulting adaptation of \boost remains group-strategyproof. An easy way to see this, is by realizing that this adapted mechanism mimics \boost on the instance obtained from the reduction described next. 

\item[(E3)] \boost\ can also be used to handle instances of $\text{RSGAP}^+$ by a simple reduction to $\bmp^+$. 
Recall that for an instance $\mathcal{I}_{\text{RSGAP}^+}=(G[\declaredEdges], \values, \bm{s}, \vec{\jcap}, \hat{M})$ of $\text{RSGAP}^+$ it holds that all agents have the same size $s_j$ with respect to a resource $j \in \taskSet$, i.e., $s_{ij} = s_j$ for all $i \in \agentSet$. It is not hard to see that we can reduce $\mathcal{I}_{\text{RSGAP}^+}$ to an equivalent instance $\mathcal{I}_{\bmp^+}$ of $\bmp^+$: For each resource $j \in \taskSet$, we introduce $m_j = \lfloor \jcap_j/s_j \rfloor$ copies $j_1, \dots, j_{m_j}$. Each copy $j_{\ell}$ with $\ell \in [m_j]$ inherits the set of edges $\declaredEdges_j$ incident to $j$ and the value matrix $\values'$ is defined accordingly. 
Similarly, for each resource $j \in R$ with $k_j \le m_j$ predicted edges $(i_{1},j), (i_{2}, j), \dots, (i_{k_j},j) \in \hat{M}(j)$, edges $(i_{1},j_1), (i_{2}, j_2), \dots, (i_{k_j},j_{k_j})$ are added to the predicted matching $\hat{M}'$.
Now, there is a natural correspondence between the compatibility declarations of an agent in $\mathcal{I}_{\text{RSGAP}^+}$ and $\mathcal{I}_{\bmp^+}$. Similarly, each many-to-one assignment in $\mathcal{I}_{\text{RSGAP}^+}$ corresponds to a matching in $\mathcal{I}_{\bmp}^+$. 
It is not hard to prove that the two instances are equivalent. 
\end{enumerate}

The latter observation leads to the following corollary. 

\begin{corollary} \label{corollary:rsigap}
Fix some error parameter $\hat{\eta} \in [0,1]$. Consider the class of instances of $\text{RSGAP}^+$ in the private graph model with prediction error at most $\hat{\eta}$. Then, for any confidence parameter $\gamma \ge 1$, \boost\ is group-strategyproof and has an approximation guarantee of
\[
g(\hat{\eta}, \gamma) = 
\begin{cases}
    \frac{1+\gamma}{\gamma(1-\hat{\eta})} & \text{if $\hat{\eta} \le 1 - \nicefrac{1}{\gamma}$,} \\
    1 + \gamma & \text{otherwise.}
\end{cases}
\]
\end{corollary}

Note that the above implies that \boost\ is group-strategyproof and $2$-approximate for RSGAP without predictions (i.e., by choosing $\gamma=1$). To the best of our knowledge, the current best mechanism for this problem is the randomized, universally strategyproof, $4$-approximate mechanism by \citet{chen14}.

\section{Beyond Bipartite Matching Via Greedy Mechanisms} \label{sec:greedy}

In this section, we first introduce a generic mechanism design template for $\gap^+$ that provides a unifying building block for several of our mechanisms. After that, we provide a first application of this template and derive a deterministic mechanism for $\emkar^+$. 

\begin{mechanism}[t]
\SetKwInput{Input}{Input}
\SetKwInput{Output}{Output}
\DontPrintSemicolon
\caption{\label{alg:greedy-template}
\greedy$(\mathcal{I}_{\gap^+}, \vec{z})$}
\Input{An instance $\mathcal{I}_{\gap^+}=(G[\declaredEdges], \values, \vec{\sizes}, \vec{C}, \hat{M})$, a ranking function  $\vec{z}: L \times R \mapsto \mathbb{R}^k$ for some $k \in \mathbb{N}$.}
\Output{A feasible assignment for $\mathcal{I}_{\gap^+}$.}
Set $\mathcal{L} = \msort(\declaredEdges, (z(e))_{e \in \declaredEdges})$.
\algcom{sort edges in $\declaredEdges$ lexicographically by $(z(e))_{e \in \declaredEdges}$}
Initialize $M= \emptyset$.\\
\While(\tcp*[f]{process edges in sorted order of $\mathcal{L}$}){$\mathcal{L} \neq \mlist{}$}{
    Let $(i,j)$ be the first edge of $\mathcal{L}$ and remove it.
    \algcom{remove first edge $(i,j)$ from $\mathcal{L}$}
    \If(\tcp*[f]{$i$ can be added to $M(j)$ without exceeding capacity}){$\sum_{t \in M(j) \cup \{i\}}s_{t j} \leq C_j$}
    {
        Set $M(j) = M(j) \cup \{i\}$. \algcom{add $i$ to $M(j)$} 
        Remove all edges of agent $i$ from $\mathcal{L}$. \algcom{update $\mathcal{L}$} 
    }
}
\Return $M$
\end{mechanism}

\subsection{A Template Of Greedy Mechanisms} \label{sec:GreedyTemplate}

At high level, the greedy mechanism template behaves as follows: The mechanism first orders all declared edges according to some specific ranking (which is given as part of the input). According to this order, the mechanism then greedily adds as many edges as possible (while maintaining feasibility) to construct an assignment. We refer to this mechanism as \greedy; see Mechanism~\ref{alg:greedy-template}.

\greedy\ receives as input an instance $\mathcal{I}_{\gap^+}=(G[\declaredEdges], \vec{v}, \vec{\sizes}, \vec{C}, \hat{M})$ of $\gap^+$ and a \emph{ranking function} $\vec{z}: L \times R \mapsto \mathbb{R}^k$ for some $k \in \mathbb{N}$. It then uses the $\msort$ operator (as defined in \eqref{eq:sort}) to sort the set of declared edges $\declaredEdges$ in lexicographic decreasing order according to their values $(z_1, \dots, z_k)$. As a result, the list $\mathcal{L} = \mlist{e_{\pi_1}, \dots, e_{\pi_{|\declaredEdges|}}}$ output by $\msort$ satisfies $\vec{z}(e_{\pi_i}) \lex \vec{z}(e_{\pi_{j}})$ for all $i < j$. \greedy\ then processes the edges in this order by always removing the first element $(i,j)$ from $\mathcal{L}$, and greedily assigns agent $i$ to resource $j$ whenever this maintains the feasibility of the constructed assignment $M$. If $i$ can be  assigned to $j$, the assignment $M$ is updated accordingly and all edges of $i$ are removed from the list $\mathcal{L}$. \greedy\ terminates if there are no more edges in $\mathcal{L}$.

It is important to realize that \greedy\ coupled with an arbitrary ranking function $\vec{z}$ may not result in a strategyproof mechanism for $\gap^+$ in general. However, we show that mechanisms derived through this template allow us to obtain meaningful results for several $\gap^+$ variants studied in this paper. Definition \ref{def:truth-inducing} captures a sufficient condition for group-strategyproofness for a ranking function $\vec{z}$ for $\greedy$, as we show in Theorem \ref{thm:greedy-gsp-sufficient}. 

\begin{definition} \label{def:truth-inducing} 
Consider some class of instances of $\gap^+$ in the private graph model. We say that a ranking function $\vec{z}$ is \emph{truth-inducing} if for every instance 
$\mathcal{I_{\gap^+}}$ of this class it holds that: 
    \begin{enumerate}
        \item The extended lexicographic order of $L \times R$ with respect to $\vec{z}$ is strict and total.
        \item For every agent $i \in L$, and every $e=(i,j)$, $e'=(i,j') \in L \times R$ with $\vec{z}(e) \lex \vec{z}(e')$, it holds that $v_{e} \geq v_{e'}$.
    \end{enumerate}
\end{definition}

\begin{restatable}{theorem}{fivetwo}\label{thm:greedy-gsp-sufficient}
Consider some class of instances of $\gap^+$ in the private graph model and let $\vec{z}$ be a ranking function that is truth-inducing with respect to this class. Then $\greedy$ coupled with $\vec{z}$ is a group-strategyproof mechanism.    
\end{restatable}

\begin{proof}
\sloppy
First, note that as $\vec{z}$ is truth-inducing, there is a strict extended lexicographic order of $L \times R$. \greedy\ always considers declared edges according to this order, so for two edges $e$ and $f$ both in $\declaredEdges$ and $\declaredEdges'$, either $e$ is always considered before $f$ or vice verse. Secondly, note that by construction, \greedy\ never unassigns an edge during the entire execution. And lastly, consider an instance $\mathcal{I_{\gap^+}}=(G[\declaredEdges], \vec{v}, \vec{s}, \vec{C}, \hat{M})$ with truth-inducing ranking function $\vec{z}$, $e \in \declaredEdges$ and an instance $\mathcal{I'_{\gap^+}}=(G[\declaredEdges \backslash \{e\}], \vec{v}, \vec{s}, \vec{C}, \hat{M})$. Let $M=\greedy(\mathcal{I_{\gap^+}},\vec{z})$ and $M'=\greedy(\mathcal{I'_{\gap^+}},\vec{z})$. Note that by construction, if $e \notin M$, then $M = M'$. In other words, removing a single unassigned edge from the instance has no influence on the execution of \greedy.

\fussy
Now, consider an instance $\mathcal{I_{\gap^+}}=(G[\declaredEdges], \vec{v}, \vec{s}, \vec{C}, \hat{M})$ with private graph $G[\privateEdges]$. Let $\vec{z}$ be a ranking function that is truth-inducing for $\mathcal{I_{\gap^+}}$. Consider an agent $i \in L$ with $\declaredEdges_i = \privateEdges_i$. 
Let $\declaredEdges'_{i}$ be a deviation of agent $i$ with $\mathcal{I'_{\gap^+}}$ the corresponding instance of this unilateral deviation of $i$. Let $M=\greedy(\mathcal{I_{\gap^+}},\vec{z})$ and let $M'=\greedy(\mathcal{I'_{\gap^+}},\vec{z})$. 

As $\vec{z}$ is truth-inducing for $\mathcal{I_{\gap^+}}$, declared edges of $i$ are always considered in non-increasing order of value. Therefore, and as \greedy\ never unassigns an edge, agent $i$ cannot benefit from reporting edges in $\declaredEdges'_{i}$ that are not in $\privateEdges_i$. If such an edge $(i,j) \in \declaredEdges'_{i} \backslash \privateEdges_i$ is in $M'$, this edge will never be unassigned leading to a utility of zero for agent $i$. If these edges are unassigned in $M'$, they have no influence of the execution of \greedy. Therefore, consider deviations such that $\declaredEdges'_i \subset \privateEdges_i$. Again, as $\vec{z}$ is truth-inducing for $\mathcal{I_{\gap^+}}$, hiding an edge $e \in \privateEdges_i$ that is not in $M$ has no influence on the the execution of \greedy. So consider the case in which $\exists e=(i,j) \in M$ and $e \notin \declaredEdges'_i$. However, edges $f=(i,k) \in \declaredEdges'_i$, so also $f \in \privateEdges_i$, with $z(f) \lex z(e)$ will not be assigned in $M'$ as nothing has changed in the execution of \greedy\ when these edges are considered. Therefore, the utility of agent $i$ will not strictly increase.

For group-strategyproofness, consider a subset $S \subseteq L$ of agents such that $\forall i \in S$ it holds that $\declaredEdges_i = \privateEdges_i$. Let $\declaredEdges'_{S}$ be a group deviation of $S$ with $\mathcal{I'_{\gap^+}}$ the corresponding instance of this unilateral deviation of $S$. Let $M=\greedy(\mathcal{I_{\gap^+}},\vec{z})$ and let $M'=\greedy(\mathcal{I'_{\gap^+}},\vec{z})$. By the same reasoning as above, an agent $i \in S$ cannot benefit from reporting edges in $\declaredEdges'_{i}$ that are not in $\privateEdges_i$. Therefore, consider deviations such that $\declaredEdges'_i \subset \privateEdges_i$ for all $i \in S$. Again, if for all $i \in S$ it holds that all $e \in \privateEdges_i \backslash \declaredEdges'_i$ are not in $M$, this has no influence on the the execution of \greedy. So consider the case in which $\exists i \in S$, $\exists e=(i,j) \in M$ and $e \notin \declaredEdges'_i$. Of all such agents, let $i$ be the agent such that $e=(i,j) \in M$ was the first of these edges to be assigned. Again, edges $f=(i,k) \in \declaredEdges'_i$, so also $f \in \privateEdges_i$, with $z(f) \lex z(e)$ will not be assigned in $M'$ as nothing has changed in the execution of \greedy\ when these edges are considered. Therefore, the utility of agent $i$ will not strictly increase, and no such $i$ will join such a deviation. 
\end{proof}

Finally, we stress that, given an instance $\mathcal{I}_{\gap^+}=(G[D], \vec{v}, \vec{s}, \vec{C}, \hat{M})$, $\greedy$ is not necessarily dependent on the predicted assignment $\hat{M}$; it can handle a non-augmented instance of $\gap$ as well. However, the flexibility is in place for the accompanying ranking function $\vec{z}$ to use $\hat{M}$ in a beneficial manner for the underlying optimization problem. In the following section we present an implementation of this concept.

\subsection{Restricted Multiple Knapsack} \label{subsec:mkar-deterministic}

In this section, we focus on devising a deterministic, group-strategyproof mechanism for $\emkar^+$.\footnote{We readdress the general problem in Section \ref{subsec:sigap-vcgap} through a more general variant, by devising a randomized mechanism.} We write $\mathcal{I}_{\emkar^+}=(G[\declaredEdges], (v_i = s_i)_{i \in \agentSet}, \vec{\jcap}, \hat{M})$ to denote an instance of $\emkar^+$. Interestingly, $\emkar$ is strongly NP-hard as shown by \citet{dawande00}. 
For this class of instances, we show that coupling our \greedy\ mechanism with a carefully chosen ranking function $\vec{z}$ gives a deterministic group-strategyproof mechanism achieving constant approximation guarantees (for every fixed $\gamma \ge 1$). Our ranking function $\vec{z}$ combines our boosted offer notion $\theta_{ij}(\gamma, \hat{M})$ with a specific tie-breaking rule to favor edges in the predicted assignment $\hat{M}$. 
This allows us to derive improved approximation guarantees if the prediction error is small, while at the same time retaining bounded robustness if the prediction is erroneous. 

Let $\gamma \ge 1$ be fixed arbitrarily. We define the ranking function $\vec{z}: L \times R \mapsto \mathbb{R}^4$ as follows: 
Let the boosted offer $\theta_{ij}(\gamma, \hat{M})$ be defined as in \eqref{eq:offer}. Also, let $\mathds{1}_{(i,j) \in \hat{M}}$ be the indicator function which is $1$ if and only if $(i,j) \in \hat{M}$. Then, for each $(i,j) \in L \times R$, we define
\begin{equation}
    \label{eq:mkar-zeta}
    \vec{z}((i,j)) := \Big(\theta_{ij}(\gamma, \hat{M}),\  \mathds{1}_{(i,j) \in \hat{M}},\ -i,\ -j \Big).
\end{equation}

The intuition behind our ranking function is to rank the edges in $\declaredEdges$ by their $\gamma$-boosted value. Recall that for \emkar we have for each agent $i \in \agentSet$, $v_i = v_{ij}$ for all $j \in \taskSet$. In particular, for $\gamma > 1$, the first-order criterion $\theta_{ij}(\gamma, \hat{M})$ ensures that the predicted edge of agent $i$ in $\declaredEdges_i$ is ordered before the non-predicted ones. In fact, crucially, the second-order criterion ensures that this property also holds for $\gamma = 1$. Put differently, whenever $\theta_e(\gamma, \hat{M})=\theta_{e'}(\gamma, \hat{M})$ we make sure that priority is given to edges in $\hat{M} \cap D$. Remarkably, the preference we give to the predictions in case of ties leads to improved approximation guarantees even for $\gamma=1$, if the prediction error $\eta$ is small, i.e., for $\eta < \nicefrac{1}{3}$. If any ties remain, they are broken in increasing index of first $i$ and then $j$.

We use $\mkarGreedy$ to refer to the mechanism that we derive from \greedy\ with the ranking function $\vec{z}$ as defined in \eqref{eq:mkar-zeta}. 

\begin{restatable}{theorem}{mkardeterministic}
    \label{thm:mkar-deterministic}
    Fix some error parameter $\hat{\eta}\in [0,1]$. Consider the class of instances of $\emkar^+$ in the private graph model and prediction error at most $\hat{\eta}$. Then, for every confidence parameter $\gamma \geq 1$, \mkarGreedy is group-strategyproof and has an approximation of guarantee
        \begin{equation}
        \label{eq:apx-mkar-det}
        g(\hat{\eta}, \gamma) = \begin{cases*}
                                    \frac{1+\gamma}{\gamma(1-\hat{\eta})} & if $\hat{\eta}\leq 1-\frac{\gamma+1}{\gamma(\gamma+2)},$\\
                                    2+\gamma & otherwise.
          \end{cases*}
        \end{equation}
    In particular, $\mkarGreedy$ is $(1+\nicefrac{1}{\gamma})$-consistent and $(2+\gamma)$-robust.
\end{restatable}

Note that for $\gamma = 1$ our result implies a $3$-approximate, group-strategyproof mechanism for $\emkar$. To the best of our knowledge, no deterministic strategyproof mechanism was known for this problem prior to our work.

We devote the remainder of this section for the proof of Theorem \ref{thm:mkar-deterministic}. We start by showing group-strategyproofness.

\begin{proof}[Proof of Theorem~\ref{thm:mkar-deterministic} (Group-Strategyproofness)]
    We show that the ranking function $\vec{z}$ is truth-inducing (Definition \ref{def:truth-inducing}) for the class of $\emkar^+$ instances. Let $\gamma \geq 1$. Let $\mathcal{I}_{\emkar^+}=(G[\declaredEdges],(v_i = \isize_i)_{i \in \agentSet}, \vec{C}, \hat{M})$ be any instance of $\emkar^+$.  Note that the extended lexicographic order with respect to $\vec{z}$ as defined in \eqref{eq:mkar-zeta} is strict and total. Also, for each agent $i \in \agentSet$ and all $e = (i,j),\ e' = (i,j') \in \agentSet \times \taskSet$, we have that $\theta_{ij}(\gamma, \hat{M}) \ge \theta_{ij'}(\gamma, \hat{M})$ trivially implies $v_{ij} \ge v_{ij'}$ because $v_{ij} = v_{ij'} = v_i$. Since the ranking function $\vec{z}$ is truth-inducing, the group-strategyproofness of this greedy mechanism follows from Theorem \ref{thm:greedy-gsp-sufficient}.
\end{proof}

Our next objective towards the proof of Theorem \ref{thm:mkar-deterministic} is to show that $\mkarGreedy$ is $(2+\gamma)$-robust. We prove two auxiliary lemmas. 

The first lemma shows that the assignment returned by $\mkarGreedy$ provides a ``utilization guarantee'' for each resource which missed optimal agents.\footnote{This is a generalization of a key idea used by \citet{dawande00} to our environment with predictions. In their work, they show that the assignment of the natural greedy algorithm for this problem guarantees that each resource/knapsack is $\nicefrac{1}{2}$-full. For $\gamma=1$, our mechanism retains this property.}

\begin{lemma}
\label{lemma:MKAR-capacity-ub}
    Let $\gamma \geq 1$. Let $\mathcal{I}_{\emkar^+}=(G[\declaredEdges],(v_i = \isize_i)_{i \in \agentSet}, \vec{C}, \hat{M})$ be an instance of $\emkar^+$ and $M$ be the assignment returned by $\mkarGreedy(\mathcal{I}_{\emkar^+}, \gamma)$. For every resource $j \in R$ for which there exists some agent $i \in M^*_{D}(j)$ with $M(i)=\emptyset$, it holds that $2 v(M(j)) + (\gamma-1)  v(M(j) \cap \hat{M}(j)) \geq C_j$.
\end{lemma}

\begin{proof}
    Let $j \in R$ and $i \in M^*_{D}(j)$ with $M(i)=\emptyset$. Since $i \in M^*_{\declaredEdges}(j)$ and $M(i)=\emptyset$, in the iteration in which \greedy\ considered edge $(i,j)$ the if-statement in Line 5 failed and $(i,j)$ was not added to $M$. Let $T \subseteq \agentSet$ denote the set of agents that were assigned to $j$ before this iteration. Note that $T \neq \emptyset$, since $i \in M^{*}_{\declaredEdges}(j)$ so $v_i \leq C_j$. Thus, we must have 
    \begin{equation}
        \label{eq:mkar-constraint-violation}
        C_j < v(T) + v_i \leq v(M(j)) + v_{i}.
    \end{equation}
    The second inequality holds because \greedy\ only adds edges to $M$, and thus $T \subseteq M(j)$. 
     
     Fix an agent $k \in T$. Let $\theta_{ij}:=\theta_{ij}(\gamma, \hat{M})$ and $\theta_{kj}:=\theta_{kj}(\gamma, \hat{M})$. Because $\greedy$ considered edge $(k,j)$ before edge $(i,j)$, by the first-order ranking of $\mkarGreedy$ in \eqref{eq:mkar-zeta} we have $\theta_{ij} \leq \theta_{kj}$. We prove the following simple claim to obtain a useful upper bound on $\theta_{kj}$.

    \begin{claim}
        \label{claim:mkar-theta-ub}
        For all $t \in M(j)$, it holds that $\theta_{tj} \leq v(M(j)) + (\gamma-1) \cdot v(M(j) \cap \hat{M}(j))$.
    \end{claim}
    
    \begin{proof}
    There are two cases to consider for agent $t \in M(j)$.
    \begin{enumerate}
        \item $t \in \hat{M}(j)$. Here, $\theta_{tj} = \gamma v_t \leq \gamma \cdot v(M(j) \cap \hat{M}(j)) \leq v(M(j)) + (\gamma-1) \cdot v(M(j) \cap \hat{M}(j))$. The first equality follows by the definition of $\theta_{tj}$ in \eqref{eq:offer}. Then, the first inequality is true since $t \in M(j) \cap \hat{M}(j)$.
        \item $t \not\in \hat{M}(j)$. Similarly, $\theta_{tj} = v_t \leq v(M(j)) \leq v(M(j)) + (\gamma-1) \cdot v(M(j) \cap \hat{M}(j))$, where the last inequality is true since $\gamma \geq 1$. \qedhere
    \end{enumerate}
    \end{proof}

    \noindent
    We now expand \eqref{eq:mkar-constraint-violation} as follows:
    \begin{equation*}
        C_j < v(M(j)) + v_{i} \leq v(M(j)) + \theta_{ij} \leq v(M(j)) + \theta_{kj} \leq  2 \cdot v(M(j)) + (\gamma-1) \cdot v(M(j) \cap \hat{M}(j)).
    \end{equation*}
    The second inequality holds by \eqref{eq:offer} and the third inequality holds because $\theta_{ij} \leq \theta_{kj}$ as argued above. Finally, the last inequality follows by applying Claim \ref{claim:mkar-theta-ub} for agent $k$. This concludes the proof.
\end{proof}

The next lemma will be useful to prove the $(2+\gamma)$-robustness guarantee of $\mkarGreedy$. This lemma will also be useful in Section~\ref{sec:randomized} when we devise our randomized mechanisms. 

\begin{lemma}
    \label{lemma:MKAR-prerobustness-lemma}
    Let $\gamma \geq 1$. Let $\mathcal{I}_{\emkar^+}=(G[\declaredEdges],(v_i = \isize_i)_{i \in \agentSet}, \vec{C}, \hat{M})$ be an instance of $\emkar^+$ and $M$ be the assignment returned by $\mkarGreedy(\mathcal{I}_{\emkar^+}, \gamma)$. Then, $3 v(M) + (\gamma -1)  v(M \cap \hat{M}) \geq v(M^*_D)$.
\end{lemma}

\begin{proof}
We use $M^*:=M^*_{\declaredEdges}$ for brevity. Let $S=\{i \in L\mid M(i)= \emptyset \land M^*(i) \neq \emptyset \}$ and let $T=\{j \in R \mid S \cap M^*(j) \neq \emptyset\}$. We have:
    \begin{align}
        \sum_{i \in S} v_i &= \sum_{j \in T}v(S \cap M^*(j))
             \leq \sum_{j \in T}C_j
             \leq \sum_{j \in T}\big( 2  v(M(j)) + (\gamma-1)  v(M(j) \cap \hat{M}(j) \big) \notag \\
             &= 2  \sum_{j \in T} v(M(j)) +(\gamma-1)  \sum_{j \in T}v(M(j) \cap \hat{M}(j)) 
             \leq 2  v(M) + (\gamma-1)  v(M \cap M). \label{eq:mkar-vs-bound}
    \end{align}
    Here, the first equality follows from the definitions of $S$ and $T$. The first inequality is due to the feasibility of the assignment $M^*$. The second inequality holds by applying Lemma \ref{lemma:MKAR-capacity-ub} for each resource $j \in T$. (Note that the preconditions of Lemma \ref{lemma:MKAR-capacity-ub} are satisfied for each resource $j \in T$: under the assignment $M^*$ there exists at least one agent $i \in M^*(j)$ with $M(i)=\emptyset$; in fact, every such agent must be in $S$). Finally, the last inequality holds since $T \subseteq R$ and $\gamma \geq 1$.
    
    Furthermore, the definition of $S \subseteq L$ implies that
    \begin{equation}
        \label{eq:mkar-opt-bound}
        v(M^*) 
        = \sum_{i \in L \setminus S}v_{iM^*(i)} + \sum_{i \in S}v_{iM^*(i)}
        = \sum_{\substack{i \in L: \\ M(i) \neq \emptyset}}v_{iM^*(i)} + \sum_{i \in S} v_i \leq v(M)+ \sum_{i \in S} v_i.
    \end{equation}
    By summing \eqref{eq:mkar-vs-bound} and \eqref{eq:mkar-opt-bound}, the lemma follows. \qedhere
\end{proof}

The proof that \mkarGreedy\ is $(2+\gamma)$-robust now follows easily. 

\begin{proof}[Proof of Theorem~\ref{thm:mkar-deterministic} (Robustness)]
Let $\gamma \ge 1$ be fixed.
Let $\mathcal{I}_{\emkar^+} = (G[\declaredEdges],(v_i = \isize_i)_{i \in \agentSet}, \vec{C}, \hat{M})$ be an instance of $\emkar^+$ and let $M$ be the assignment returned by $\mkarGreedy(\mathcal{I}_{\emkar^+}, \gamma)$. 
Further, let $\mstar{\declaredEdges}$ be an optimal matching. By Lemma~\ref{lemma:MKAR-prerobustness-lemma}, we have
\[ 
v(\mstar{\declaredEdges}) \leq 3 \cdot v(M) + (\gamma -1) \cdot v(M \cap \hat{M})
\leq 3 \cdot v(M) + (\gamma -1) \cdot v(M). 
\qedhere
\]
\end{proof}

\sloppy
\medskip
In the remainder of this section, we establish the approximation guarantee of $\mkarGreedy$ as stated in Theorem \ref{thm:mkar-deterministic}. We first show that whenever $\mkarGreedy$ makes a ``mistake'', meaning that for an edge $(i,j) \in \hat{M}$ it chooses not to assign agent $i$ to their predicted resource $j$, then resource $j$ is already $\nicefrac{\gamma}{\gamma+1}$-utilized.
\fussy

\begin{lemma}
    \label{lemma:mkar-s2-bound}
    Let $\gamma \geq 1$. Let $\mathcal{I}_{\emkar^+}=(G[\declaredEdges],(v_i = \isize_i)_{i \in \agentSet}, \vec{C}, \hat{M})$ be an instance of $\emkar^+$ and let $M$ be the assignment returned by $\mkarGreedy(\mathcal{I}_{\emkar^+}, \gamma)$. Every resource $j \in R$ with $((\hat{M}(j) \cap \declaredEdges) \setminus M(j)) \neq \emptyset $ satisfies $(1+\nicefrac{1}{\gamma}) \cdot v(M(j)) \geq C_j$.
\end{lemma}

\begin{proof}
Fix an agent $i \in (\hat{M}(j) \cap D) \setminus M(j)$ and denote by $T \subseteq L$ the set of agents assigned to resource $j$ before $\greedy$ considered $(i,j)$. Fix an agent $k \in T \setminus (\hat{M}(j) \cap D)$.\footnote{Note that $T \setminus (\hat{M}(j) \cap D) \neq \emptyset$. Otherwise, by the feasibility of $\hat{M} \cap D$, we would have that $v(T) +v_i \leq C_j$, a contradiction since $i \not \in M(j)$.} We have:
\begin{equation*}
     C_j < v(T) +v_i \leq v(M(j)) +v_i \leq  v(M(j)) +\frac{v_k}{\gamma} \leq \bigg(1 + \frac{1}{\gamma} \bigg) \cdot v(M(j)).
\end{equation*}
 The first inequality holds since $i \not \in M(j)$ implies that in the iteration in which \greedy\ considered edge $(i,j)$, the if-statement in Line 5 failed and $(i,j)$ was not added to $M$. The second inequality follows because $T \subseteq M(j)$ (\greedy\ only adds edges to $M$). The third inequality follows because $\greedy$ considered $(k,j)$ before $(i,j)$, and by the ranking $\vec{z}$ of $\mkarGreedy$ (as defined in \eqref{eq:mkar-zeta}) it thus holds that $\gamma v_i = \theta_{ij}(\gamma, \hat{M}) < \theta_{kj}(\gamma, \hat{M})= v_k$. 
 Finally, the last inequality is true since $k \in  M(j)$. This concludes the proof. 
\end{proof}

The following lemma will further facilitate our proof. It implies that \greedy\ will never reject a predicted edge for a resource $j \in R$, unless it has already selected an agent not in $\hat{M}(j)$ in a previous iteration.

    \begin{lemma}
        \label{claim:mkar-no-subsets}
        Let $\gamma \geq 1$. Let $\mathcal{I}_{\emkar^+}=(G[\declaredEdges],(v_i = \isize_i)_{i \in \agentSet}, \vec{C}, \hat{M})$ be an instance of $\emkar^+$ and let $M$ be the assignment returned by $\mkarGreedy(\mathcal{I}_{\emkar^+}, \gamma)$. There is no resource $j \in R$ such that $M(j) \subset \hat{M}(j) \cap D$.
    \end{lemma}
    
    \begin{proof}
        Suppose for contradiction that there exists a resource $j \in R$ with $M(j) \subset \hat{M}(j) \cap D$. Fix an agent $i \in (\hat{M}(j) \cap \declaredEdges) \setminus M(j)$. We first argue that $(i,j)$ is the first edge of agent $i$ that is considered by $\greedy$.
        
        Since $i \in \hat{M}(j)$, by the first-order criterion of $\vec{z}$ in \eqref{eq:mkar-zeta}, it holds that $\theta_{ij}(\gamma, \hat{M}) \geq \theta_{i \ell}(\gamma, \hat{M})=v_{i}$ for all edges $(i, \ell) \in \declaredEdges_i \setminus\{(i, j)\}$. Furthermore, in case of ties, i.e., when $\theta_{ij}(\gamma, \hat{M})=v_i$ (which can occur when $\gamma=1$), by the second-order criterion of $\vec{z}$, edge $(i,j)$ still precedes all other edges $(i ,\ell) \in \declaredEdges_i \setminus \{(i,j)\}$. Note that we exploit that $\hat{M}$ is a many-to-one matching here. Thus, according to the order used by \greedy, edge $(i,j)$ is the first edge of agent $i$ that will be considered.
    
        Denote by $T$ the set of agents that were already assigned to resource $j$ by $\greedy$ before $(i,j)$ was considered. We have 
        $$C_j < v(T) + v_i \leq v(M(j)) + v_i \leq v((\hat{M}(j) \cap D) \setminus\{i\}) +v_i = v(\hat{M}(j) \cap D).$$
        Note that the first inequality holds because $(i,j)$ is the first edge of $i$ considered by \greedy, but was not selected ($i \notin M(j)$ by assumption), and thus the if-statement in Line 5 must have failed. The second inequality follows because $T \subseteq M(j)$ (\greedy\ only adds edges to $M$). The third inequality follows because $M(j) \subseteq (\hat{M}(j) \cap D) \setminus \{i\}$ (by assumption).
        
        The claim now follows by recalling that for the special case considered here we have $v_i = \isize_i$ for all $i \in \agentSet$. The above inequality thus means that $\hat{M} \cap \declaredEdges$ is an infeasible assignment, which is a contradiction. 
    \end{proof}

We now have all ingredients together to prove that the assignment computed by \mkarGreedy\ is a $(1+\nicefrac{1}{\gamma})$-approximation with respect to the predicted assignment. 
This is the final building block to complete the proof of the approximation guarantee.

\begin{lemma}
\label{lemma:mkar-approx-without-eta}
Let $\gamma \geq 1$. Let $\mathcal{I}_{\emkar^+}=(G[\declaredEdges],(v_i = \isize_i)_{i \in \agentSet}, \vec{C}, \hat{M})$ be an instance of $\emkar^+$ and let $M$ be the assignment returned by $\mkarGreedy(\mathcal{I}_{\emkar^+}, \gamma)$. Then, $(1+\nicefrac{1}{\gamma}) \cdot v(M) \geq v(\hat{M}\cap D)$.
\end{lemma}

\begin{proof}
     Let $S_{1}=\{j \in R \mid M(j) \supseteq \hat{M}(j) \cap \declaredEdges\}$ and $S_{2}=\{j \in R \mid (\hat{M}(j) \cap \declaredEdges) \setminus M(j) \neq \emptyset  \}$. 
     Note that $S_1$ and $S_2$ partition $\taskSet$, i.e., $\taskSet = S_1 \cup S_2$, because for each resource $j \in \taskSet$, $M(j) \not{\subset} \hat{M}(j) \cap D$ by Lemma~\ref{claim:mkar-no-subsets}. 
          
     By the definition of $S_1$, we obtain
    \begin{equation}
        \label{eq:mkar-s1}
        \sum_{j \in S_1} v(\hat{M}(j) \cap D) \leq \sum_{j \in S_1} v(M(j)) <  \bigg( 1+\frac{1}{\gamma} \bigg) \cdot  \sum_{j \in S_1} v(M(j)).
    \end{equation}
    
    Further, we have
    \begin{equation}
        \label{eq:mkar-s2}
        \sum_{j \in S_2} v(\hat{M}(j) \cap D)\leq \sum_{j \in S_2} C_j  \leq \bigg( 1+\frac{1}{\gamma} \bigg) \cdot  \sum_{j \in S_2} v(M(j)).
    \end{equation}
    Here, the first inequality follows because $v_i = \isize_i$ for all agents $i \in \agentSet$ and from the feasibility of $\hat{M}\cap \declaredEdges$. 
    The second inequality follows by applying Lemma \ref{lemma:mkar-s2-bound} for each resource $j \in S_2$.
    By summing \eqref{eq:mkar-s1} and \eqref{eq:mkar-s2}, the proof follows.
\end{proof}

We can now complete the proof of Theorem~\ref{thm:mkar-deterministic}.

\begin{proof}[Proof of Theorem~\ref{thm:mkar-deterministic} (Approximation)]
Let $\gamma \ge 1$ be fixed arbitrarily. Consider an instance $\mathcal{I}_{\emkar^+}=(G[\declaredEdges],(v_i = \isize_i)_{i \in \agentSet}, \vec{C}, \hat{M})$ of $\emkar^+$ with prediction error $\eta(\mathcal{I}_{\emkar^+}) \le \hat{\eta}$. Let $M$ be the assignment returned by $\mkarGreedy(\mathcal{I}_{\emkar^+}, \gamma)$. Note that by Lemma~\ref{lemma:mkar-approx-without-eta} we have $(1+\nicefrac{1}{\gamma}) v(M) \ge v(\hat{M} \cap \declaredEdges)$. Now, using the Lifting Lemma (Lemma~\ref{lem:lifting}) we conclude that \mkarGreedy\ attains an approximation of $(1+\nicefrac{1}{\gamma})/(1-\hat{\eta})$. Further, the robustness guarantee of $2+\gamma$  holds independently of the prediction error $\hat{\eta}$. The claimed bound on the approximation guarantee $g(\hat{\eta}, \gamma)$ in \eqref{eq:apx-mkar-det} now follows by combining these two bounds. 
\end{proof}

\section{Randomized Mechanisms for GAP With Predictions} \label{sec:randomized}

In this section, we devise randomized mechanisms for variants of the Generalized Assignment Problem with predictions. It is important to note that all randomized mechanisms in this section attain the stronger property of universal group-strategyproofness (rather than just universal strategyproofness).

A common thread among all mechanisms in this section is that they may return the outcome of $\trust$ with some probability. As discussed in Section \ref{sec:lowerBounds}, $\trust$ alone lacks a robustness guarantee. However, we can obtain randomized schemes with improved (expected) robustness guarantees by mixing between $\trust$ and other mechanisms. First, we present our methodology in Section \ref{subsec:randomized-separation} by applying it to $\bmp^+$. Then, in Section~\ref{app-subsec:mkar-randomized} we show that the same idea can be applied to $\emkar^+$. Finally, in Section \ref{subsec:sigap-vcgap} we derive our randomized mechanisms for the more general variants of $\gap$.

\subsection{Improved Robustness via Randomization: A Separation Result for Matching} \label{subsec:randomized-separation}

We demonstrate our idea by applying it to $\bmp^+$. Our mechanism \randomizedBoost\ randomizes over two deterministic mechanism, one with and one without a robustness guarantee. \randomizedBoost\ can be summarized as follows: given an instance $\mathcal{I}_{\bmp^+}$ and a confidence parameter $\gamma \geq 1$, with probability $p$ the mechanism outputs $M_1=\boost(\mathcal{I}_{\bmp^+}, \delta(\gamma))$, with $\delta(\gamma)=\sqrt{2(\gamma+1)}-1$. With probability $1-p$, \randomizedBoost\ outputs the matching $M_2=\trust(\mathcal{I}_{\bmp^+})=\hat{M} \cap D$. As we have discussed in Section \ref{sec:lowerBounds}, using \trust implies confidence in the quality of the prediction $\hat{M}$. In expectation, this trade-off gives \randomizedBoost\ an edge regarding robustness compared to simply running $\boost$ and, at the same time, allows it to retain the approximation of $\boost$.

\begin{mechanism}[t]
\SetKwInput{Input}{Input}
\SetKwInput{Output}{Output}
\DontPrintSemicolon
\caption{\label{alg:randomized-bpm}
$\randomizedBoost(\mathcal{I}_{\bmp^+}, \gamma)$}
\Input{An instance $\mathcal{I}_{\bmp^+} =(G[\declaredEdges], \values, \hat{M})$, confidence parameter $\gamma \geq 1$.}
\Output{A probability distribution over matchings for $\mathcal{I}_{\bmp^+}$.}
Let $\delta(\gamma)=\sqrt{2(\gamma+1)}-1$.  \tcp*{Note that $\delta(\gamma) \geq 1$ for all $\gamma \geq 1$.}
Set $M_1=\boost(\mathcal{I}_{\bmp^+}, \delta(\gamma))$.\\
Set $M_2=\trust(\mathcal{I}_{\bmp^+}, \declaredEdges)$.\\
Set $p=2/(\delta(\gamma)+1)$. \tcp*{Note that $p \in (0,1]$ for all $\gamma \geq 1$.}
\Return $M_1$ with probability $p$ and $M_2$ with probability $1-p$.
\end{mechanism}

\begin{restatable}{theorem}{sixone}\label{thm:bpm-randomized}
Fix some error parameter $\hat{\eta}\in [0,1]$. Consider the class of instances of $\bmp^+$ in the private graph model with prediction error at most $\hat{\eta}$. Then, for every confidence parameter $\gamma \geq 1$, \randomizedBoost\ is universally group-strategyproof and has an expected approximation guarantee of
\[
g(\bar{\eta}, \gamma) = 
\begin{cases}
    \frac{1+\gamma}{\gamma(1-\hat{\eta})} & \text{if $\hat{\eta} \le 1 - \frac{\sqrt{2(\gamma+1)}}{2\gamma}$,} \\[1.5ex]
    \sqrt{2(\gamma+1)} & \text{otherwise.}
\end{cases}
\]
In particular, \randomizedBoost\ is $(1+\nicefrac{1}{\gamma})$-consistent and $\sqrt{2(\gamma+1)}$-robust (both in expectation).
\end{restatable}

\begin{remark}
As argued above (see Corollary \ref{corollary:rsigap}), $\boost$ can be adapted to handle instances of RSGAP$^+$ with the same approximation guarantees as for $\bmp^+$ and remain group-strategyproof. Therefore, Theorem \ref{thm:bpm-randomized} extends to $\text{RSGAP}^+$ as well.
\end{remark}

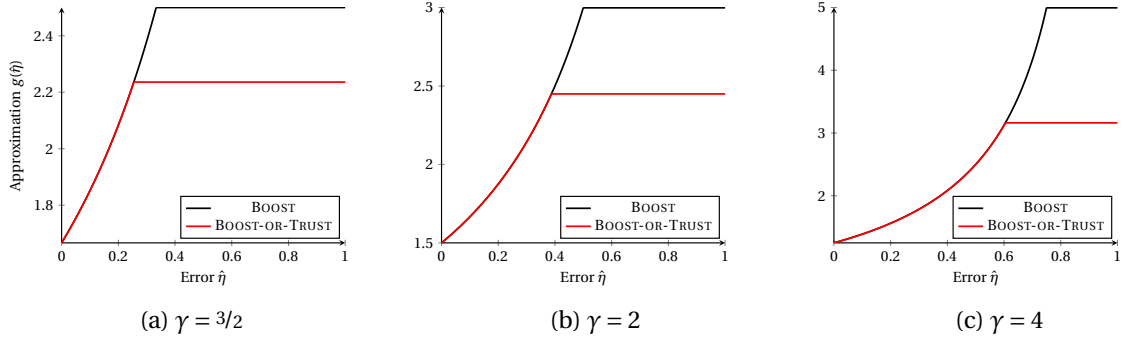
\begin{figure}[t]
\begin{subfigure}[b]{0.3\linewidth}
    \begin{tikzpicture}[scale=0.55]
        \begin{axis}[ thick, axis lines = left, xlabel = \( \text{Error } \hat{\eta}\), ylabel = {\(  \text{Approximation\ } g(\hat{\eta})\)}, legend style={at={(0.7,0.2)}, anchor=north}]
        \addplot [domain=0:(1-1/(3/2)), samples=100, color=black, forget plot, very thick] {(1+1/(3/2))/(1-x)};
        \addplot [domain=(1-1/(3/2)):1, samples=100, color=black, very thick] {1+1.5};
        \addplot [domain=0:(1-(2*(3/2+1))^(1/2)/(3)), samples=100, color=red, very thick] {(1+2/3)/(1-x)}; 
        \addplot [domain=(1-(2*(3/2+1))^(1/2)/(2*3/2)):1, samples=100, color=red, forget plot, very thick] {(2*(5/2))^(1/2)};
        \addlegendentry{$\boost$}
        \addlegendentry{$\randomizedBoost$}
        \end{axis}
    \end{tikzpicture}
    \caption{$\gamma=\nicefrac{3}{2}$}       
\end{subfigure}
\quad
\begin{subfigure}[b]{0.3\linewidth}
    \begin{tikzpicture}[scale=0.55]
        \begin{axis}[ thick, axis lines = left, xlabel = \( \text{Error\ } \hat{\eta}\), legend style={at={(0.7,0.2)}, anchor=north}]
        \addplot [domain=0:((4-1)/(4+2)), samples=100, color=black, very thick] {(1+1/2)/(1-x)}; 
        \addplot [ domain=((4-1)/(4+2)):1, samples=100, color=black, ultra thick, forget plot] {1+2};
        \addplot [domain=0:(1-(6)^(1/2)/4), samples=100, color=red, very thick] {(1+1/2)/(1-x)}; 
        \addplot [ domain=(1-(6)^(1/2)/4:1, samples=100, color=red, forget plot, very thick] {(6)^(1/2)};
        \addlegendentry{$\boost$}
        \addlegendentry{$\randomizedBoost$}
        \end{axis}
    \end{tikzpicture}
    \caption{$\gamma=2$}   
\end{subfigure}
\quad
\begin{subfigure}[b]{0.3\linewidth}
    \begin{tikzpicture}[scale=0.55]
        \begin{axis}[ thick, axis lines = left, xlabel = \( \text{Error } \hat{\eta}\), legend style={at={(0.7,0.2)}, anchor=north}]
        \addplot [domain=0:((16-1)/(16+4)), samples=100, color=black, very thick] {(1+1/4)/(1-x)}; 
        \addplot [ domain=((16-1)/(16+4)):1, samples=100, color=black, ultra thick, forget plot] {1+4};
        \addplot [domain=0:(1-(10)^(1/2)/8), samples=100, color=red, very thick] {(1+1/4)/(1-x)}; 
        \addplot [domain=(1-(10)^(1/2)/8:1, samples=100, color=red, forget plot, very thick] {(10)^(1/2)};
        \addlegendentry{$\boost$}
        \addlegendentry{$\randomizedBoost$}
        \end{axis}
    \end{tikzpicture}
    \caption{$\gamma =4$}
\end{subfigure}
\caption{Approximation guarantees \(g(\hat{\eta})\) of $\boost$ and $\randomizedBoost$ (in expectation) as a function of \(\hat{\eta}\) for various values of \(\gamma\).}
\label{fig:deterministic-vs-randomized-bmp}
\end{figure}

\sloppy
Recall that we concluded from Theorem \ref{th:tradeOffConsisBoundedRobust} and Theorem~\ref{thm:bpm} that $\boost$ attains the optimal consistency-robustness trade-off among all deterministic strategyproof mechanisms. 
However, as shown in Theorem \ref{thm:bpm-randomized}, the (expected) consistency and robustness guarantees achieved by $\randomizedBoost$ is strictly better than that of any deterministic mechanism (see Figure \ref{fig:deterministic-vs-randomized-bmp}). This implies a separation between the two classes of mechanisms for $\bmp^+$ in our environment with predictions.

\fussy

\begin{proof}[Proof of Theorem~\ref{thm:bpm-randomized}]
Let $\gamma \ge 1$ be fixed arbitrarily. Consider an instance $\mathcal{I}_{\bmp^+}=(G[\declaredEdges], \values, \hat{M})$ of $\bmp^+$ in the private graph model with prediction error $\eta(\mathcal{I}_{\bmp^+}) \le \hat{\eta}$. Observe that, for every such instance of $\bmp^+$, both matchings that are potentially returned by \randomizedBoost\ are the outcome of a deterministic group-strategyproof mechanism. In particular, \boost\ is shown to be group-strategyproof in Lemma \ref{lem:gsp}, whereas the group-strategyproofness of \trust\ follows from Theorem \ref{lemma:trust-gsp}.

\sloppy
We show that \randomizedBoost\ is $\sqrt{2(\gamma+1)}$-robust. Let $M_1=\boost(\mathcal{I}_{\bmp^+}, \delta(\gamma))$ and $M_{2} = \trust(\mathcal{I}_{\bmp^+})$. We have: 
\begin{align*}
       v(M^*_{\declaredEdges}) \leq 2v(M_1) + (\delta(\gamma)-1)v(M_1 \cap \hat{M})&\leq 2v(M_1) + (\delta(\gamma)-1)v(M_2)\\
       &=(\delta(\gamma) +1 ) \cdot \bigg(\frac{2}{\delta(\gamma)+1}v(M_1)  + \frac{\delta(\gamma)-1}{\delta(\gamma)+1}v(M_2)\bigg)\\
       &=(\delta(\gamma)+1) \cdot \bigg (pv(M_1) + (1-p) v(M_2) \bigg )\\
       &=(\delta(\gamma)+1) \cdot \mathbb{E}[v(\randomizedBoost(\mathcal{I}_{\bmp^+}, \gamma))].
\end{align*}
The first and second inequality follow from Lemma \ref{lem:missing-lemma} and the fact that $(M_1\cap \hat{M}) \subseteq (D \cap \hat{M}) = M_2$, respectively. Finally, the penultimate equality is true since, by line 4 of $\randomizedBoost$, it holds that $p=2/\sqrt{2(\gamma+1)}=2/(\delta(\gamma)+1)$. Since $\delta(\gamma)+1=\sqrt{2(\gamma+1)}$, the robustness guarantee follows.

We now show that $\randomizedBoost$ achieves an approximation guarantee of $g(\hat{\eta}, \gamma)$. We have:
\begin{equation*}
        \mathbb{E}[v(\randomizedBoost(\mathcal{I}_{\bmp^+}, \gamma))] = p v(M_1) + (1-p)v(\declaredEdges \cap \hat{M})\geq \frac{\delta(\gamma)+1-p}{\delta(\gamma)+1}v(D \cap \hat{M}),
\end{equation*}
where the equality holds since $M_2 = D \cap \hat{M}$, and the inequality follows by applying Lemma \ref{lem:bpm-approx} for $M_1$. After setting $p=2/(\delta(\gamma)+1)$ and $\delta(\gamma)= \sqrt{2(\gamma+1)}-1$ in the above (as $\randomizedBoost$ prescribes), it is a matter of simple calculus to obtain $(1+\nicefrac{1}{\gamma}) \cdot \mathbb{E}[v(\randomizedBoost(\mathcal{I}_{\bmp^+}, \gamma))] \geq v(D \cap \hat{M})$. Thus, by the Lifting Lemma (Lemma~\ref{lem:lifting}), we conclude that \randomizedBoost\ is $(1+\nicefrac{1}{\gamma})/(1-\hat{\eta})$-approximate in expectation. Furthermore, the expected robustness guarantee of $\sqrt{2(\gamma+1)}$ holds independently of the prediction error $\hat{\eta}$. The claimed bound on the expected approximation guarantee $g(\hat{\eta}, \gamma)$ now follows by combining these two bounds. 
\end{proof}

\subsection{Randomized Mechanism for $\emkar^+$}\label{app-subsec:mkar-randomized}

Using the technique from Section $\ref{subsec:randomized-separation}$, we give a randomized universally group-strategyproof mechanism for $\emkar^+$ which we first considered in Section \ref{subsec:mkar-deterministic}.

\begin{mechanism}[t]
\SetKwInput{Input}{Input}
\SetKwInput{Output}{Output}
\DontPrintSemicolon
\caption{\label{alg:randomized-mkar}
\randomizedMKAR$(\mathcal{I}_{\emkar^+}, \gamma)$}
\Input{An instance $\mathcal{I}_{{\emkar^+}}=(G[\declaredEdges], (v_i = s_i)_{i \in \agentSet}, \vec{\jcap})$, confidence parameter $\gamma \geq 1$.}
\Output{A probability distribution of feasible assignments for $\mathcal{I}_{\emkar^+}$.}
Set $\delta(\gamma) = (\sqrt{12\gamma + 13}-3)/2$. \tcp*{Note that $\delta \geq 1$ for all $\gamma \geq 1$.}
Set $M_1=\mkarGreedy(\mathcal{I}_{\emkar^+}, \delta(\gamma))$. \\
Set $M_2=\trust(\mathcal{I}_{\emkar^+})$.\\
Set $p=3/(2+\delta(\gamma))$. \tcp*{Note that $p \in (0,1]$ for all $\gamma \geq 1$.}
\Return $M_1$ with probability $p$ and $M_2$ with probability $1-p$.
\end{mechanism}

\begin{theorem} \label{thm:mkar-randomized} 
    Fix some error parameter $\hat{\eta}\in [0,1]$. Consider the class of instances of $\emkar^+$ in the private graph model and prediction error at most $\hat{\eta}$. Then, for every confidence parameter $\gamma \geq 1$, \randomizedMKAR\ is universally group-strategyproof and has an expected approximation guarantee of
    \[g(\hat{\eta}, \gamma) = \begin{cases*}
                                    \frac{1+\gamma}{\gamma(1-\hat{\eta})} & if $\hat{\eta}\leq 1-\frac{2(1+\gamma)}{\gamma(\sqrt{12\gamma+13} + 1)}$\\
                                    \frac{\sqrt{12\gamma+13} + 1}{2} & otherwise.
    \end{cases*}\]
    In particular, \randomizedMKAR\ is $(1+\nicefrac{1}{\gamma})$-consistent and $\frac12(\sqrt{12\gamma+13} + 1)$-robust (both in expectation).
\end{theorem}

\begin{proof}
Let $\gamma \ge 1$ be fixed arbitrarily. Consider an instance $\mathcal{I}_{{\emkar^+}}=(G[\declaredEdges], (v_i = s_i)_{i \in \agentSet}, \vec{\jcap})$ of this special case of $\emkar^+$ in the private graph model with prediction error $\eta(\mathcal{I}_{\bmp^+}) \le \hat{\eta}$. $\randomizedMKAR$ is universally group-strategyproof since both deterministic mechanisms in its support are group-strategyproof. We show that this holds for \mkarGreedy and \trust in Theorems \ref{thm:mkar-deterministic} and \ref{lemma:trust-gsp}, respectively.

Let $M_1=\mkarGreedy(\mathcal{I}_{\emkar^+}, \delta(\gamma))$ and $M_2=\trust(\mathcal{I}_{\emkar^+})$. We continue by showing that $\randomizedMKAR$ is $\frac12 (\sqrt{12\gamma+13} + 1)$-robust. Note that by construction $\delta(\gamma) \geq 1$ for all $\gamma \ge 1$. For an optimal assignment $M^*_\declaredEdges$ we obtain by Lemma \ref{lemma:MKAR-prerobustness-lemma} that
    \begin{align*}
        v(M^*_{\declaredEdges}) &\leq 3v(M_1) + (\delta(\gamma)-1) v(M_1 \cap \hat{M})\leq 3v(M_1) + (\delta(\gamma) -1)v(D \cap \hat{M})\\
        &= 3v(M_1) + (\delta(\gamma) -1)v(M_2) = (2+\delta(\gamma)) \cdot \bigg(\frac{3}{2+\delta(\gamma)}v(M_1) + \frac{\delta(\gamma) -1}{2+\delta(\gamma)}v(M_2)\bigg)\\
        &= (2+\delta(\gamma)) \cdot \bigg(pv(M_1)+(1-p)v(M_2) \bigg)=(2+\delta(\gamma)) \cdot \mathbb{E}[\randomizedMKAR(\mathcal{I}_{\emkar^+}, \gamma)].
    \end{align*}
The second to last equality follows by observing that $p=3/(2+\delta(\gamma))$, as $\randomizedMKAR$ prescribes. Then, the claimed expected robustness guarantee follows since $\delta(\gamma)+2 = \frac{1}{2}(\sqrt{12\gamma+13}+1)$.

To show the claimed approximation $g(\hat{\eta}, \gamma)$, note that $\randomizedMKAR$ also satisfies the following:
    \begin{align*}
        \mathbb{E}[\randomizedMKAR(\mathcal{I}_{\emkar^+}, \delta(\gamma))]&=pv(M_1) + (1-p)v(D \cap \hat{M})\\
        &\geq v(\hat{M} \cap \declaredEdges) \bigg(\frac{p \delta(\gamma)}{\delta(\gamma) + 1} +1-p\bigg)\\
        &=v(\hat{M} \cap \declaredEdges) \bigg(\frac{3\delta(\gamma)}{(\delta(\gamma)+1)(\delta(\gamma)+2)} + 1 - \frac{3}{2+\delta(\gamma)} \bigg)\\
        &=v(\hat{M} \cap \declaredEdges) \bigg(1-\frac{3}{(\delta(\gamma) +1)(\delta(\gamma) +2)} \bigg)
    \end{align*}    
    
The inequality follows by Lemma \ref{lemma:mkar-approx-without-eta}.
Observe that, for $\delta(\gamma)=\frac{1}{2}(\sqrt{12\gamma + 13}-3)$ it holds that 
    \begin{equation*}
        \frac{(\delta(\gamma)+1)(\delta(\gamma)+2)}{3} = \frac{(\sqrt{12\gamma + 13}-1)(\sqrt{12\gamma + 13}+1)}{12}=\frac{12\gamma +12}{12}=\gamma+1.
    \end{equation*}
Thus, by rearranging terms, we obtain that $(1+\nicefrac{1}{\gamma})\mathbb{E}[\randomizedMKAR(\mathcal{I}_{\emkar^+}, \delta(\gamma))] \geq v(\hat{M} \cap D)$. Thus, by the Lifting Lemma (Lemma~\ref{lem:lifting}), we conclude that \randomizedMKAR\ is $(1+\nicefrac{1}{\gamma})/(1-\hat{\eta})$-approximate in expectation. Further, the robustness guarantee of $\frac{1}{2}(\sqrt{12\gamma+13} + 1)$ in expectation holds independently of the prediction error $\hat{\eta}$. The claimed bound on the expected approximation guarantee $g(\hat{\eta}, \gamma)$ now follows by combining these two bounds.
\end{proof}

\subsection{More General Variants of $\gap^+$} \label{subsec:sigap-vcgap}

In this section, we devise randomized universally group-strategyproof mechanisms for two variants of $\gap^+$, namely $\vcgap^+$ and $\sigap^+$. We write $\mathcal{I}_{\vcgap^+}=(G[\declaredEdges], \values, \sizes, \caps, \hat{M})$ to denote an instance of $\vcgap^+$. Similarly, we use $\mathcal{I}_{\sigap^+}=(G[\declaredEdges], \values, \bm{s}, \vec{\jcap}, \hat{M})$ to denote an instance of $\sigap^+$. In the setting without predictions, \citet{dughmi10} and \citet{chen14} studied multiple $\gap$ variants for the private graph model. But, to the best of our knowledge, \vcgap has not been considered in the literature before, and no deterministic strategyproof $O(1)$-approximate mechanisms are known for $\vcgap$ and $\sigap$ in the private graph model.

\minisec{Greedy Mechanism for $\vcgap$ and $\sigap$}
In Section \ref{subsec:mkar-deterministic} we demonstrated how Mechanism $\greedy$ combined with an appropriate ranking function, can serve as a group-strategyproof mechanism for the special case of $\emkar^+$. Here, we present two different instantiations of ranking functions, one for $\vcgap$ and one for $\sigap$, which can be coupled with $\greedy$ to obtain a group-strategyproof mechanism for their respective classes of instances. 

We present our ranking function for $\vcgap$ first. We define the function $\vec{z}_{\vcgap}: L \times R \mapsto \mathbb{R}^3$ as follows: Let $\sigma=(\sigma(1), \dots, \sigma(m))$ be the ordinal consensus permutation of resources in $R$ for this instance. Note that such a permutation is guaranteed to exist, as the instance is $\vcgap^+$. Then, for every pair $(i,j) \in L \times R$, we define
\begin{equation}
    \label{eq:vcgap-zeta}
    \vec{z}_{\vcgap}((i,j)) := \Big( - \sigma(j), \nicefrac{v_{ij}}{s_{ij}}, -i \Big). 
\end{equation}
This function instructs $\greedy$ to rank edges in $D$ so that the edges linked to the most ``sought'' resources, according to the ordinal consensus $\sigma$, are considered first. Then, among the edges linked to each resource, the second-order criterion instructs $\greedy$ to give precedence to the edge with the highest value per size ratio. If any ties remain, they are broken in increasing index of $i$.

For $\sigap^+$, we define the function $\vec{z}_{\sigap}: L \times R \mapsto \mathbb{R}^3$ as follows: For every pair $(i,j) \in L \times R$, we define
\begin{equation}
    \label{eq:sigap-zeta}
    \vec{z}_{\sigap}((i,j)) := \Big( \nicefrac{v_{ij}}{s_{ij}}, -i, -j \Big).
\end{equation}
This ranking function is particularly straightforward; it instructs $\greedy$ to prioritize edges with the highest value per size ratio in the greedy ordering. Then, in case of ties, they are broken in increasing index of first $i$ and then $j$.

Throughout this section, when invoking \greedy\ for $\vcgap^+$ and $\sigap^+$ instances, we will refer to the pairing of \greedy\ with the corresponding ranking functions $\vec{z}_{\vcgap}$ and $\vec{z}_{\sigap}$ respectively.
In Lemma \ref{lemma:greedy-sp} we show that each of the two instantiations of $\greedy$ described above is a GSP mechanism for its respective class of instances.

\begin{restatable}{lemma}{sixfour}\label{lemma:greedy-sp}
    \sloppy
    Mechanism $\greedy$ coupled with ranking function $\vec{z}_{\vcgap}$ (or $\vec{z}_{\sigap}$) is a group-strategyproof mechanism for instances of $\ \vcgap^+$ (or $\sigap^+$, respectively).
\end{restatable}

\begin{proof}
    We show that both functions are truth-inducing (Definition \ref{def:truth-inducing}). We first show it for the ranking function $\vec{z}_{\vcgap}$ and the class of $\vcgap^+$ instances. Let $\mathcal{I}_{\vcgap^+}=(G[D], \vec{v}, \vec{s}, \vec{C}, \hat{M})$ be an instance of $\vcgap^+$ and let $\sigma$ be the permutation of resources that indicates the consensus of agents for their value. Indeed, the extended lexicographic order with respect to $\vec{z}_{\vcgap}$ as defined in \eqref{eq:vcgap-zeta}, is strict and total. Furthermore, by the definition of $\sigma$, it is true that for each agent $i \in L$ and all $e = (i,j), e'=(i,j') \in L \times R$, $\sigma(j) \leq \sigma(j')$ implies $v_{ij} \geq v_{ij'}$. Since the ranking function $\vec{z}_{\vcgap}$ is truth-inducing, the group-strategyproofness of this greedy mechanism follows from Theorem \ref{thm:greedy-gsp-sufficient}.

    We show that the ranking function $\vec{z}_{\sigap}$ is truth-inducing (Definition \ref{def:truth-inducing}) for the class of $\sigap^+$ instances. Let $\mathcal{I}_{\sigap^+}=(G[D], \vec{v}, \vec{s}, \vec{C}, \hat{M})$ be an instance of $\sigap^+$. It is easy to observe that the extended lexicographic order with respect to $\vec{z}_{\vcgap}$ as defined in \eqref{eq:sigap-zeta}, is strict and total. Furthermore, for each agent $i \in L$ and all $e = (i,j), e'=(i,j') \in L \times R$, since $\mathcal{I}_{\sigap^+}$ is an instance of $\sigap^+$ (and therefore $s_{ij}=s_{ij'}$), $\frac{v_{ij}}{s_{ij}} \geq \frac{v_{ij'}}{s_{ij'}}$  directly implies $v_{ij} \geq v_{ij'}$. Since the ranking function $\vec{z}_{\sigap}$ is truth-inducing, the group-strategyproofness of this greedy mechanism follows from Theorem \ref{thm:greedy-gsp-sufficient}. 
\end{proof}

Note that both ranking functions defined above do not depend on the predicted assignment $\hat{M}$ in any way. Furthermore, the greedy mechanisms described above do not guarantee worst-case approximation guarantees when run as stand-alone mechanisms. However, both ranking functions ensure that agents are processed by $\greedy$ in an efficient way.

\begin{observation}\label{z-v/s-sigap-vcgap} 
For an instance $\mathcal{I_{\gap^+}}=(G[\declaredEdges], \vec{v}, \vec{s}, \vec{C}, \hat{M})$ of $\vcgap^+$ (or $\sigap^+$), the corresponding ranking function $\vec{z}$ as defined in \eqref{eq:vcgap-zeta} (or \eqref{eq:sigap-zeta}, respectively) satisfies the following property: for every resource $j \in R$, and every $(i,j), (i',j) \in \declaredEdges_j$ with $\vec{z}((i,j)) \lex \vec{z}((i',j))$, it holds that $\frac{v_{ij}}{s_{ij}} \geq \frac{v_{i'j}}{s_{i'j}}$.
\end{observation} 

\minisec{Randomized Mechanism for $\vcgap^+$ and $\sigap^+$}
We present our randomized universally group-strategyproof mechanism for $\vcgap^+$ and $\sigap^+$.
There are two main pillars in our approach. The first one is that we randomize over the respective \greedy\ mechanism presented in Section~\ref{sec:GreedyTemplate}, which processes agents in order of efficiency (as argued in Observation~\ref{z-v/s-sigap-vcgap} above), and a complementary mechanism, which processes agents in order of their values. While neither of these mechanisms  achieves a bounded approximation guarantee by itself, their (probabilistic) combination does in expectation. In fact, this is the key idea that \citet{chen14} used to devise the current state-of-the-art strategyproof mechanism for $\sigap$ and special cases of $\vcgap$. Inspired by this idea, we instead randomize over \greedy\ and our mechanism \boost for $\bmp^+$ to leverage the predicted assignment. Finally, we combine the above scheme with a third mechanism, namely \trust, to follow the prediction with some (small) probability. 

We refer to the resulting mechanism as $\randomizedGAP$; see Mechanism~\ref{alg:randomized-gap}. Note that the assignments $M_1$ and $M_2$ computed by \boost and \greedy, respectively, are always returned with positive probability $p = 2/(3 + \gamma)$, while the predicted assignment output by \trust is returned with probability $1-2p$, which is positive only if $\gamma > 1$ (i.e., when there is some confidence in the prediction). A subtle point that needs some clarification here is that the predicted assignment $\hat{M}$ of the constructed instance $\mathcal{I}_{\bmp^+}$ passed on to \boost\ is a many-to-one assignment. However, as argued above (see Section~\ref{rem:boost-extensions}, Extension (E1)), \boost can handle such alterations as well. In fact, $M_1$ being a one-to-one assignment output by \boost\ suffices to prove bounded approximation guarantees for \randomizedGAP. 

\begin{mechanism}[t]
\SetKwInput{Input}{Input}
\SetKwInput{Output}{Output}
\DontPrintSemicolon
\caption{\label{alg:randomized-gap}
\randomizedGAP$(\mathcal{I}_{\gap^+}, \vec{z}, \gamma)$}
\Input{An instance $\mathcal{I}_{{\gap^+}}=(G[\declaredEdges], \values, \vec{\sizes}, \vec{C}, \hat{M})$, a ranking function  $\vec{z}: L \times R \mapsto \mathbb{R}^k$ for some $k \in \mathbb{N}$ and a confidence parameter $\gamma \geq 1$} 
\Output{A probability distribution of feasible assignments for $\mathcal{I}_{\gap^+}$.} 
Construct an instance $\mathcal{I}_{\bmp^+}=(G[\declaredEdges], \values, \hat{M})$. \\
Set $M_1=\boost(\mathcal{I}_{\bmp^+}, \gamma)$ and $M_2 = \greedy(\mathcal{I}_{\gap^+}, \vec{z})$.\\
Set $p=2/(3+\gamma)$.\\
\Return $M_1$ with probability $p$, $M_2$ with probability $p$ and $\trust(\mathcal{I}_{\gap^+})$ with probability $1-2p$.
\end{mechanism}

\begin{restatable}{theorem}{sixfive}\label{thm:randomized-gap-main}
Fix some error parameter $\hat{\eta} \in [0,1]$. Consider the class of instances of $\sigap^+$ (or $\vcgap^+$, respectively) in the private graph model with prediction error at most $\hat{\eta}$. Then, for every confidence parameter $\gamma \ge 1$, \randomizedGAP is universally group-strategyproof and has an expected approximation guarantee of
\begin{equation}
    g(\hat{\eta}, \gamma) = 
    \begin{cases}
        \frac{3+\gamma}{\gamma(1-\hat{\eta})} & \text{if $\hat{\eta} \le 1 - \frac{1}{\gamma}$,} \\
        3+\gamma & \text{otherwise.}
        \end{cases} \label{eq:approxGAP}
\end{equation}
In particular, \randomizedGAP is $(1+\nicefrac{3}{\gamma})$-consistent and $(3+\gamma)$-robust (both in expectation).
\end{restatable}

In what follows, we focus on the proof of Theorem \ref{thm:randomized-gap-main}. We first show that \randomizedGAP is universally group-strategyproof for these instances.

\begin{proof}[Proof of Theorem~\ref{thm:randomized-gap-main} (Universal Group-Strategyproofness)]
Fix $\gamma \geq 1$ arbitrarily. Observe that, for every such instance of $\mathcal{I}_{\gap^+}$, each one of the three assignments potentially returned by \randomizedGAP is the outcome of a deterministic GSP mechanism. Indeed, the fact that \greedy\ and \trust\ are group-strategyproof follows from Lemma \ref{lemma:greedy-sp} and Theorem \ref{lemma:trust-gsp} respectively. Furthermore, in Section~\ref{subsec:bpm-gsp}, we showed that \boost\ is group-strategyproof. 
\end{proof}

Next, we turn our attention to the expected approximation guarantee of Theorem \ref{thm:randomized-gap-main}. To that end, we introduce some auxiliary notation. For a fixed $\gamma \geq 1$ and an instance $\mathcal{I}_{\gap^+}=(G[D], \values, \sizes, \vec{C}, \hat{M})$, let $M_1$ represent the outcome of $\boost$ and $M_2$ represent the outcome of $\greedy$, as prescribed by line 2 of $\randomizedGAP(\mathcal{I}_{\gap^+}, \vec{z}, \gamma)$. Furthermore, let $M$ be any feasible assignment for $\mathcal{I}_{\gap^+}$. We denote by $L^<(M) \subseteq L$ the set of agents who have been assigned to a resource of smaller value under both $M_1$ and $M_2$ compared to their value under $M$. Formally,
\begin{equation*}
L^<(M) =\bigg \{ i \in L \bigm| v_{iM_{1}(i)} < v_{i M(i)} \land v_{iM_{2}(i)} < v_{i M(i)} \bigg \}.
\end{equation*}
Note that $L^<(M)$ depends on both $M_1$ and $M_2$, but for ease of notation, we omit these arguments. Our approach proceeds as follows: First, we begin with a simple observation. Then, we establish two technical lemmas, Lemma~\ref{lemma:ArbitraryM} and Lemma~\ref{lemma:randomized-gap-key-lemma}, which will be useful in our analysis. With these lemmas in place, the expected approximation guarantee of Theorem \ref{thm:randomized-gap-main} will follow promptly.

We begin by an observation which follows by the definition of the set $L^<(M)$, given an assignment $M$ and assignments $M_1$ and $M_2$ as computed in line 2 of the mechanism.
\sloppy
\begin{observation}\label{observation:not-in-s-randomized}
Let $\gamma \geq 1$. Let $\mathcal{I}_{\gap^+}=(G[D], \values, \sizes, \vec{C}, \hat{M})$ be an instance of $\vcgap^+$ (or $\sigap^+$, respectively) and let $M_1$ and $M_2$ be the assignments computed by \boost and \greedy\ in line 2 of $\randomizedGAP(\mathcal{I}_{\gap^+}, \vec{z}, \gamma)$. Furthermore, let $M$ be a feasible assignment for $\mathcal{I}_{\gap^+}$. Then,  $v(M_1) + v(M_2) \ge \sum_{j\in R}v(M(j) \setminus L^<(M))$.
\end{observation} 

We continue with the two technical lemmas. The first one exploits the fact that $\greedy$ considers agents in order of efficiency.

\begin{lemma} \label{lemma:ArbitraryM}
\sloppy
Let $\gamma \geq 1$. Let $\mathcal{I}_{\gap^+}=(G[D], \values, \sizes, \vec{C}, \hat{M})$ be an instance of $\vcgap^+$ (or $\sigap^+$, respectively) and let $M_1$ and $M_2$ be the assignments computed by \boost and \greedy\ in line 2 of $\randomizedGAP(\mathcal{I}_{\gap^+}, \vec{z}, \gamma)$. Furthermore, let $M$ be any feasible assignment for $\mathcal{I}_{\gap^+}$ and let $j \in \taskSet$ be a resource with $M(j) \cap L^<(M) \neq \emptyset$. Finally, let $\ell \in \argmax_{i \in M(j) \cap L^<(M)}s_{ij}$. It holds that $v(M_2(j)) + v_{\ell j} \geq v\big(M(j) \cap L^<(M)\big)$.
\end{lemma} 

\begin{proof}
Consider any agent $k \in M(j) \cap L^<(M)$. By the definition of $L^<(M)$, it holds that $v_{k j}> v_{k M_2(k)}$. Observe that, for both ranking functions $\vec{z}_{\vcgap}$ and $\vec{z}_{\vcgap}$ this implies that $(k, j) \lex (k, M_2(k))$. Therefore, agent $k$ was first considered by $\greedy$ to be matched to resource $j$ but was rejected and instead was matched to resource $M_2(k)$.

Let $T(k) \subseteq M_2(j)$ be the set of agents that were assigned to resource $j$ by $\greedy$ \emph{before} considering matching agent $k$ to resource $j$. Note that $T(k)$ is not empty since, by assumption $s_{kj} \leq C_j$. We have:
\begin{equation}
    \label{eq:6.2-not-fit}
    \sum_{i \in T(k)}s_{ij} +s_{kj} > C_j \geq \sum_{i \in M(j)\cap L^<(M)}s_{ij}.
\end{equation}
The first inequality follows from the fact that the condition on line 5 of $\greedy$ evaluated to \emph{False}. Then, the second inequality follows from the feasibility of $M$.

Further, by Observation \ref{z-v/s-sigap-vcgap}, for each agent $i \in T(k)$ it holds that $\frac{v_{kj}}{s_{kj}} \leq \frac{v_{ij}}{s_{ij}}$. By summing over all agents in $T(k)$ and rearranging terms, we obtain that
\begin{equation}
    \label{eq:6.2-efficiency}
    v_{kj} \leq s_{kj} \frac{\sum_{i \in T(k)}v_{ij}}{\sum_{i \in T(k)}s_{ij}}
\end{equation}
Recall that $\ell = \argmax_{i \in M(j) \cap L^<(M)}s_{ij}$. We can conclude the following:
\begin{align*}
    v(M(j) \cap L^<(M))-v_{\ell j}&= \sum_{k \in (M(j) \cap L^<(M))\setminus\{ \ell \}}v_{kj} \leq \sum_{k \in (M(j) \cap L^<(M))\setminus\{ \ell \}} s_{kj} \frac{\sum_{i \in T(k)}v_{ij}}{\sum_{i \in T(k)}s_{ij}}\\
    &\leq v(M_2(j)) \cdot \sum_{k \in (M(j) \cap L^<(M))\setminus\{ \ell \}} \frac{s_{kj}}{\sum_{i \in T(k)}s_{ij}}\\
    &\leq v(M_2(j)) \cdot  \sum_{k \in (M(j) \cap L^<(M))\setminus\{ \ell \}} \frac{s_{kj}}{\sum_{i \in (M(j) \cap L^<(M))}s_{ij}-s_{kj}}\\
    &\leq v(M_2(j)) \cdot  \sum_{k \in (M(j) \cap L^<(M))\setminus\{ \ell \}} \frac{s_{kj}}{\sum_{i \in (M(j) \cap L^<(M))}s_{ij}-s_{\ell j}}=v(M_2(j)).
\end{align*}
The first inequality follows from \eqref{eq:6.2-efficiency} and the second inequality by the fact that $T(k) \subseteq M_2(j)$, for each agent $k \in (M(j) \cap L^<(M)) \setminus \{ \ell \}$. Then, the next inequality follows by applying \eqref{eq:6.2-not-fit} for each such agent $k$. Finally, the last inequality follows from the definition of agent $\ell$. The Lemma follows.
\end{proof}
Lemma \ref{lemma:randomized-gap-key-lemma} will be the key to complete the proof of Theorem \ref{thm:randomized-gap-main}, as we show below.

\sloppy
\begin{lemma}
    \label{lemma:randomized-gap-key-lemma}
    Let $\gamma \geq 1$. Let $\mathcal{I}_{\gap^+}=(G[D], \values, \sizes, \vec{C}, \hat{M})$ be an instance of $\vcgap^+$ (or $\sigap^+$, respectively) and let $M_1$ and $M_2$ be the assignments computed by \boost and \greedy\ in line 2 of $\randomizedGAP(\mathcal{I}_{\gap^+}, \vec{z}, \gamma)$. Furthermore, let $M$ be any feasible assignment for $\mathcal{I}_{\gap^+}$. For each $j \in R$ with $M(j) \cap L^<(M) \neq \emptyset$, fix an agent $\ell(j) \in \argmax_{i \in M(j) \cap L^<(M)}s_{ij}$. Then,
    \begin{equation}\label{eq:randomized-gap-key-eq}
        v(M) \leq v(M_1) + 2v(M_2) + \sum_{\substack{j \in R: \\ M(j) \cap L^<(M) \neq \emptyset}}v_{\ell(j)j}
    \end{equation}
\end{lemma}

\begin{proof}
    Indeed, by applying Observation \ref{observation:not-in-s-randomized} for $M$ and afterwards Lemma \ref{lemma:ArbitraryM} for $M$, each $j \in R$ with $M(j) \cap L^<(M) \neq \emptyset$ and each $\ell(j)$ we obtain that
    \begin{align*}
        v(M) &= \sum_{j \in R}v(M(j)) =  \sum_{j \in R}v(M(j) \setminus L^<(M)) + \sum_{j \in R}v(M(j) \cap L^<(M))\\
        &\leq v(M_1) + v(M_2) + \sum_{j \in R}v(M(j) \cap L^<(M)) \leq v(M_1) + 2v(M_2) + \sum_{\substack{j \in R: \\ M(j) \cap L^<(M) \neq \emptyset}}v_{\ell(j)j}.\qedhere
    \end{align*}
\end{proof}

\begin{proof}[Proof of Theorem~\ref{thm:randomized-gap-main} (Robustness)]
Let $\gamma \geq 1$. Let $\mathcal{I}_{\gap^+}=(G[D], \values, \sizes, \vec{C}, \hat{M})$ be an instance of $\vcgap^+$ (or $\sigap^+$, respectively) and let $M_1$ and $M_2$ be the assignments computed by \boost and \greedy\ in line 2 of the mechanism. Let $M$ the outcome of the randomized mechanism. 

Consider an edge $(i,j)$ so that $i \in \argmax_{t \in M^*_{D}(j) \cap L^<(M^*_{D})}{s_{tj}}$. By the construction of $\boost$ and the definition of $L^<(M^*_{\declaredEdges})$, it must be that agent $i$ proposed to resource $M^*_{\declaredEdges}(i)=j$. However, resource $j$ rejected the $\gamma$-boosted offer $\theta_{ij}(\gamma, \hat{M})$ of agent $i$ and instead opted for the $\gamma$-boosted offer of a different agent $M_1(j)$. Using the above fact, we obtain that:
\begin{equation}\label{eq:robustness-s1-randomized-ub}
 v_{ij} \leq \theta_{ij} \leq \theta_{M_1(j)j} = v(M_1(j)) + \mathds{1}_{M_1(j) \in \hat{M}(j)} \cdot (\gamma -1)v(M_1(j)).
\end{equation}

In \eqref{eq:robustness-s1-randomized-ub} we use $\theta_{ij}:=\theta_{ij}(\gamma, \hat{M})$ and $\theta_{M_1(j)j}:=\theta_{M_1(j)j}(\gamma, \hat{M})$ for brevity. Note that both equalities follow from the definition of the $\gamma$-boosted value in \eqref{eq:offer}. We continue by applying Lemma \ref{lemma:randomized-gap-key-lemma} to $M^*_{D}$ and each $j \in R$ with $M^*_{D}(j) \cap L^<(M^*_{D}) \neq \emptyset$ and agent $\ell(j) \in \argmax_{i \in M^*_D(j) \cap L^<(M^*_D)}s_{ij}$. Thus, we can expand \eqref{eq:randomized-gap-key-eq} as follows:
\begin{align*}
    v(M^*_{\declaredEdges}) &\leq v(M_1) + 2v(M_2) + \sum_{\substack{j \in R: \\ M^*_D(j) \cap L^<(M^*_D) \neq \emptyset}}v_{\ell(j)j}\\
    &\leq v(M_1) + 2v(M_2) +  \sum_{\substack{j \in R: \\ M^*_D(j) \cap L^<(M^*_D) \neq \emptyset}} \bigg(v(M_1(j)) + \mathds{1}_{M_1(j) \in \hat{M}(j)} \cdot (\gamma -1)v(M_1(j)) \bigg)\\
    &\leq v(M_1) + 2v(M_2) + \sum_{j \in R}\bigg(v(M_1(j)) + \mathds{1}_{M_1(j) \in \hat{M}(j)} \cdot (\gamma -1)v(M_1(j)) \bigg)\\
    &=  2v(M_1) + 2v(M_2) +(\gamma-1) \cdot \sum_{j \in R}\mathds{1}_{M_1(j) \in \hat{M}(j)}v(M_1(j))\\
    &= 2v(M_1) + 2v(M_2) + (\gamma-1)v(M_1 \cap \hat{M}) \leq 2v(M_1) + 2v(M_2) + (\gamma-1)v(\trust(\mathcal{I}_{\gap}^+))\\
    &= (3+\gamma) \cdot \bigg( pv(M_1)+pv(M_2) + (1-2p) v(\trust \big(\mathcal{I}_{\gap^+} \big) \bigg)= (3+\gamma) \mathbb{E}[M].
\end{align*}
In the above analysis, the last inequality follows from the fact that $M_1 \cap \hat{M} \subseteq D \cap \hat{M} =\trust(\mathcal{I}_{\gap^+}).$ Finally, the next equality follows from the fact that the randomized mechanism sets $p=2/(3+\gamma)$. The claimed expected robustness guarantee ensues.
\end{proof}

\begin{proof}[Proof of Theorem~\ref{thm:randomized-gap-main} (Approximation)]
Let $\gamma \ge 1$ be fixed arbitrarily. Consider an instance $\mathcal{I}_{\gap^+}=(G[D], \values, \sizes, \vec{C}, \hat{M})$ of $\vcgap^+$ (or $\sigap^+$, respectively) with prediction error $\eta(\mathcal{I}_{\gap^+}) \le \hat{\eta}$. Let $M_1$ and $M_2$ be the assignments computed by \boost and \greedy\ in line 2 of \randomizedGAP and let $M$ be the assignment returned by $\randomizedGAP$. For notational convenience, we use $\hat{M}_{\declaredEdges}:= \hat{M} \cap \declaredEdges$ and, for each $j \in R$, $\hat{M}_{\declaredEdges}(j):=\hat{M}(j)\cap \declaredEdges$.

Consider an edge $(i,j)$ so that $i \in \argmax_{t \in \hat{M}_{D}(j) \cap L^<(\hat{M}_{D})}{s_{tj}}$. By the construction of $\boost$ and the definition of $L^<(\hat{M}_{\declaredEdges})$, it must be that agent $i$ proposed to resource $j$ for which $i$ is predicted to be assigned. However, resource $j$ rejected the $\gamma$-boosted offer $\theta_{ij}(\gamma, \hat{M})$ of agent $i$ and instead opted for the $\gamma$-boosted offer of a different agent $M_1(j)$. Using the above fact, we obtain that:
\begin{equation}\label{eq:approximation-s1-randomized-ub}
 v_{ij} = \frac{\theta_{ij}}{\gamma} \leq \frac{\theta_{M_1(j)j}}{\gamma} \leq \frac{\gamma v(M_1(j))}{\gamma}= v(M_1(j)). 
\end{equation}
In \eqref{eq:approximation-s1-randomized-ub} we use $\theta_{ij}:=\theta_{ij}(\gamma, \hat{M})$ and $\theta_{M_1(j)j}:=\theta_{M_1(j)j}(\gamma, \hat{M})$ for brevity. Note that the first equality and the second inequality follow from the definition of the $\gamma$-boosted value in \eqref{eq:offer}. Additionally, the first equality holds since $(i,j) \in \hat{M}_{\declaredEdges}$. We now apply Lemma \ref{lemma:randomized-gap-key-lemma} for $\hat{M}_{\declaredEdges}$ and each $j \in R$ with $ \hat{M}_{\declaredEdges}(j) \cap L^<(\hat{M}_\declaredEdges) \neq \emptyset$ and agent $\ell(j) \in \argmax_{\hat{M}_{\declaredEdges}(j) \cap L^<(\hat{M}_\declaredEdges)}s_{ij}$. We expand \eqref{eq:randomized-gap-key-eq} as follows:
\begin{align}
    v(\hat{M}_{\declaredEdges}) &\leq v(M_1) + 2v(M_2) + \sum_{\substack{j \in R: \\ \hat{M}_D(j) \cap L^<(\hat{M}_D) \neq \emptyset}}v_{\ell(j)j}\leq v(M_1) + 2v(M_2) +  \sum_{\substack{j \in R: \\ \hat{M}_D(j) \cap L^<(\hat{M}_D) \neq \emptyset}} v(M_1(j))\nonumber\\
    &\leq v(M_1) + 2v(M_2) + \sum_{j \in R}v(M_1(j))=2v(M_1) + 2v(M_2).\label{eq:randomized-gap-apx-half}
\end{align}
The second inequality follows by applying \eqref{eq:approximation-s1-randomized-ub} for each resource $j\in R$ with $\hat{M}_{\declaredEdges}(j) \cap L^<(\hat{M}_\declaredEdges) \neq \emptyset$ and agent $\ell(j)$ as specified above. Thus,
\begin{equation*}
   \mathbb{E}[M]= pv(M_1)+pv(M_2) + (1-2p)v(\hat{M}_{\declaredEdges}) \geq v(\hat{M}_{\declaredEdges}) \bigg(\frac{p}{2} +1 -2p \bigg)=v(\hat{M}_{\declaredEdges}) \cdot \frac{\gamma}{3+\gamma}.
\end{equation*}
The first equality follows from the fact that $\hat{M}_{\declaredEdges}=\trust(\mathcal{I}_{\gap^+})$. Then, the inequality follows from \eqref{eq:randomized-gap-apx-half}. Finally, the last equality is due to the fact that $p = {2}/{(3+\gamma)}$ on line 3 of our mechanism.

By rearranging terms we obtain ${(3+\gamma)}/{\gamma} \cdot \mathbb{E}[M] \geq v(\hat{M}_{\declaredEdges})$. By the Lifting Lemma (Lemma \ref{lem:lifting}) we conclude that $\randomizedGAP$ attains an expected approximation of $\frac{3+\gamma}{\gamma(1-\hat{\eta})}$. At the same time, the expected robustness guarantee of $3+\gamma$ holds independently of the prediction error $\hat{\eta}$. The claimed bound on the expected approximation guarantee $g(\hat{\eta}, \gamma)$ in \eqref{eq:approxGAP} now follows by combining these two bounds.
\end{proof}
This concludes the proof of Theorem \ref{thm:randomized-gap-main}.

\section{Conclusions}

In this work, we contribute to the emerging line of research on approximate mechanism design with predictions for environments without monetary transfers. We study generalized assignment problems with predictions in the private graph model introduced by \citet{dughmi10}.
For the Bipartite Matching Problem $\bmp^+$, we derive a new mechanism, $\boost$ (in Section \ref{sec:matching}), and show that it achieves the optimal trade-off between consistency and robustness. Further, we show that our mechanism achieves an approximation guarantee that smoothly transitions between the consistency and robustness guarantees, depending on a natural error parameter. Given our lower bound (in Section~\ref{sec:lowerBounds}), this approximation guarantee might still be improved, but it is off by at most a factor of $1 + \nicefrac{1}{\gamma}$. We leave it for future work to close this gap.
Furthermore, we use $\boost$ as a core component in most of our randomized mechanisms in Section \ref{sec:randomized}. We show that combining $\boost$ with the mechanism $\trust$ (when randomization is allowed) can lead to improved expected approximation guarantees for $\bmp^+$. Additionally, combining the aforementioned mechanisms with simple greedy deterministic mechanisms (introduced in Section \ref{sec:greedy}) yields randomized, universally group-strategyproof mechanisms for more general $\gap^+$ problems (as we show in Section \ref{sec:randomized}).

We believe that our work offers a comprehensive treatment of leveraging predictions within the private graph model. However, several avenues for future research emerge. Firstly, we believe that the consistency-robustness trade-offs and, more generally, the approximation guarantees for the more general variants of $\gap^+$ considered in this paper can be further improved. In fact, these are intriguing problems even in the setting without predictions. 
Secondly, we think it is worthwhile to investigate to which extent our techniques can be applied to other mechanism design problems without monetary transfers, when the setting is augmented with structural predictions such as the sought combinatorial object. 
Finally, an interesting new direction would be to study matching-based problems with more general preference orders of the agents in a learning-augmented environment.

\section*{Acknowledgements}
This work was supported by the Dutch Research Council (NWO) through its Open Technology Program, proj.~no.~18938, and the Gravitation Project NETWORKS, grant no.~024.002.003. It has also been funded by the European Union under the EU Horizon 2020 Research and Innovation Program, Marie Skłodowska-Curie Grant Agreement, grant no.~101034253.

% Bibliography
\bibliographystyle{plainnat}
\bibliography{arxiv}

\newpage 
\section*{Appendix}
\appendix

\section{Self-Contained Proof: \boost is GSP} \label{app:GSP}

In this section, we provide a self-contained proof of the group-strategyproofness of $\boost$ (Mechanism \ref{alg:bpm}). We adapt the proof of \citet{GS85} to our setting. 

Let $\gamma \ge 1$ be fixed. Consider an instance $\mathcal{I}_{\bmp^+}=(G[\declaredEdges], \values, \hat{M})$ of $\bmp^+$ with compatibility declarations $\declaredEdges$ and a private graph $G[\privateEdges]$. Let $M$ be a matching in $G[\declaredEdges]$. 
In our context, an edge $(i,j) \in \declaredEdges$ \emph{blocks} $M$ if $i$ and $j$ both prefer to be matched to each other rather than to their respective mates, i.e., (1) $v_{ij} > v_{iM(i)}$ and (2) $\theta_{ij} > \theta_{M(j)j}$. Throughout this section, whenever we use $>$ to compare two values or offers, we implicitly assume that the respective agent preference order $\succ_i$ and resource preference order $\succ_j$ as defined in Section~\ref{subsec:bpm-gsp} is used.
For ease of notation, we define $v_{iM(i)} = 0$ if $M(i) = \emptyset$ (i.e., $i$ is unmatched) and $\theta_{M(j) j} = 0$ if $M(j) = \emptyset$ (i.e., $M(j)$ is unmatched). 
Recall that a matching $M$ is \emph{stable} if it is not blocked by any edge. The following lemma will be useful. 

\begin{lemma}[Adapted from \citet{GS85}]\label{lem:GSP-key}
Let $M(\declaredEdges)$ be the matching computed by \boost for compatibility declarations $\declaredEdges$, and let $M'$ be an arbitrary matching in $G[\declaredEdges]$. 
Let $X \subseteq \agentSet$ be the set of agents that prefer their mate in $M'$ over their mate in $M$, i.e., $X = \sset{i \in \agentSet}{v_{iM'(i)} > v_{iM(i)}}$. Then there exists an edge $(i,j) \in \declaredEdges$ with $i \notin X$ that blocks $M'$. 
\end{lemma}

\begin{proof}
Let $Y = M(X)$ and $Y' = M'(X)$ be the sets of resources that are matched with the agents in $X$ under $M$ and $M'$, respectively. 
Note that, by the definition of $X$, each agent $i \in X$ first proposes to $M'(i)$ but is rejected (immediately or later) and only later proposes to $M(i)$ and is accepted (if $M(i) \neq \emptyset$) or remains unmatched (if $M(i) = \emptyset$). In particular, we must have $|Y'| \ge |Y|$. We distinguish two cases: 

\underline{Case 1:} $Y' \setminus Y \neq \emptyset$. 
Consider some resource $j \in Y' \setminus Y$ and let $k \in X$ be the agent that is matched to $j$ under $M'$, i.e., $M'(k) = j$.
As argued above, when $k$ proposes to $j$ it is rejected. 
In particular, this implies that $j$ must have a mate $i = M(j)$ that it prefers over $k$, i.e., $\theta_{ij} > \theta_{M'(j)j}$.
Also, $i$ cannot be part of $X$ because otherwise $M(i) = j \in Y \cap Y'$, contradicting our assumption. Thus, $v_{ij} > v_{iM'(i)}$. It follows that $(i,j) \in \declaredEdges$ with $i \notin X$ blocks $M'$.

\underline{Case 2:} $Y' \setminus Y = \emptyset$. 
Note that each resource $j \in Y = Y'$ has exactly two edges, $(M(j), j)$ and $(M'(j), j)$, incident to it. As observed above, each agents $i \in X$ first proposes to $M'(i)$ and is rejected and then proposes to $M(i)$ and is accepted. In particular, each resource $j$ rejects the offer by $M'(j)$ before it accepts $M(j)$. Let $k \in X$ be the agent who proposes last to a resource $j \in Y$. Then $k = M(j)$ must be accepted. Also, $M'(j)$ proposed to $j$ before and was rejected. This implies that $j$ was tentatively matched to some $i \notin X$ when $j$ accepts $k$. Agent $i$ is rejected by $j$ and thus $v_{ij} > v_{iM(i)}$. In addition, $i \notin X$ and thus $v_{iM(i)} > v_{iM'(i)}$. We thus have $v_{ij} > v_{iM'(i)}$. Also, note that $M'(j)$ got rejected before $i$ and thus $\theta_{ij} > \theta_{M'(j)j}$. We conclude that $(i,j) \in \declaredEdges$ with $i \notin X$ blocks $M'$.
\end{proof}

We can now show group-strategyproofness. 

\begin{lemma}\label{lem:gsp}
Let $S \subseteq \agentSet$ be a subset of agents, and let $\declaredEdges_{-S}$ be arbitrary compatibility declarations of the agents in $\agentSet \setminus S$. 
Then for any compatibility declarations $\declaredEdges'_{S}$ of the agents in $S$, there always exists an agent $i \in S$ such that $u_i(\privateEdges_S, \declaredEdges_{-S}) \ge 
u_i(\declaredEdges'_S, \declaredEdges_{-S})$.
\end{lemma}

\begin{proof}
The proof is by contradiction. 
Suppose there exists some set $S \subseteq \agentSet$ and declaration profiles $\declaredEdges = (\privateEdges_S, \declaredEdges_{-S})$ and $\declaredEdges' = (\declaredEdges'_S, \declaredEdges_{-S})$ such that for every agent $i \in S$, $u_i(\privateEdges_S, \declaredEdges_{-S}) < u_i(\declaredEdges'_S, \declaredEdges_{-S})$.
Let $M$ and $M'$ be the matchings output by \boost for profiles $\declaredEdges$ and $\declaredEdges'$, respectively. 
Let $X \subseteq \agentSet$ be the subset of agents that prefer their mate in $M'$ over their mate in $M$, i.e., $
X = \sset{i \in \agentSet}{v_{iM'(i)} > v_{iM(i)}}$. Then $S \subseteq X$. 
Also, note that for each $i \in S$, the edge $(i, M'(i))$ must be part of $\privateEdges$, as otherwise $u_i(\declaredEdges'_S, \declaredEdges_{-S}) = 0$, which gives a contradiction.

Consider the matching $\tilde{M} = M' \cap \declaredEdges$. Note that $\tilde{M}$ is a feasible matching in $G[\declaredEdges]$ and we can thus apply Lemma~\ref{lem:GSP-key}. Thus, there exists an edge $(i,j) \in \declaredEdges$ with $i \notin X$ that blocks $\tilde{M}$ with respect to $\declaredEdges$. 
Note that, because $i \notin X$, the set of edges $\declaredEdges_i$ that $i$ declares is the same in $\declaredEdges$ and $\declaredEdges'$. Given that the preferences of $i$ and $j$ are the same in $\declaredEdges$ and $\declaredEdges'$ ($i \notin X$ and $j$ is not strategic), $(i,j)$ blocks $\tilde{M}$ with respect to $\declaredEdges'$. But then $(i,j)$ blocks $M'$ with respect to $\declaredEdges'$ as well. 
But this is a contradiction to the fact that \boost computes a stable matching when run on $\declaredEdges'$.
\end{proof}

\end{document}